\documentclass[12pt,a4paper]{article}
\pdfoutput=1
\usepackage{graphicx,times}
\usepackage{amsmath} 
\usepackage{amssymb}
\usepackage{amsfonts}
\usepackage{amsthm}
\usepackage{harvard}
\usepackage{color}
\usepackage{bm,amsmath,tikz,bbm,amsfonts,nicefrac,latexsym,amsmath,amsfonts,amsbsy,amscd,amsxtra,amsgen,amsopn,bbm,amsthm,amssymb, enumerate, listings,subcaption}
\usepackage{comment}

\usepackage{caption}
\usepackage{subcaption}
\usepackage{booktabs}
\usepackage[tableposition=above]{caption}
 \usepackage{booktabs}
\usepackage{multirow}
\usepackage{float}

\renewcommand{\cite}{\citeyear}
\usepackage{textcomp}

\theoremstyle{plain}

\newtheorem{lemma}{Lemma}[section]

\newtheorem{proposition}{Proposition}[section]

\numberwithin{equation}{section}
\begin{document}
\renewcommand{\thefootnote}{\fnsymbol{footnote}}
\begin{center}
{\Large\textbf{A Consistent Stochastic Model of the Term Structure\\[1ex] of Interest Rates for Multiple Tenors}}\footnote{We would like to thank Alex Backwell, Bruno Bouchard, Alan Brace, Jos\'e da Fonseca, Marc Henrard, Andrea Macrina, Michael Nealon and David Skovmand for helpful discussions on earlier versions of this paper. The usual disclaimer applies.}\\[1em]
{\Large Mesias Alfeus\footnote{University of Technology Sydney, Australia\\ e-mail: {\texttt{Mesias.Alfeus@student.uts.edu.au}}}, Martino Grasselli\footnote{Dipartimento di Matematica, Universit\`a degli Studi di Padova (Italy) and L\'eonard de Vinci P\^ole Universitaire,
Research Center, Finance Group, 92 916 Paris La D\'efense Cedex (France).\\ e-mail:
{\texttt{grassell@math.unipd.it}}} and Erik Schl\"ogl\footnote{University of Technology Sydney, Australia\\ e-mail: {\texttt{Erik.Schlogl@uts.edu.au}}}}\\[2em]
\today
\end{center}
\renewcommand*{\thefootnote}{\arabic{footnote}}
\setcounter{footnote}{0}
\begin{abstract}
Explicitly taking into account the risk incurred when borrowing at a shorter tenor versus lending at a longer tenor (``roll-over risk''), we construct a stochastic model framework for the term structure of interest rates in which a frequency basis (i.e. a spread applied to one leg of a swap to exchange one floating interest rate for another of a different tenor in the same currency) arises endogenously. This roll-over risk consists of two components, a credit risk component due to the possibility of being downgraded and thus facing a higher credit spread when attempting to roll over short--term borrowing, and a component reflecting the (systemic) possibility of being unable to roll over short--term borrowing at the reference rate (e.g., LIBOR) due to an absence of liquidity in the market. The modelling framework is of ``reduced form'' in the sense that (similar to the credit risk literature) the \emph{source} of credit risk is not modelled (nor is the source of liquidity risk). However, the framework has more structure than the literature seeking to simply model a different term structure of interest rates for each tenor frequency, since relationships between rates for all tenor frequencies are established based on the modelled roll-over risk. We proceed to consider a specific case within this framework, where the dynamics of interest rate and roll-over risk are driven by a multifactor Cox/Ingersoll/Ross--type process, show how such model can be calibrated to market data, and used for relative pricing of interest rate derivatives, including bespoke tenor frequencies not liquidly traded in the market.
\end{abstract}
\textbf{JEL Classfication: C6, C63, G1, G13} \\
\line(1,0){250}\\[1em]
\section{Introduction}
The phenomenon of the frequency basis (i.e. a spread applied to one leg of a swap to exchange one floating interest rate for another of a different tenor in the same currency) contradicts textbook no--arbitrage conditions and has become an important feature of interest rate markets since the beginning of the Global Financial Crisis (GFC) in 2008. As a consequence, stochastic interest rate term structure models for financial risk management and the pricing of derivative financial instruments in practice now reflect the existence of multiple term structures, i.e. possibly as many as there are tenor frequencies. While this pragmatic approach can be made mathematically consistent (see Grasselli and Miglietta (\citeyear*{GM15}) and Grbac and Runggaldier (\citeyear*{GrbRung:2015}) for a recent treatise, as well as the literature cited therein), it does not seek to explain this proliferation of term structures, nor does it allow the extraction of information potentially relevant to risk management from the basis spreads observed in the market.

In the pre-GFC understanding of interest rate swaps (see explained in, e.g., Hull (\citeyear*{Hull2008})), the presence of a basis spread in a floating--for--floating interest rate swap would point to the existence of an arbitrage opportunity, unless this spread is too small to recover transaction costs. As documented by Chang and Schl\"ogl (\citeyear*{ChaSch:2015}), post--GFC the basis spread cannot be explained by transaction costs alone, and therefore there must be a new perception by the market of risks involved in the execution of textbook ``arbitrage'' strategies. Since such textbook strategies to profit from the presence of basis spreads would involve lending at the longer tenor and borrowing at the shorter tenor, the prime candidate for this is ``roll--over risk.'' This is the risk that in the future, once committed to the ``arbitrage'' strategy, one might not be able to refinance (``roll over'') the borrowing at the prevailing market rate (i.e., the reference rate for the shorter tenor of the basis swap). This ``roll--over risk,'' invalidating the ``arbitrage'' strategy, can be seen as a combination of ``downgrade risk'' (i.e., the risk faced by the potential arbitrageur that the credit spread demanded by its creditors will increase relative to the market average) and ``funding liquidity risk'' (i.e., the risk of a situation where funding in the market can only be accessed at an additional premium).

We propose to model this roll-over risk explicitly, which endogenously leads to the presence of basis spreads between interest rate term structures for different tenors. This is in essence the ``reduced--form'' or ``spread--based'' approach to multicurve modelling, similar to the approach taken in the credit risk literature, where the risk of loss due to default gives rise to credit spreads. The model allows us to extract the forward--looking ``market's view'' of roll-over risk from the observed basis spreads, to which the model is calibrated. Preliminary explorations using a simple model of deterministic basis spreads in Chang and Schl\"ogl (\citeyear*{ChaSch:2015}) indicate an improving stability of the calibration, suggesting that the basis swap market has matured since the turmoil of the GFC and pointing toward the practicability of constructing and implementing a full stochastic model.

The bulk of the literature on modelling basis spreads is in a sense even more ``reduced--form'' than what we propose here, in the sense that basis spreads are recognised to exist, and are modelled to be either deterministic or stochastic in a mathematically consistent fashion, but there are no structural links between term structures of interest rates for different tenors (in a sense, the analogue of this approach applied to credit risk would be to model stochastic credit spreads directly, without any link to probabilities of default and losses in the event of default). This strand of the literature can be traced back to Boenkost and Schmidt (\citeyear*{BoeSch:2004}), who used this approach to construct a model for cross currency swap valuation in the presence of a basis spread. This was subsequently adapted by Kijima, Tanaka and Wong (\citeyear*{KijTanWon:2009}) to modelling a single--currency basis spread. Henrard (\citeyear*{Henrard:2010}) took an axiomatic approach to the problem, modelling a deterministic multiplicative spread between term structures associated with different tenors. Initially, these models were not reconciled with the requirement of the absence of arbitrage. Subsequent work, however, such as Fujii, Shimada and Takahashi (\citeyear*{FujShiTak:2009}), gave explicit consideration to this requirement. This pragmatic way of modelling interest rates in the presence of spreads between term structures of interest rates for different tenors has been pursued further in a number of papers, including Mercurio (\citeyear*{Mercurio:2009},\citeyear*{Mercurio2010}) in a LIBOR Market Model setting, Kenyon (\citeyear*{Kenyon:2010}) in a short--rate modelling framework, stochastic additive basis spreads in Mercurio and Xie (\citeyear*{MerXie:2012}), and Henrard (\citeyear*{Henrard:2013}) for stochastic multiplicative basis spreads. Moreni and Pallavicini (\citeyear*{MorPal:2014}) construct a model of two curves, riskfree instantaneous forward rates and forward LIBORs, which is Markovian in a common set of state variables. Macrina and Mahomed (\citeyear*{risks6010018}) construct a pricing kernel framework for multicurve models (be it discount count curves in different currencies, or real vs. nominal interest rates, or for different tenors), but again this approach does not attempt to model the structural links between different term structures.

Early work incorporating some of the potential causes of basis spreads into models of the single--currency ``multicurve'' environment post--GFC includes Morini (\citeyear*{Morini:2009}) and Bianchetti (\citeyear*{Bianchetti:2010}), who focus on counterparty credit risk. The model of Cr\'epey (\citeyear*{Crepey:2015}) links funding cost and counterparty credit risk in a credit valuation adjustment (CVA) framework, but does not explicitly consider spreads between different tenor frequencies arising from roll--over risk.

Recently, there has been an emerging view that ``roll--over risk'' is what prevents pre--crisis textbook arbitrage strategies to exploit the basis spreads between tenor frequencies, and that modelling this risk can provide the link between overnight index swaps (OIS), the XIBOR (e.g. LIBOR, EURIBOR, etc.) style money market, the vanilla swap market, and the basis swap market. An important contribution in this vein is Filipovi\'c and Trolle (\citeyear*{FilTro:2013}), who estimate the dynamics of interbank risk from time series data from these markets. They define ``interbank risk'' as ``the risk of direct or indirect loss resulting from lending in the interbank money market.'' Decomposing the term structure of interbank risk into what they identify as default and non-default (liquidity) components, they study the associated risk premia. Filipovi\'c and Trolle interpret the ``default'' component in terms of the risk of a deterioration of creditworthiness of a LIBOR reference panel bank resulting in it dropping out of the LIBOR panel,\footnote{This is also known as the ``renewal effect,'' see Collin--Dufresne and Solnik (\citeyear*{ColSol:2001}) and Grinblatt (\citeyear*{Grinblatt:2001}).} in which case this bank immediately would no longer be able to roll over debt at the overnight reference rate, while the rate on any LIBOR borrowing would remain fixed until the end of the accrual period (i.e., typically for several months). In their analysis, this differential impact of downgrade risk on rolling debt explains part of the LIBOR/OIS spread; the residual is labelled the ``liquidity'' component.\footnote{Past empirical studies, in particular of the GFC, also indicate that credit risk alone is insufficient to explain the LIBOR/OIS spread; see e.g. Eisenschmidt and Tapking (\citeyear*{EisTap:2009}).} It is important to note that both components manifest themselves in the risk of additional cost when rolling over debt, i.e., ``downgrade risk'' and ``funding liquidity risk'' combining to form a total ``roll--over risk.''\footnote{``Funding liquidity risk'' has also been considered explicitly in a separate strand of the literature. For example, Acharya and Skeie (\citeyear*{AchSke:2011}) model liquidity hoarding by participants in the interbank market. In their model, there is a positive feedback effect between roll--over risk and liquidity hoarding (via term premia on interbank lending rates), which in the extreme case can lead to a freeze of interbank lending. Brunnermeier and Pedersen (\citeyear*{BrunnermeierPedersen2009}) model a similar adverse feedback effect between market liquidity and funding liquidity.}

Based on similar considerations, Cr\'epey and Douady  (\citeyear*{CreDou:2013}) model the spread between LIBOR and OIS as a combination of credit and liquidity risk premia, where in particular they focus on providing some model structure for the latter. They construct a stylised equilibrium model of credit risk and funding liquidity risk to explain the LIBOR/OIS spread, arguing (unlike Filipovi\'c and Trolle) that the overnight rate underlying OIS (e.g., the Fed Funds or EONIA rate) is riskfree (we will return to this point in our model setup below).

An alternative approach at the more fundamental end of the modelling spectrum is the recent work by Gallitschke, M\"uller and Seifried (\citeyear*{GalMulSei:2014}), who propose a model for interbank cash transactions and the relevant credit and liquidity risk factors, which endogenously generates multiple term structures for different tenors. In particular, they explicitly model a mechanism by which XIBOR is determined by submissions of the member banks of a panel, which adds substantial complexity to the model.

Our aim is to construct a consistent stochastic model encompassing OIS, XIBOR, vanilla and basis swaps in a single currency.\footnote{Since the GFC a \emph{cross currency basis} exceeding pre-crisis textbook arbitrage bounds has also emerged, see for example Chang and Schl\"ogl (\citeyear*{ChaSch:2012}). The approach presented here could be extended to multiple currencies, but it is our view that across currencies there may be other factors than various forms of roll--over risk giving rise to a basis spread. For example, Andersen et al. (\citeyear*{OZ:And&Duf&Son:2017}) demonstrate how funding value adjustments (FVAs) can prevent potential arbitrageurs from enforcing covered interest parity (CIP) across currencies. If the CIP arbitrage channel is blocked, then this would allow for CIP violations driven by, say, different supply/demand equilibria in FX spot versus FX forward markets.} Our ``reduced--form'' approach explicitly models both the credit and the funding component of roll--over risk to link multiple yield curves. In that, it is more parsimonious than the ``pragmatic'' way of modelling extant in the literature (reviewed above), where stochastic dynamics for basis spreads are specified directly without recourse to the underlying roll--over risk. In particular, this allows the relative pricing of bespoke tenors in a model calibrated to basis spreads between tenor frequencies for which liquid market data is available. It does not require the introduction of a new stochastic factor (or deterministic spread) for each new tenor frequency. However, the approach is ``reduced--form'' in the sense that it abstracts from structural causes of downgrade risk and funding liquidity risk --- in this sense, our approach is closest in spirit to the ``reduced--form'' models of credit risk, doing for basis spreads what those models have done for credit spreads. The framework which we propose below departs from that of Chang and Schl\"ogl (\citeyear*{ChaSch:2015}) in that, rather than focusing exclusively on basis swaps, we treat OIS, XIBOR, vanilla and basis swaps in a unified framework, as well as calibrating credit risk to credit default swaps, and model both the credit and funding liquidity components of roll--over risk in a ``reduced--form'' manner.

The remainder of the paper is organised as follows. Section 2 expresses the basic instruments in terms of the model variables, i.e., the overnight rate, credit spreads and a spread representing pure funding liquidity risk. Section 3 calibrates a concrete specification of the model in terms of multifactor Cox/Ingersoll/Ross--type dynamics to market data for OIS, interest rate and basis swaps --- this is the version of the model focused solely on the frequency basis. In order to separate roll--over risk into its credit and liquidity component in calibration to market data, we need to include repos (for the ``default--free'' interest rate) and credit default swaps (CDS) --- this is done in Section 4. Section 5 concludes.

\section{The model}\label{modelsection}
\subsection{Model variables}
We model a frictionless market free of arbitrage opportunities in which trading takes place continuously over the time interval $[0,T]$, where $T$ is an arbitrary positive final date.
Uncertainty in the market is modeled through a filtered probability space $(\Omega, {\cal F}, ({\cal F}_t)_{t \in [0,T]}, \mathbb{Q})$, supporting all the price processes we are about to introduce. Here $\mathbb{Q}$ denotes the risk neutral measure, so that we are modeling the market directly under the pricing measure. Also, $\mathbb{E}^{\mathbb Q}_t$ is a shorthand for $\mathbb{E}^{\mathbb Q}[.\vert {\cal F}_t]$.

Denote by $r_c$ the continuously compounded short rate abstraction of the interbank overnight rate (e.g., Fed funds rate or EONIA). This is equal to the riskless (default--free) continuously compounded short rate $r$ plus a credit spread. In the simplest case, one could adopt a ``fractional recovery in default,'' a.k.a. ``recovery of market value,'' model\footnote{See Duffie and Singleton (\citeyear*{OZ:Duf&Sin:99}).} and denote by $q$ the (assumed constant) loss fraction in default. Then
\begin{equation}
r_c(s)=r(s)+\Lambda(s)q
\label{collateral}\end{equation}
where $\Lambda(s)q$ is the average (market aggregated) credit spread across the panel and $\Lambda(s)$ is the corresponding default intensity.\footnote{Here, with ``risk neutral measure'' we mean an equivalent martingale measure associated with discounting by the instantaneous, default--free interest rate $r$ (i.e., the numeraire is the associated continuously compounded savings account). In the presence of risks (such as credit risk) with respect to which the market may be incomplete, this measure is not necessarily unique, but we follow the bulk of the credit derivatives literature in assuming that the model, calibrated to the extent possible to available liquid market instruments, gives us the ``correct'' dynamics under the assumed pricing measure. In other words, model prices are arbitrage--free with respect to the prices observed in the market, but this does not necessarily imply that any departure from these model prices would result in an exploitable arbitrage opportunity.} Note that although a significant part of the ``multicurve'' interest rate modelling literature mentioned in the introduction heuristically advances the argument that (mainly due to its short maturity) the interbank overnight rate is essentially free of default risk, in intensity--based models of default this is incorrect even in the instantaneous limit. In our modelling, we do not require this argument. Instead, it suffices that $r_c$ is the appropriate rate at which to discount payoffs of any fully collateralised derivative transaction, because the standard ISDA Credit Support Annex (CSA) stipulates that posted collateral accrues interest at the interbank overnight rate.\footnote{For a detailed discussion of this latter point, see Piterbarg (\citeyear*{Piterbarg2010}).}

Roll--over risk is modelled via the introduction of a $\pi(s)$, denoting the spread over $r_c(s)$ which an arbitrary but fixed entity must pay when borrowing overnight. $\pi(s)$ has two components,
\begin{equation}
\pi(s)=\phi(s)+\lambda(s)q
\end{equation}
where $\phi(s)$ is pure funding liquidity risk\footnote{In Section 3, $\phi(s)$ will be modelled as a diffusion as a ``first--cut'' concrete specification of our model. Modelling liquidity freezes properly may require permitting $\phi(s)$ to jump --- the framework laid out in the present section would allow for such an extension.} (both idiosyncratic and systemic) and $\lambda(s)q$ is the idiosyncratic credit spread over $r_c$ (initially, e.g. at time 0, $\lambda(0)=0$ by virtue of the fact that at time of calibration to basis spread data, we are considering market aggregated averages). The default intensity of any given (but representative) XIBOR panel member is $\Lambda(s)+\lambda(s)$. Thus $\lambda(s)q$ represents the ``credit'' (a.k.a. ``renewal'') risk component of roll--over risk, i.e. the risk that a particular borrower will be unable to roll over overnight (or instantaneously, in our mathematical abstraction) debt at $r_c$, instead having to pay an additional spread $\lambda(s)q$ because their credit quality is lower than that of the panel contributors determining $r_c$.

\subsection{OIS with roll--over risk, one--period case}

\begin{table}[t]
  \centering

    \begin{tabular}{|c|c|c|}
    \hline
    \textbf{Action} & \multicolumn{2}{c|}{\textbf{Time }} \\
    \hline
          & $t$ & $T$ \\
          \hline
    Borrow overnight, & \multirow{3}[1]{*}{1} & \multirow{3}[1]{*}{$-e^{\int_t^Tr_c(s)ds}$} \\
    rolling from time &       &  \\
    $t$ to $T$. &       &  \\
   \hline
    Enter OIS & 0     &  $e^{\int_t^Tr_c(s)ds}-(1+(T-t)\text{OIS}(t,T))$\\
    \hline
    Lend at LIBOR & -1    & $1+(T-t)L(t,T)$  \\
    \hline
    \hline
    \textbf{Net outcome} & 0     &  $(T-t)(L(t,T)-\text{OIS}(t,T))$ \\
    \hline
    \hline
    \end{tabular}%
    \caption{Strategy to exploit the LIBOR/OIS spread assuming the absence of roll--over risk.}\label{arbstrategy}
  \label{tab:addlabel}%
\end{table}%

Let us consider first the simplest case, i.e. the LIBOR/OIS spread over a single accrual period. In the absence of roll--over risk, one could construct a strategy to take advantage of this spread (similar strategies can be constructed for multiple accrual periods, or to take advantage of the frequency (tenor) basis): Borrow at the overnight rate, rolling over the borrowing daily.\footnote{Note again that for mathematical convenience we are equating ``daily'' with the continuous--time infinitesimal limit --- if this simplification is considered to have material impact, it could be lifted at the cost of some additional mathematical tedium.} Enter into an OIS receiving floating and paying fixed, thus eliminating the exposure to interest rate risk. Lend at LIBOR. Although this not strictly an arbitrage strategy,\footnote{This strategy is not strictly riskless, because of the potential impact of default risk. However, this default risk should be reflected in any interbank borrowing credit spread, be it overnight or for a longer period. It is the roll--over risk which introduces a new distinction between shorter and longer tenor borrowing, as we demonstrate below.} if we assume that one is guaranteed to be able to roll over the borrowing at $r_c$ (i.e., in the absence of roll--over risk), and credit spreads apply symmetrically to the borrower and the lender, then this strategy results in a profit equal to the LIBOR/OIS spread, as summarised in Table \ref{arbstrategy}. We interpret this profit as a compensation for roll--over risk, in a manner which we will now proceed to make specific.

Borrowing overnight from $t$ to $T$ (setting $\delta=T-t$), rolling principal and interest forward until maturity, an arbitrary but fixed entity pays at time $T$:
\begin{equation}\label{overnight}
-e^{\int_t^Tr_c(s)ds}e^{\int_{t}^T\pi(s)ds}
\end{equation}
Assuming symmetric treatment of credit risk when borrowing and lending,\footnote{We need to assume symmetric treatment of roll--over risk, both the ``credit'' and the ``funding liquidity'' component, in order to maintain additivity of basis spreads: Swapping a one--month tenor into a three--month tenor, and then swapping the three--month tenor into a 12--month tenor, is financially equivalent to swapping the one--month tenor into the 12--month tenor, and thus (ignoring transaction costs) the 1m/12m basis spread must equal the the sum of the 1m/3m and 3m/12m basis spreads.} lending to an arbitrary but fixed entity from $t$ to $T$ will incur credit risk with intensity $\Lambda(s)+\lambda(s)$. Discounting with
$$
r(s)+(\Lambda(s)+\lambda(s))q=r_c(s)+\lambda(s)q
$$
the present value of (\ref{overnight}) is
\begin{equation}
-\mathbb{E}_t^\mathbb{Q}\left[e^{\int_{t}^T\phi(s)ds}\right] \label{borrow}
\end{equation}
Enter OIS to receive the overnight rate and pay the fixed rate $\text{OIS}(t,T)$, discounting with $r_c$ (due to collateralisation of OIS), the present value of the payments is
\begin{equation}
\mathbb{E}_t^\mathbb{Q}\left[1-e^{-\int_t^Tr_c(s)ds}-e^{-\int_t^Tr_c(s)ds}\text{OIS}(t,T)\delta\right] \label{OIS}
\end{equation}
Spot LIBOR observed at time $t$ for the accrual period $[t,T]$ is denoted by $L(t,T)$. Lending at LIBOR $L(t,T)$, we receive (at time $T$) the credit risky payment
$$
1+\delta L(t,T)
$$
Discounting as above with $r_c(s)+\lambda(s)q$, the present value of this is
\begin{equation}
\mathbb{E}_t^\mathbb{Q}\left[e^{-\int_{t}^T(r_c(s)+\lambda(s)q)ds}(1+\delta L(t,T))\right] \label{LIBOR}
\end{equation}
We must have (because the initial investment in the strategy is zero) that the sum of the three terms ((\ref{borrow}), (\ref{OIS}) and (\ref{LIBOR})) is zero, i.e.
\begin{multline} \label{noarb}
\mathbb{E}_t^\mathbb{Q}\left[e^{\int_{t}^T\phi(s)ds}\right]=\\
\mathbb{E}_t^\mathbb{Q}\left[1+e^{-\int_{t}^T(r_c(s)+\lambda(s)q)ds}(1+\delta L(t,T))-e^{-\int_t^Tr_c(s)ds}(1+\delta \text{OIS}(t,T))\right]
\end{multline}
Consequently, if renewal risk is zero, i.e. $\lambda(s)\equiv0$, then the LIBOR/OIS spread is solely due to funding liquidity risk:
\begin{equation}\label{oneperiod}
\mathbb{E}_t^\mathbb{Q}\left[e^{\int_{t}^T\phi(s)ds}\right]=
\mathbb{E}_t^\mathbb{Q}\left[1+e^{-\int_t^Tr_c(s)ds}\delta(L(t,T)-\text{OIS}(t,T))\right]
\end{equation}
Define the discount factor implied by the overnight rate as
\begin{equation}\label{DOIScond}
D^{OIS}(t,T)=\mathbb{E}_t^\mathbb{Q}\left[e^{-\int_t^Tr_c(s)ds}\right]
\end{equation}
Since the mark--to--market value of the OIS at inception is zero, we have
\begin{eqnarray}
\text{OIS}(t,T) &=& \frac{1-D^{OIS}(t,T)}{\delta D^{OIS}(t,T)}\\
\Leftrightarrow D^{OIS}(t,T) &=& \frac1{1+\delta\text{OIS}(t,T)}\label{DOIS}
\end{eqnarray}
Thus the dynamics of $r_c$ should be consistent with (\ref{DOIS}) (the term structure of the $D^{OIS}(t,T)$), and the dynamics of $\phi(s)$ and $\lambda(s)$ should be consistent with (\ref{noarb}).

The following remarks are worth noting:
\begin{itemize}
\item (\ref{DOIS}) implies
$$
\mathbb{E}_t^\mathbb{Q}\left[e^{-\int_t^Tr_c(s)ds}(1+\delta \text{OIS}(t,T))\right]=1
$$
Therefore (\ref{noarb}) implies that we cannot have
\begin{equation}\label{LIBORPV}
\mathbb{E}_t^\mathbb{Q}\left[e^{-\int_{t}^T(r_c(s)+\lambda(s)q)ds}(1+\delta L(t,T))\right]=1
\end{equation}
unless $\phi(s)=0$, i.e. unless the LIBOR/OIS spread is solely due to renewal risk.
\item It may seem counterintuitive that (\ref{LIBORPV}) doesn't hold, but this is due to the fact that the discounting in (\ref{LIBORPV}) only takes into account credit risk (including ``renewal risk''), i.e. it does not take into account the premium a borrower of LIBOR (as opposed to rolling overnight borrowing) would pay for avoiding the roll--over risk inherent in $\phi(s)$.
\end{itemize}

\subsection{OIS with roll--over risk, multiperiod case}
OIS may pay more frequently than once at $T$ for an accrual period $[t,T]$ (especially when the period covered by the OIS exceeds one year). In this case the strategy of the previous section needs to be modified as follows.

Borrowing overnight from $t=T_0$ to $T_n$ (normalising $T_j-T_{j-1}=\delta, \forall 0<j\leq n$), rolling principal until maturity and interest forward until each $T_j$, an arbitrary but fixed entity pays at time $T_j$ $(\forall 0<j<n)$:
\begin{equation}
1-e^{\int_{T_{j-1}}^{T_j}r_c(s)ds}e^{\int_{T_{j-1}}^{T_j}\pi(s)ds}
\end{equation}
and at time $T_n$:
\begin{equation}
-e^{\int_{T_{n-1}}^{T_n}r_c(s)ds}e^{\int_{T_{n-1}}^{T_n}\pi(s)ds}
\end{equation}
Discounting with $r_c(s)+q\lambda(s)$, the present value of this is
\begin{multline}
\sum_{j=1}^{n-1}\left(E_t^\mathbb{Q}\left[e^{-\int_t^{T_{j}}(r_c(s)+q\lambda(s))ds}\right]
-\mathbb{E}_t^\mathbb{Q}\left[e^{-\int_t^{T_{j-1}}(r_c(s)+q\lambda(s))ds}e^{\int_{T_{j-1}}^{T_i}\phi(s)ds}\right]\right) \label{nborrow} \\
-\mathbb{E}_t^\mathbb{Q}\left[e^{-\int_t^{T_{n-1}}(r_c(s)+q\lambda(s))ds}e^{\int_{T_{n-1}}^{T_n}\phi(s)ds}\right]
\end{multline}
Enter OIS to receive the overnight rate and pay the fixed rate $\text{OIS}(t,T_n)$ at each $T_j$ $(\forall 0<j\leq n)$, discounting with $r_c$, the present value of the payments is
\begin{equation}
\sum_{j=1}^{n}(D^{OIS}(t,T_{j-1})-(1+\delta\text{OIS}(t,T_n))D^{OIS}(t,T_j)) \label{nOIS}
\end{equation}
If lending to time $T_n$ at $L(t,T_n)$ is possible (i.e., a LIBOR $L(t,T_n)$ is quoted in the market), and supposing that LIBOR is quoted with annual compounding, with $T_n-t=m$, the present value of interest and repayment of principal is
\begin{equation}
\mathbb{E}_t^\mathbb{Q}\left[e^{-\int_t^{T_{n}}(r_c(s)+q\lambda(s))ds}(1+L(t,T_n))^m\right] \label{nLIBOR}
\end{equation}
We must have (because the initial investment in the strategy is zero) that the sum of the three terms ((\ref{nborrow}), (\ref{nOIS}) and (\ref{nLIBOR})) is zero. Since the mark--to--market value of the OIS at inception is zero, we can drop (\ref{nOIS}) and write this directly as
\begin{multline} \label{nnoarb}
\sum_{j=1}^{n}\mathbb{E}_t^\mathbb{Q}\left[e^{-\int_t^{T_{j-1}}(r_c(s)+q\lambda(s))ds}e^{\int_{T_{j-1}}^{T_j}\phi(s)ds}\right]=\\
\mathbb{E}_t^\mathbb{Q}\left[e^{-\int_t^{T_{n}}(r_c(s)+q\lambda(s))ds}(1+L(t,T_n))^m\right]
+\sum_{j=1}^{n-1}\mathbb{E}_t^\mathbb{Q}\left[e^{-\int_t^{T_{j}}(r_c(s)+q\lambda(s))ds}\right]
\end{multline}
Analogously to the single--period case, since the mark--to--market value of the OIS at inception is zero, we have
\begin{equation}\label{DOISmult}
\text{OIS}(t,T_n)=\frac{1-D^{OIS}(t,T_n)}{\delta\sum_{j=1}^{n}D^{OIS}(t,T_j)}
\end{equation}

\subsection{Rates faced by an arbitrary but fixed panel member over a longer accrual period}
Substituting (\ref{DOIS}) into (\ref{noarb}), we obtain the dynamics of the $L(t,T)$ as conditional expectations over the dynamics of $\phi$, $r_c$ and $\lambda$:
\begin{equation}
L(t,T)=\frac1{\delta}\left(\frac{\mathbb{E}_t^\mathbb{Q}\left[e^{\int_{t}^T\phi(s)ds}\right]}{\mathbb{E}_t^\mathbb{Q}\left[e^{-\int_{t}^T(r_c(s)+\lambda(s)q)ds}\right]}
-1\right)\label{libor}
\end{equation}
Rewriting this as a discount factor
\begin{equation}
D^L(t,T)=(1+\delta L(t,T))^{-1}= \frac{\mathbb{E}_t^\mathbb{Q}\left[e^{-\int_{t}^T(r_c(s)+\lambda(s)q)ds}\right]}{\mathbb{E}_t^\mathbb{Q}\left[e^{\int_{t}^T\phi(s)ds}\right]}
\label{Dlibor}
\end{equation}
we see that if we assume independence of the dynamics of $r_c(s)+\lambda(s)q$ from $\phi(s)$, the ``instantaneous spread'' (admittedly a theoretical abstraction) over $r_c$ inside the expectation becomes simply $\pi(s)=\phi(s)+\lambda(s)q$:
\begin{equation}
D^L(t,T)= \mathbb{E}_t^\mathbb{Q}\left[e^{-\int_{t}^T(r_c(s)+\phi(s)+\lambda(s)q)ds}\right]
\label{Dlibora}
\end{equation}

\subsection{Going beyond LIBOR maturities with interest rate swaps (IRS)}
The application of roll--over risk must be symmetric, i.e. applied to roll--over of both borrowing and lending (with opposite sign for borrowing vs. lending), because otherwise contradictions are inescapable.\footnote{Such contradictions would arise in particular because asymmetry in the treatment of roll-over risk would prevent additivity of basis spreads: Combing a position of paying one--month LIBOR versus receiving three--month LIBOR with a position of paying three--month LIBOR versus receiving six--month LIBOR is equivalent to paying one--month LIBOR versus receiving six--month LIBOR; therefore (up to transaction costs), the one--month versus six--month basis spread should equal the sum of the one--month versus three--month and three--month versus six--month spreads.}

Suppose we are swapping LIBOR with tenor structure $T^{(L)}$, from $T^{(L)}_0=T_0=t$ to $T^{(L)}_{n_L}$ into a fixed rate $s^{(1)}(t,T^{(1)}_{n_1})$, paid based on a tenor structure $T^{(1)}$, with $T^{(L)}_{n_L}=T^{(1)}_{n_1}$, using an interest rate swap.\footnote{For example, a vanilla USD swap would exchange three--month LIBOR (paid every three months) against a stream of fixed payments paid every six months.}
For a (fully collateralised) swap transaction, it must hold that
\begin{eqnarray}
&&\sum_{j=1}^{n_L}\mathbb{E}_t^\mathbb{Q}\left[e^{-\int_{t}^{T^{(L)}_j}r_c(s)ds}(T^{(L)}_j-T^{(L)}_{j-1})L(T^{(L)}_{j-1},T^{(L)}_j)\right]\nonumber\\
&=& \sum_{j=1}^{n_1}D^{OIS}(t,T^{(1)}_j)(T^{(1)}_j-T^{(1)}_{j-1})s^{(1)}(t,T^{(1)}_{n_1})\label{swapcond}
\end{eqnarray}
Typically, we have market data on (vanilla) fixed--for--floating swaps of different maturities, for a single tenor frequency. Combining this with market data on basis swaps, we get a matrix of market calibration conditions (\ref{swapcond}), i.e. one condition for each (maturity,tenor) combination. Note that the basis spread is typically added to the shorter tenor of a basis swap, so (\ref{swapcond}) needs to be modified accordingly. For example, if $T^{(L)}$ corresponds to a three--month frequency, $T^{(1)}$ corresponds to a six--month frequency, and $T^{(2)}$ corresponds to a one--month frequency of the same maturity ($T^{(L)}_{n_L}=T^{(2)}_{n_2}$), then combining (\ref{swapcond}) with a basis swap with spread $b(t,T^{(L)}_{n_L},T^{(2)}_{n_2})$, we obtain
\begin{eqnarray}
&&\sum_{j=1}^{n_2}\mathbb{E}_t^\mathbb{Q}\left[e^{-\int_{t}^{T^{(2)}_j}r_c(s)ds}(T^{(2)}_j-T^{(2)}_{j-1})L(T^{(2)}_{j-1},T^{(2)}_j)\right]\nonumber\\
&=& \sum_{j=1}^{n_1}D^{OIS}(t,T^{(1)}_j)(T^{(1)}_j-T^{(1)}_{j-1})s^{(1)}(t,T^{(1)}_{n_1})\nonumber\\
&&\qquad-\ \sum_{j=1}^{n_2}D^{OIS}(t,T^{(2)}_j)(T^{(2)}_j-T^{(2)}_{j-1})b(t,T^{(L)}_{n_L},T^{(2)}_{n_2})\label{basisswapcond1}
\end{eqnarray}
Conversely, because the convention of basis swaps is that the basis swap spread is added to the shorter tenor, if $T^{(L)}$ corresponds to a three--month frequency, $T^{(1)}$ corresponds to a six--month frequency, and $T^{(3)}$ corresponds to a twelve--month frequency of the same maturity ($T^{(L)}_{n_L}=T^{(3)}_{n_3}$), then combining (\ref{swapcond}) with a basis swap with spread $b(t,T^{(L)}_{n_L},T^{(3)}_{n_3})$, we obtain
\begin{eqnarray}
&&\sum_{j=1}^{n_3}\mathbb{E}_t^\mathbb{Q}\left[e^{-\int_{t}^{T^{(3)}_j}r_c(s)ds}(T^{(3)}_j-T^{(3)}_{j-1})L(T^{(3)}_{j-1},T^{(3)}_j)\right]\nonumber\\
&=& \sum_{j=1}^{n_1}D^{OIS}(t,T^{(1)}_j)(T^{(1)}_j-T^{(1)}_{j-1})s^{(1)}(t,T^{(1)}_{n_1})\nonumber\\
&&\qquad+\ \sum_{j=1}^{n_3}D^{OIS}(t,T^{(3)}_j)(T^{(3)}_j-T^{(3)}_{j-1})b(t,T^{(L)}_{n_L},T^{(3)}_{n_3})\label{basisswapcond2}
\end{eqnarray}


\subsection{Stochastic modelling}
We assume that the model is driven by a  time-homogeneous affine Markov process  $X$ taking values in a non-empty convex subset $E$ of $\mathbb{R}^d$ ($d\geq 1$), endowed with the inner product $\langle \cdot, \cdot \rangle$.
We assume that $X = (X_t)_{t \in [0,\infty)}$ admits the transition semigroup $(P_t)_{t \in [0,\infty)}$ acting on ${\cal B}(E)_b$ (the space of bounded Borel functions on $E$).
%

Let us now define the model variables $r_c(t)$, $\lambda(t)$ and $\phi(t)$ as follows:
\begin{eqnarray}
r_c(t)    &=& a_0(t)+\langle a, X(t)\rangle\label{r_c}\\
\lambda(t) &=& b_0(t)+\langle b, X(t)\rangle\label{lambda}\\
\phi(t)    &=& c_0(t)+\langle c,X(t)\rangle\label{phi}
\end{eqnarray}
where $a,b,c$  are arbitrary constant projection vectors in $E$ and $a_0,b_0,c_0$ are scalar deterministic functions.
Note that the process $(r_c,\lambda, \phi )$ does not enjoy, in general, the Markov property and when it does, it is \emph{a priori} time-inhomogeneous.

In order to preserve analytical tractability, in this paper we will take the Markov process $X$ in the class of affine processes, introduced by Duffie and Kan (\citeyear*{DuffieKan1996}) and then classified by Duffie et al. (\citeyear*{DuffieEtAl2003}) in the canonical state space domain $E=\mathbb{R}^m_+ \times \mathbb{R}^n$. Affine processes have been recently recovered thanks to the interesting extension to the state space  of positive semidefinite matrices (see Bru (\citeyear*{Bru1991}), Gourieroux and Sufana (\citeyear*{GourierouxSufana2003}),
Da Fonseca et al. (\citeyear*{DaFonsecaEtAl2007}),
 Grasselli and Tebaldi (\citeyear*{GrasselliTebaldi2008})  and Cuchiero et al. (\citeyear*{CuchieroEtAl2011})).\footnote{Affine processes belong to the  more general family of polynomial processes recently investigated by e.g.  Filipovic and Larsson (\citeyear*{polynomial}), see also references therein. Our approach can be easily extended to this larger class with minor changes and a slightly different technique in the computation of the expectations. The same holds true as well as for other processes, like for example the L\'evy driven models, see e.g. Eberlein and Raible (\citeyear*{EberleinRaible(1999)}).  Actually, the choice of the stochastic model is just instrumental as far as the computation of the relevant expectations involved in the sequel can be efficiently performed, in view of our calibration exercise.}

In Appendix \ref{affine} we recall the definition of the affine processes as well as their characterisation in terms of the solution of Riccati ODEs. Then, in Appendix \ref{expectations}, we develop the computation of the relevant expectations involved in the pricing of swaps as well as non-linear instruments like caps. We emphasise that our framework is extremely analytically tractable.

\section{Calibration to market data}\label{calsec}
\begin{table}[t]
\begin{center}
\begin{tabular}{|l|l|l|l|l|l|l|l|l|l|}
\hline
\multicolumn{9}{|c|}{\;\;\;\;\;\;\;\;\;\;\;\;\;\;\;IRS \;\;\;\;\;\;\;\; \;\;\;\;\;\;\;\;\;\;\;\;OIS \;\;\;\;\;\;\;\;\;\;\;\;1m/3m\;\;\;\;\;\;\;\;\;\;\;\;\;3m/6m} \\\hline
T& bid& ask& bid& ask& bid& ask& bid& ask\\ \hline
0.5   & 0.50825 & 0.50825 & 0.13  & 0.17  & 9.6   & 9.6   & 19.32 & 21.32 \\
    1     & 0.311 & 0.331 & 0.125 & 0.165 & 8.29  & 9     & 16.25 & 18.25 \\
    2     & 0.37  & 0.395 & 0.125 & 0.165 & 7.8   & 9.8   & 13.56 & 15.56 \\
    3     & 0.5   & 0.5   & 0.12  & 0.16  & 7.21  & 9.21  & 11.7  & 13.7 \\
    4     & 0.657 & 0.667 & 0.12  & 0.16  & 7.7   & 7.7   & 10.64 & 12.64 \\
    5     & 0.83  & 0.87  & 0.109 & 0.16  & 7.3   & 7.3   & 9.74  & 11.74 \\
    6     & 1.0862 & 1.0978 & 0.13  & 0.17  & 6.8   & 6.8   & 10.2  & 10.2 \\
    8     & 1.5001 & 1.5149 & 0.226 & 0.276 & 6     & 6     & 9.7   & 9.7 \\
    9     & 1.6496 & 1.6684 & 0.368 & 0.418 & 5.7   & 5.7   & 9.7   & 9.7 \\
    10    & 1.836 & 1.837 & 0.563 & 0.613 & 5.3   & 5.3   & 8.67  & 10.67 \\
\hline
\end{tabular}\\[1ex]
\caption{Market data quotes on 01/01/2013. Source: Bloomberg }\label{Marketdata:20130101}
\end{center}
\end{table}

\begin{table}[t]
\begin{center}
\begin{tabular}{|l|l|l|l|l|l|l|l|l|l|}
\hline
\multicolumn{9}{|c|}{\;\;\;\;\;\;\;\;\;\;\;\;\;\;\;IRS \;\;\;\;\;\;\;\; \;\;\;\;\;\;\;\;\;\;\;\;OIS \;\;\;\;\;\;\;\;\;\;\;\;1m/3m\;\;\;\;\;\;\;\;\;\;\;\;\;3m/6m} \\\hline
T& bid& ask& bid& ask& bid& ask& bid& ask\\ \hline
0.5 & 0.3263 & 0.3263&0.091 & 0.096 & 8.2 & 8.2 & 1.96875 & 2.46875 \\
    1     & 0.337 & 0.3395 & 0.092 & 0.097 & 8.75  & 9.25  & 2.09375& 2.21875 \\
    2     & 0.7191 & 0.7231 & 0.094 & 0.099 & 9.625 & 10.125 & 2.125 & 2.25 \\
    3     & 1.157 & 1.161 & 0.098 & 0.103 & 10.25 & 10.75 & 2.125 & 2.25 \\
    4     & 1.5234 & 1.5274 & 0.126 & 0.131 & 10.75 & 11.25 & 2.15625 & 2.28125 \\
    5     & 1.8   & 1.8038 & 0.184 & 0.189 & 11    & 11.5  & 2.15625 & 2.28125 \\
    6     & 2.0166 & 2.0206 & 0.345 & 0.355 & 10.875 & 11.375 & 2.15625 & 2.28125 \\
    8     & 2.3348 & 2.3388 & 0.949 & 0.989 & 10.125 & 10.625 & 2.15625 & 2.28125 \\
    9     & 2.4551 & 2.4591 & 1.296 & 1.346 & 9.625 & 10.125 & 2.15625 & 2.28125 \\
    10    & 2.559 & 2.563 & 1.562 & 1.602 & 9.125 & 9.625 & 2.15625 & 2.28125 \\

\hline
\end{tabular}\\[1ex]
\caption{Market data quotes on 08/09/2014. Source: Bloomberg }\label{Marketdata:20140908}
\end{center}
\end{table}

In this section, we proceed to calibrate the version of the model based on multifactor Cox/Ingersoll/Ross (CIR) dynamics to market data for overnight index swaps, vanilla and basis swaps, step by step adding more instruments to the calibration. Thus, the objective is to obtain a single calibrated, consistent roll--over risk model for interest rate term structures of all tenors. More explicit separation of roll--over risk into its credit and funding liquidity components is left for Section \ref{CDScal} below, which includes credit default swaps in the set of calibration instruments. This version of the model is obtained by making the dynamics (\ref{r_c})--(\ref{phi}) specific as
\begin{eqnarray}
r_c(t)    &=& a_0(t)+\sum_{i=1}^{d} a_iy_{i}(t)\label{r_cspec}\\
\lambda(t) &=& b_0(t)+\sum_{i=1}^{d} b_iy_{i}(t)\label{lambdaspec}\\
\phi(t)    &=& c_0(t)+\sum_{i=1}^{d} c_iy_{i}(t)\label{phispec}
\end{eqnarray}
where the $y_{i}$ follow the Cox--Ingersoll--Ross (CIR) dynamics under the pricing measure, i.e. \\
\begin{equation}
dy_{i}(t) = \kappa_{i}(\theta_{i} -y_{i}(t)) dt + \sigma_{i}\sqrt{y_{i}(t)} dW_{i}(t),\label{dy_i}
\end{equation}\\
where $dW_{i}(t) \quad (i = 1, \cdots, d)$ are independent Wiener processes. In order to keep the model analytically tractable, in keeping with (\ref{r_c})--(\ref{phi}) we do not allow for time--dependent coefficients at this stage.\footnote{For example, following Schl\"ogl and Schl\"ogl \cite{OZ:Sch&Sch:00}, the coefficients could be made piecewise constant to facilitate calibration to at--the--money option price data, while retaining most analytical tractability.} Since each of the factors follow independent CIR--type dynamics, the sufficient condition for each factor to remain positive is $2\kappa_{i}\theta_{i}\geq \sigma_{i}^{2}, \forall i$, as discussed for the one--factor case in Cox et al. (\citeyear*{CIR1985}).

\begin{table}[t]
\begin{center}
\begin{tabular}{|l|l|l|l|l|l|l|l|l|l|}

\hline
\multicolumn{9}{|c|}{\;\;\;\;\;\;\;\;\;\;\;\;\;\;\;IRS \;\;\;\;\;\;\;\; \;\;\;\;\;\;\;\;\;\;\;\;OIS \;\;\;\;\;\;\;\;\;\;\;\;1m/3m\;\;\;\;\;\;\;\;\;\;\;\;\;3m/6m} \\\hline
T& bid& ask& bid& ask& bid& ask& bid& ask\\ \hline
    0.5   & 0.4434 & 0.4434 & 0.136 & 0.146 & 9.1   & 9.1   & 2.375 & 2.875 \\
    1     & 0.51  & 0.511 & 0.159 & 0.164 & 10    & 10.5  & 2.19675 & 2.19675 \\
    2     & 0.896 & 0.898 & 0.175 & 0.18  & 11.75 & 12.25 & 1.9375 & 2.05625 \\
    3     & 1.2354 & 1.2404 & 0.189 & 0.194 & 13    & 13.5  & 1.84375 & 1.96875 \\
    4     & 1.5103 & 1.5153 & 0.255 & 0.26  & 13.625 & 14.125 & 1.78125 & 1.90625 \\
    5     & 1.698 & 1.757 & 0.333 & 0.338 & 13.875 & 14.375 & 1.78125 & 1.90625 \\
    6     & 1.918 & 1.922 & 0.541 & 0.581 & 13.75 & 14.25 & 1.75  & 1.875 \\
    8     & 2.1877 & 2.1917 & 0.991 & 1.041 & 13.035 & 14.035 & 1.8125 & 1.9375 \\
    9     & 2.2857 & 2.288 & 1.255 & 1.305 & 13    & 13.5  & 1.84375 & 1.96875 \\
    10    & 2.3665 & 2.369 & 1.478 & 1.518 & 12.75 & 13.25 & 1.90625 & 2.03125 \\
\hline
\end{tabular}\\[1ex]
\caption{Market data quotes on 18/06/2015. Source: Bloomberg }\label{Marketdata:20150618}
\end{center}
\end{table}

\begin{table}[t]
\begin{center}
\begin{tabular}{|l|l|l|l|l|l|l|l|l|l|}

\hline
\multicolumn{9}{|c|}{IRS \;\;\;\; \;\;\;\;\;\;\;\;\;\;\;\;OIS \;\;\;\;\;\;\;\;\;\;\;\;\;\;\;\;\;1m/3m\;\;\;\;\;\;\;\;\;\;\;\;\;\;\;\;\;3m/6m} \\\hline
T& bid& ask& bid& ask& bid& ask& bid& ask\\ \hline

    0.5   & 0.90415 & 0.90415 & 0.379 & 0.385 & 17.3  & 17.875 & 4.96875 & 5.96875 \\
    1     & 0.7705 & 0.774 & 0.395 & 0.402 & 15.75 & 17.232 & 5.35  & 5.625 \\
    2     & 0.8935 & 0.8975 & 0.4055 & 0.4125 & 15    & 14.4  & 4.725 & 5.1 \\
    3     & 1.0005 & 1.0054 & 0.419 & 0.42  & 14.625 & 13.875 & 4.4   & 4.65 \\
    4     & 1.1075 & 1.1095 & 0.452 & 0.459 & 13.111 & 13.623 & 4.075 & 4.325 \\
    5     & 1.208 & 1.21  & 0.485 & 0.492 & 12.245 & 13.237 & 3.825 & 4.075 \\
    6     & 1.3097 & 1.312 & 0.54  & 0.548 & 12.2  & 12.755 & 3.825 & 3.825 \\
    8     & 1.4895 & 1.4895 & 0.685 & 0.695 & 10.063 & 11.337 & 3.7   & 3.7 \\
    9     & 1.5655 & 1.5675 & 0.7838 & 0.7938 & 9.665 & 10.904 & 3.575 & 3.825 \\
    10    & 1.634 & 1.636 & 0.863 & 0.903 & 9.359 & 10.625 & 3.625 & 3.85 \\

\hline
\end{tabular}\\[1ex]
\caption{Market data quotes on 20/04/2016. Source: Bloomberg }\label{Marketdata:20160420}
\end{center}
\end{table}

\begin{table}[t]
\begin{center}
\begin{tabular}{|l|l|l|l|l|l|l|l|l|l|}

\hline
\multicolumn{9}{|c|}{IRS \;\;\;\; \;\;\;\;\;\;\;\;\;\;\;\;OIS \;\;\;\;\;\;\;\;\;\;\;\;\;\;\;\;\;1m/3m\;\;\;\;\;\;\;\;\;\;\;\;\;\;\;\;\;3m/6m} \\\hline
T& bid& ask& bid& ask& bid& ask& bid& ask\\ \hline

   0.5 & 1.43128 & 1.43128 & 0.9809 & 0.989 & 13.7091 & 16.4909 &19.5232 & 20.2268 \\
    1     & 1.385 & 1.39  & 1.1165 & 1.1165 & 14.016 & 15.184 & 17.4461 & 18.6439 \\
    2     & 1.6098 & 1.6188 & 1.3335 & 1.3335 & 13.6717 & 15.0783 & 15.4479 & 16.752 \\
    3     & 1.795 & 1.7962 & 1.494 & 1.504 & 11.7628 & 15.4372 & 13.9918 & 15.4082 \\
    4     & 1.9304 & 1.9354 & 1.601 & 1.649 & 11.4041 & 14.0959 & 13.785 & 14.365 \\
    5     & 2.038 & 2.0401 & 1.718 & 1.728 & 11.6383 & 12.3617 & 13.6495 & 14.1505 \\
    6     & 2.1249 & 2.1299 & 1.796 & 1.804 & 10.7223 & 11.7777 & 13.4117 & 14.3383 \\
    8     & 2.2614 & 2.2664 & 1.914 & 1.924 & 10.1162 & 10.7838 & 13.4088 & 14.2912 \\
    9     & 2.3159 & 2.3209 & 0.9809 & 0.989 & 9.7987 & 10.5013 & 14.0925 & 14.3075 \\
    10    & 2.3671 & 2.368 & 1.987 & 2.037 & 9.5421 & 10.258 & 13.6338 & 14.6062 \\
\hline
\end{tabular}\\[1ex]
\caption{Market data quotes on 22/03/2017. Source: Bloomberg }\label{Marketdata:20170322}
\end{center}
\end{table}

\begin{table}[t]
\begin{center}
\begin{tabular}{|l|l|l|l|l|l|l|l|l|l|}

\hline
\multicolumn{9}{|c|}{IRS \;\;\;\; \;\;\;\;\;\;\;\;\;\;\;\;OIS \;\;\;\;\;\;\;\;\;\;\;\;\;\;\;\;\;1m/3m\;\;\;\;\;\;\;\;\;\;\;\;\;\;\;\;\;3m/6m} \\\hline
T& bid& ask& bid& ask& bid& ask& bid& ask\\ \hline

   0.5 & 1.57511 & 1.57511 & 1.3455 & 1.3555 & 2.5057 & 5.9943 & 9.5544 & 10.1956 \\
    1     & 1.6286 & 1.6316 & 1.4655 & 1.4715 & 6.2566 & 6.7434 & 10.127 & 11.0731 \\
    2     & 1.8071 & 1.8121 & 1.59  & 1.63  & 7.127 & 8.873 & 10.8631 & 11.3869 \\
    3     & 1.9235 & 1.9284 & 1.699 & 1.709 & 8.1502 & 8.8498 & 11.2344 & 12.0156 \\
    4     & 2.0095 & 2.0115 & 1.749 & 1.799 & 8.1319 & 9.1181 & 11.4939 & 12.3061 \\
    5     & 2.0815 & 2.0832 & 1.827 & 1.837 & 8.2078 & 9.0422 & 11.8685 & 12.8315 \\
    6     & 2.1469 & 2.1488 & 1.862 & 1.912 & 7.8617 & 8.8383 & 12.1997 & 13.1753 \\
    8     & 2.261 & 2.264 & 1.961 & 2.011 & 7.2408 & 8.0593 & 13.1587 & 13.8513 \\
    9     & 2.3086 & 2.3103 & 1.3455 & 2.042 & 7.3706 & 7.8294 & 13.5685 & 14.3315 \\
    10    & 2.3519 & 2.3533 & 2.035 & 2.085 & 6.8616 & 7.8384 & 13.4028 & 14.7973 \\
\hline
\end{tabular}\\[1ex]
\caption{Market data quotes on 31/10/2017. Source: Bloomberg }\label{Marketdata:20171031}
\end{center}
\end{table}

We begin our calibration with OIS, then include money market instruments such as vanilla swaps and basis swaps. Calibration to instruments with optionality (caps/floors and swaptions) is in principle possible, but not included in the scope of the present paper. At each step calibration consists in minimising the sum of squared deviations from calibration conditions such as (\ref{DOIScond}) and (\ref{swapcond}).

Since we are modelling several sources of risk simultaneously, the calibration problem necessarily becomes one of high-dimensional non-linear optimisation. To this end, we apply the ``Adaptive Simulated Annealing'' algorithm of Ingber \cite{Ingber:93} and Differential Evolution of Storn and Price \cite{Rainer:1997}. The data set (of daily bid and ask values) for the instruments under consideration is sourced from Bloomberg, consisting of US dollar (USD) instruments commencing on 1 January 2013 and ending on 31 October 2017. Since we are calibrating the model to ``cross--sectional'' data, i.e. ``snapshots'' of the market on a single day, we present the results for an exemplary set of dates only, having verified that results for other days in our data period are qualitatively similar.

\subsection{Market data processing}
Our goal is to jointly calibrate the model to three different tenor frequencies, i.e., to  1-month, 3-month and 6-month tenors, as well as to the OIS--based discount factors. Table \ref{Marketdata:20130101} to \ref{Marketdata:20171031} show market data quoted on the 01/01/2013, 08/09/2014,  18/06/2015, 20/04/2016, 22/03/2017 and 31/10/2017, respectively.\footnote{With maturities $T=\{0.5,1,2,3,4,5,6,8,9,10\}$, as there are no quotes for  maturity of 7 years in the basis swap market. IRS and OIS are given in $\%$ while basis swaps are quoted in basis points, i.e., one--hundredths of one percent.}

In the USD market, the benchmark interest rate swap (IRS) pays 3-month LIBOR floating vs. 6-month fixed, i.e. $\delta=0.25$ and $\Delta=0.5$. To get IRS indexed to LIBOR of other maturities, we can use combinations of basis swaps with the benchmark IRS.
Calibration conditions for 1-month, 3-month and 6-month tenors are given in Equations \eqref{basisswapcond1}, \eqref{swapcond}, and in \eqref{basisswapcond2}, respectively. We compute the right--hand side of these equations using the market data. We shall call this ``the market side.''

We use the OIS equation \eqref{OIS} for interpolation. Our approach is to calibrate the model to OIS--based discount factors, then use the calibrated model to infer discount factors of any required intermediate maturities. In particular, we bootstrap the time-dependent parameter $a_0$ for both OIS bid and ask and then interpolate the bond prices using the OIS equation to get the corresponding bid and ask values for intermediate maturities.

\subsection{Model calibration problem}
Model calibration is a process of finding model parameters such that model
prices match market prices as closely as possible, i.e., in equation \eqref{swapcond}, we want the left hand side to be equal to the right hand side. In essence, we are implying model parameters from the market data. The most common approach is to minimise the deviation between model prices and market prices, i.e., in the ``relative least--squares'' sense of
\begin{eqnarray}\label{cali}
\arg \min_{\theta\in\Theta} \sum_{i=1}^N \left(\frac{P_{\tiny \mbox{model}}(\theta)-P_{\tiny \mbox{market}}}{P_{\tiny \mbox{market}}}\right)^2,
\end{eqnarray}
where $N$ is the number of calibration instruments. $\theta$ represents a choice of parameters from $\Theta$, the admissible set of model parameter values. $P_{\tiny \mbox{model}}$ and $P_{\tiny \mbox{market}}$ are the model and market prices, respectively. In our case, $P_{\tiny \mbox{model}}$ is the left hand side of \eqref{swapcond} and $P_{\tiny \mbox{market}}$ is the right hand side of the same equation. Market instruments are usually quoted as bid and ask prices, therefore we consider market prices to be accurate only to within the bid/ask spread. Thus, we solely require the model prices to fall within bid--ask domain rather than attempting to fit the mid price exactly. The problem in \eqref{cali} can be reformulated as
\begin{eqnarray}\label{bidask}
\arg \min_{\theta\in\Theta} \sum_{i=1}^N \left(\max\left\{\frac{P_{\tiny \mbox{model}}(\theta)-P_{\tiny \mbox{market}}^{\tiny \mbox{ask}}}{P_{\tiny \mbox{market}}^{\tiny \mbox{ask}}},0\right\}+\max\left\{\frac{P_{\tiny \mbox{market}}^{\tiny \mbox{bid}}-P_{\tiny \mbox{model}}(\theta)}{P_{\tiny \mbox{market}}^{\tiny \mbox{bid}}},0\right\}\right)^2.
\end{eqnarray}

Gradient-based techniques typically are inadequate to handle high--dimensional problems of this type.  Instead, we utilise two methods for global optimisation, Differential Evolution (DE) (see Storn and Price \cite{Rainer:1995,Rainer:1997}) and Adaptive Simulated Annealing (ASA) (see Ingber \cite{Ingber:1996,Ingber:93}).

\subsection{Model parameters}
Recall that the model is made up of the following quantities: $r_c$ which is the short rate abstraction of the interbank overnight rate, $\phi$ the pure funding liquidity risk, and $\lambda$ which represents the credit or renewal risk portion of roll-over risk.

\begin{figure}[t]
\begin{center}
\includegraphics[width=0.9\textwidth]{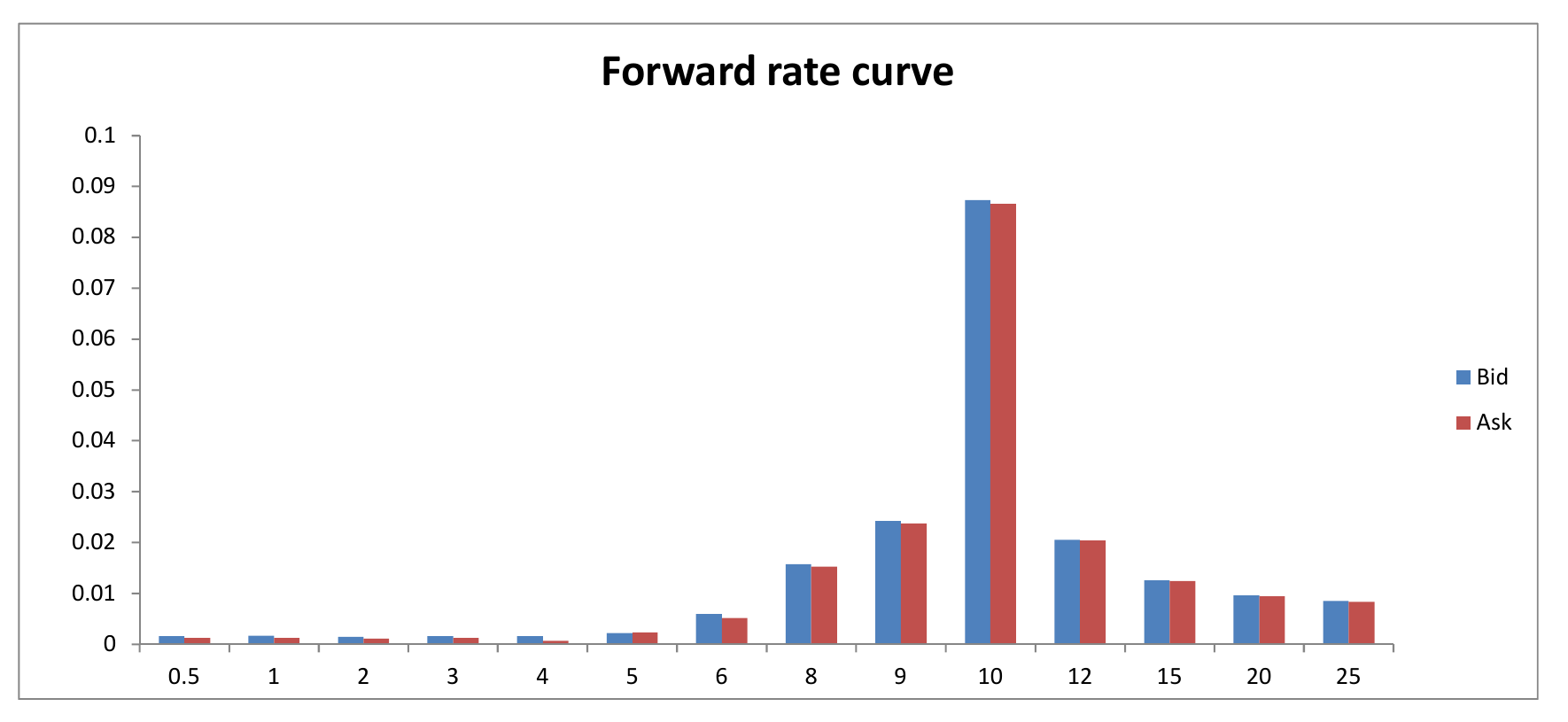}
\end{center}
\caption{OIS--implied forward rates on 1 January 2013 (based on data from Bloomberg)}\label{OISfwd}
\end{figure}
In order to keep the dimension of the optimisation problem manageable, we split the calibration into 3 stages:
\begin{itemize}
\item Stage 1: Fit the model to OIS data. The OIS model parameters consist of the set
\[\theta=\{a_i,y_i(0),\sigma_i,\kappa_i,\theta_i,\;\forall i=1,\cdots,d, \;\mbox{and}\;\; a_0(t),\; t\in [T_{k-1},T_k],\; k=1,2,\cdots,10\}.\]
where $d$ is the number of factors (one or three in the cases considered here).
\item Stage 2: Using Stage 1 calibrated parameters, calibrate the model to basis swaps. From the results in Section \ref{CIRLIBOR}, note that basis swaps depend on $c_0(t)$ and $b_0(t)$ only via $c_0(t)+qb_0(t)$, i.e. the two cannot be separated based on basis swap data alone --- this will be done when calibration to credit default swaps is implemented in Section \ref{CDScal}. Define $d_0(t)=c_0(t)+qb_0(t)$ and in this stage assume constant $d_0$. Thus, the parameters which need to be determined in this step are
\[\theta=\{b_i,c_i,q, \;\forall i=1,\cdots,d, \;\mbox{and}\;\; d_0\}.\]
\item Stage 3: Now, keeping all other parameters fixed, we improve the fit to vanilla and basis swaps by allowing $d_0(t)$ to be a piecewise constant function of time, i.e. $d_0(t)=d_0^{(k)}$ for $t\in[T_{k-1},T_k).$ Since the shortest tenor is one month, the $[T_{k-1},T_k)$ are chosen to be one--month time intervals. The calibration problem is then underdetermined, so we aim for the $d_0^{(k)}$ to vary as little as possible by imposing an additional penalty in the calibration objective function based on the sum of squared deviations between consecutive values, i.e.
    $$
    \sum_k(d_0^{(k+1)}-d_0^{(k)})^2.
    $$
    This type of smoothness criterion on time--varying parameters is well established in the literature on the calibration of interest rate term structure models, with its use in this context dating back to the seminal paper by Pedersen (\citeyear*{OZ:Pedersen:98}).
\end{itemize}

\subsection{Fitting to OIS discount factors}\label{OISfit}
\begin{figure}[t]
  \centering
    \begin{subfigure}[b]{0.53\textwidth}
        \includegraphics[width=\textwidth]{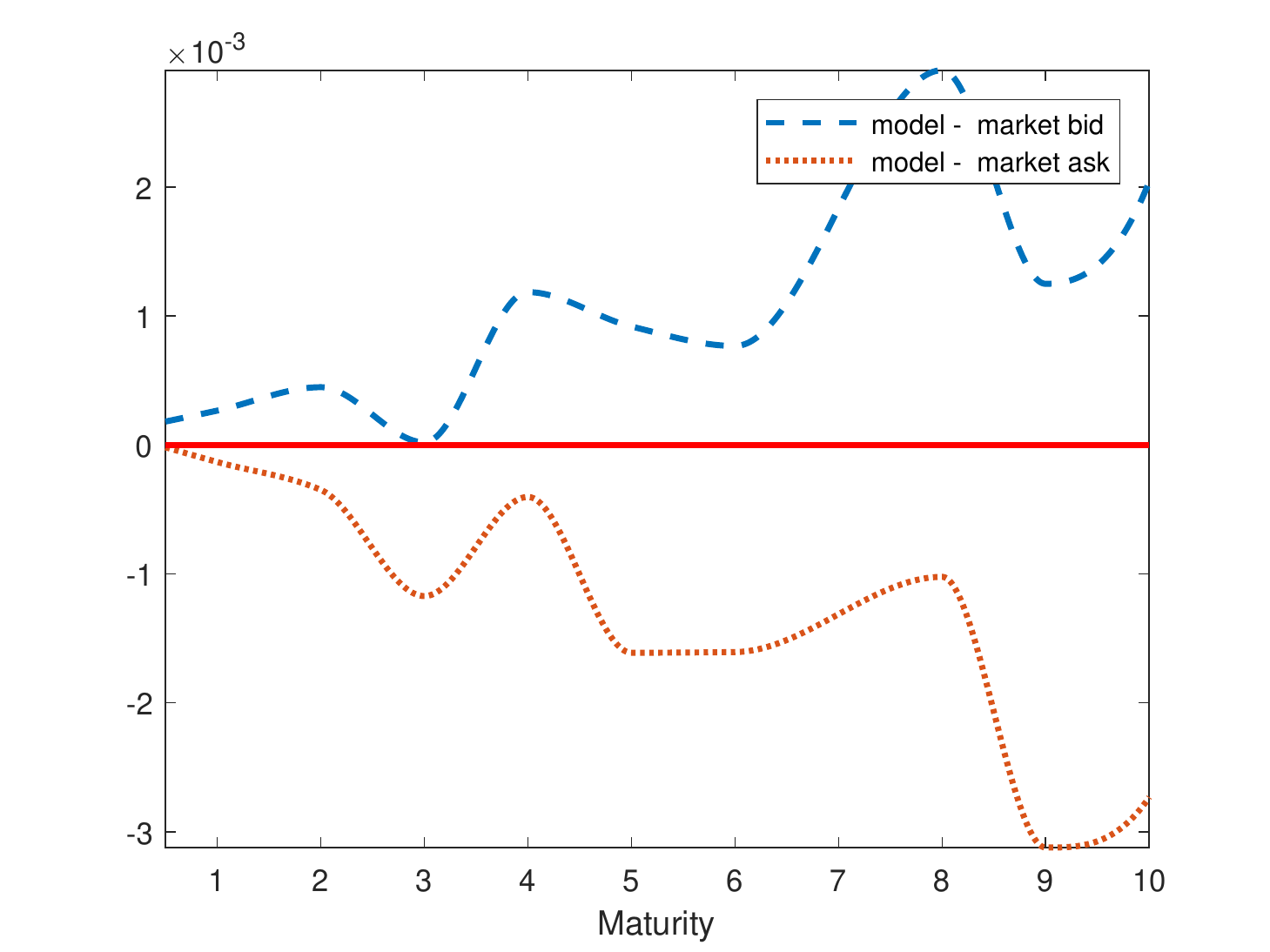}
        \caption{1 January 2013}
        \label{fig:u}
    \end{subfigure}
    ~
    \begin{subfigure}[b]{0.53\textwidth}
        \includegraphics[width=\textwidth]{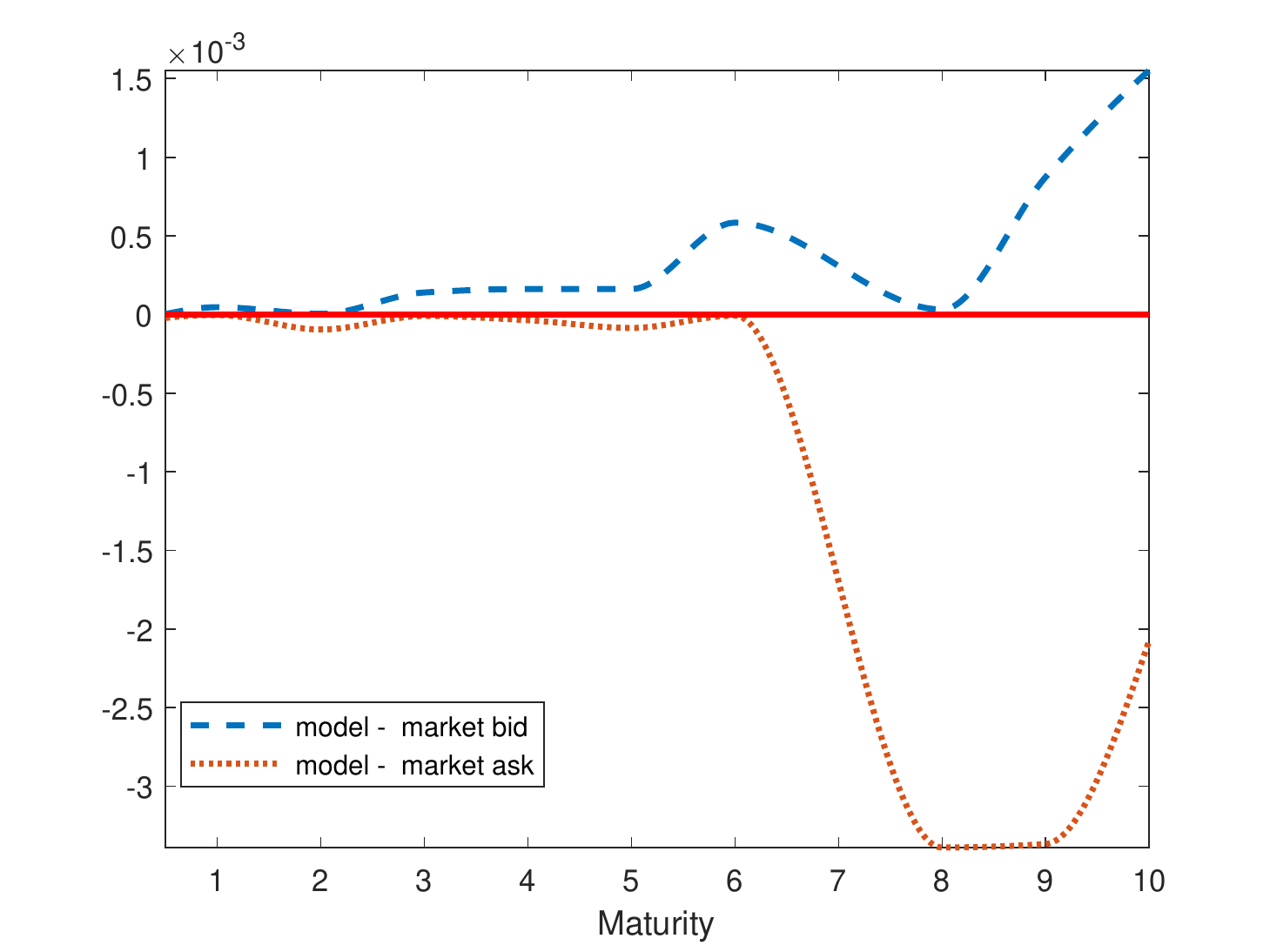}
       \caption{08 September 2014}
        \label{figx}
    \end{subfigure}

\begin{subfigure}[b]{0.53\textwidth}
        \includegraphics[width=\textwidth]{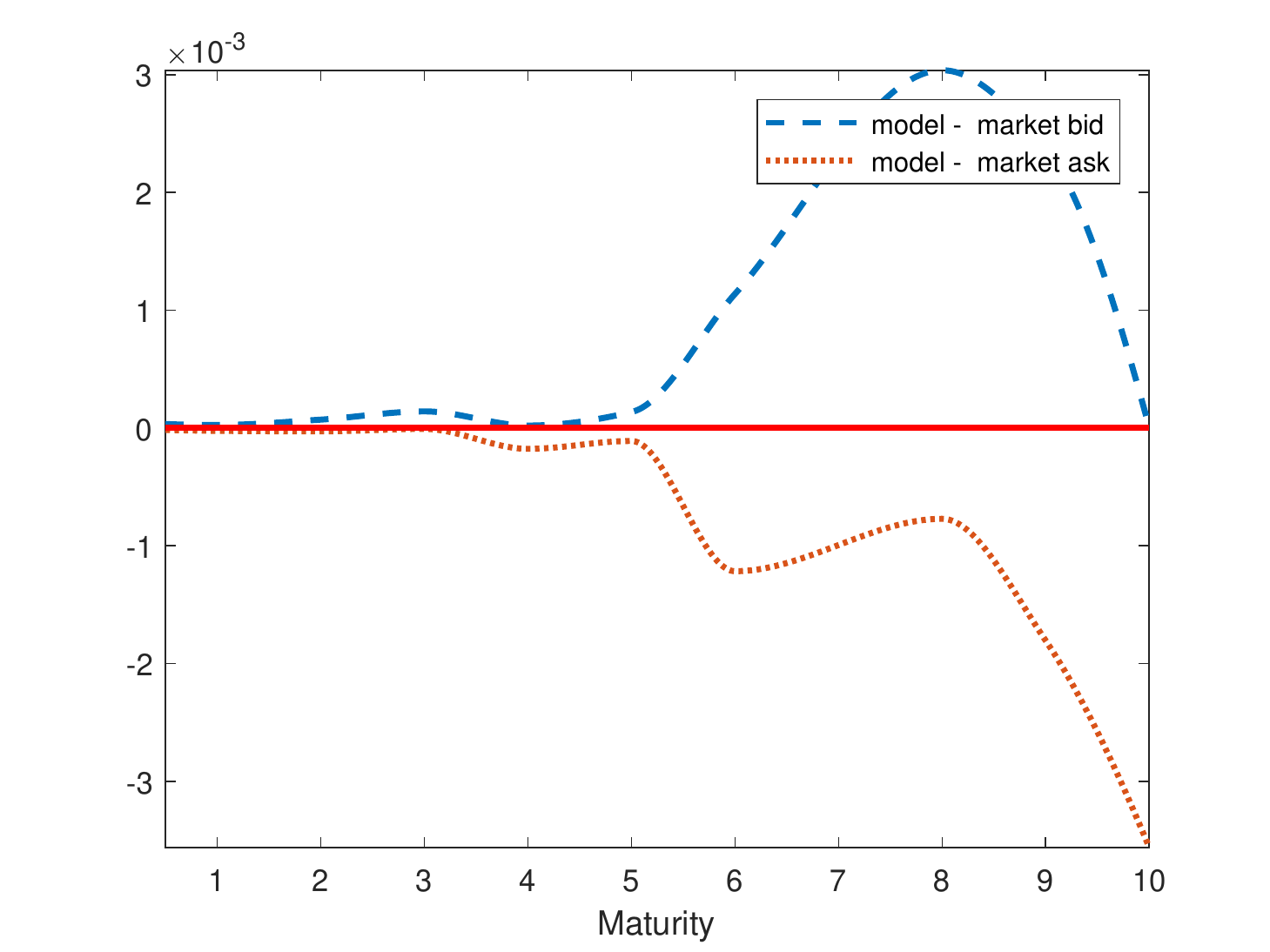}
       \caption{18 June 2015}
        \label{figx}
    \end{subfigure}
    ~
     \begin{subfigure}[b]{0.53\textwidth}
        \includegraphics[width=\textwidth]{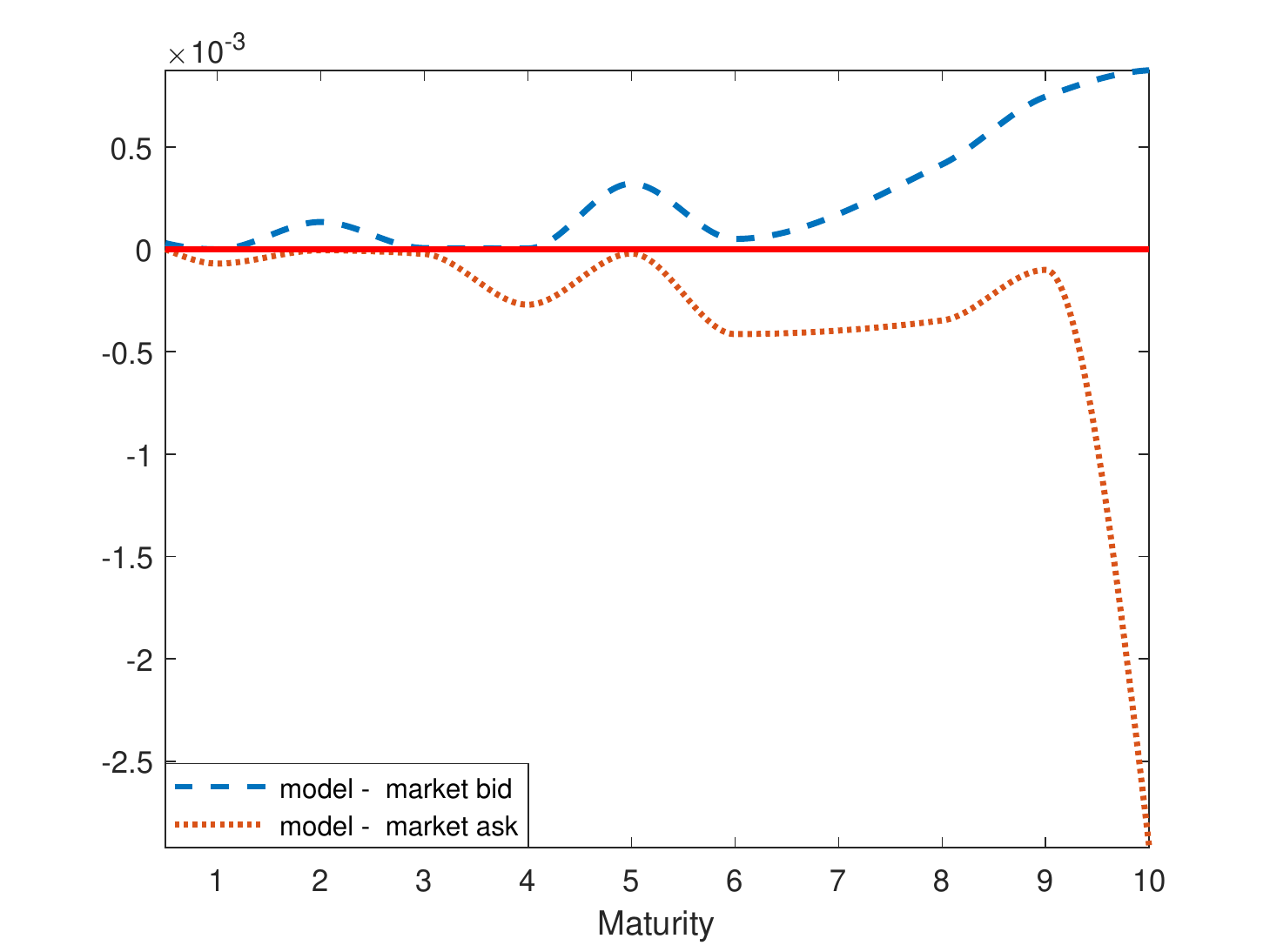}
       \caption{20 April 2016}
        \label{figxa}
    \end{subfigure}

    \begin{subfigure}[b]{0.53\textwidth}
        \includegraphics[width=\textwidth]{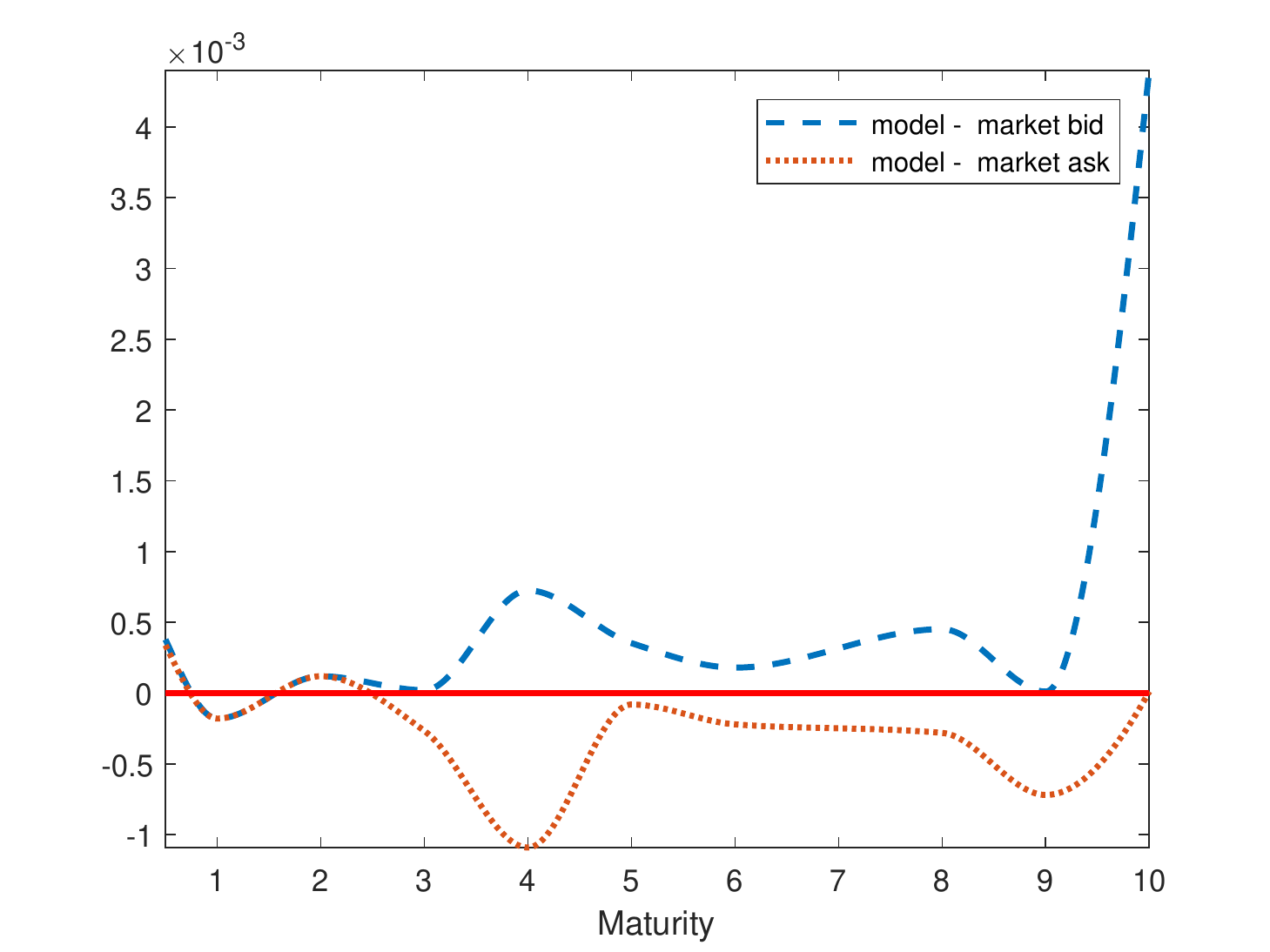}
       \caption{22 March 2017}
        \label{figx}
    \end{subfigure}
    ~
     \begin{subfigure}[b]{0.53\textwidth}
        \includegraphics[width=\textwidth]{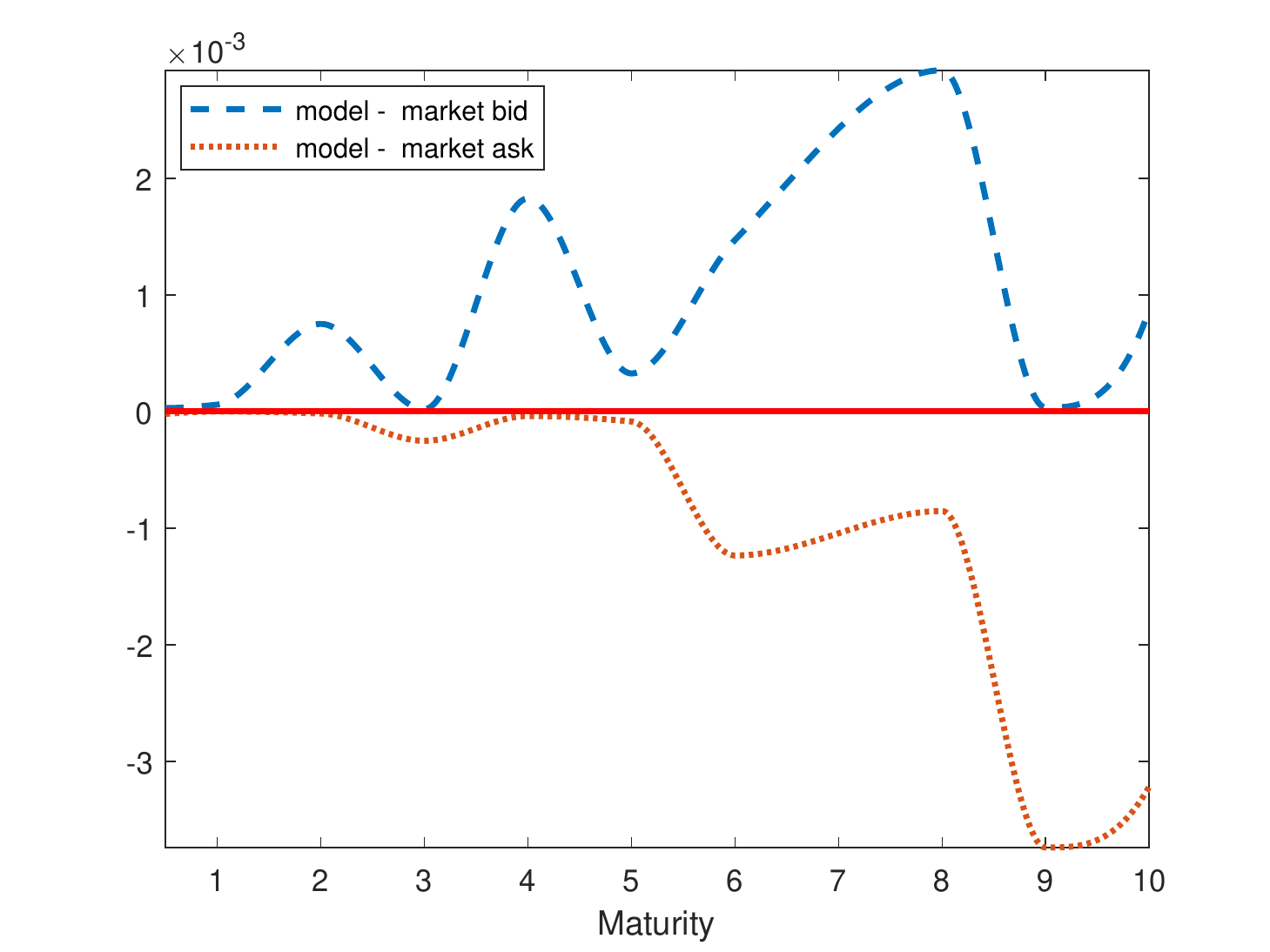}
       \caption{31 October 2017}
        \label{figxa}
    \end{subfigure}
    \caption{One--factor model fit to OIS discount factors (based on data from Bloomberg)}\label{OIS1Factor}
    \end{figure}
We use overnight index swap (OIS) data to extract the OIS discount factors $D^{\tiny \mbox{OIS}}(t,T)$ as per equations (\ref{DOIS}) and (\ref{DOISmult}). Note that this is model--independent.

Bloomberg provides OIS data out to a maturity of 30 years. However, there seems to be a systematic difference in the treatment of OIS beyond 10 years compared to maturities of ten years or less. This becomes particularly evident when one looks at the forward rates implied by the OIS discount factors: As can be seen in Figure \ref{OISfwd}, there is a marked spike in the forward rate at the 10--year mark. Therefore, for the time being we will focus on maturities out to ten years only in our model calibration.

The OIS discount factors depend only on the dynamics of $r_c$, thus calibrating to OIS discount factors only involves the parameters in equation (\ref{r_c}). As is standard practice in CIR--type term structure models,\footnote{See e.g. Brigo and Mercurio \cite{BM2006}, Section 4.3.4.} we can ensure that the model fits the observed (initial) term structure by the time dependence in the (deterministic) function $a_0(t)$. In the one--factor case ($d=1$), choosing $a_1=1$, we first fit the initial $y_1(0)$ and constant parameters $a_0$, $\kappa_1$, $\theta_1$ and $\sigma_1$ in such a way as to match the observed OIS discount factors on a given day as closely as possible,\footnote{If one were to include market instruments with option features in the calibration, the volatility parameters would be calibrated primarily to those instruments, rather than the shape of the term structure of OIS discount factors.} and then choose time--dependent (piecewise constant) $a_0(t)$ to achieve a perfect fit, as illustrated in Figure \ref{OIS1Factor}. In this figure, the difference between the discount factor based on OIS bid and the model discount factor is always negative, and the difference between the discount factor based on OIS ask and the model discount factor is always positive. This means that we have fitted the model to the market in the sense that the model price always lies between the market bid and ask.

\begin{table}[t]

  \centering
    \begin{tabular}{|r|rrrrr|}
   \hline
    \multicolumn{1}{|l|}{Factor} & \multicolumn{1}{l}{$y(0)$} & \multicolumn{1}{l}{$\kappa$} & \multicolumn{1}{l}{$\theta$} & \multicolumn{1}{l}{$\sigma$} & \multicolumn{1}{l|}{$a$} \\
     \hline
    \multirow{3}[2]{*}{3} & 0.049993 & 0.419187 & 0.199849 & 0.355935 & 0.003736 \\
          & 0.01501 & 0.691312 & 0.006249 & 0.08993 & 0.002735 \\
          & 0.002957 & 0.074725 & 0.18461 & 0.009337 & 0.00404 \\
     \hline
    1     & 0.476693 & 0.455794 & 0.134384 & 0.052677 & 0.000699 \\
    \hline
    \end{tabular}\\[1ex]

    \caption{CIR model calibrated parameters on 01/01/2013}\label{OISPar:20130101}%

   \begin{tabular}{|r|rrrrr|}
   \hline
    \multicolumn{1}{|l|}{Factor} & \multicolumn{1}{l}{$y(0)$} & \multicolumn{1}{l}{$\kappa$} & \multicolumn{1}{l}{$\theta$} & \multicolumn{1}{l}{$\sigma$} & \multicolumn{1}{l|}{$a$} \\
     \hline
    \multirow{3}[2]{*}{3} & 0.738431 & 0.5457939 & 0.8202671 & 0.850932 & 1.73E-07 \\
    3     & 0.013289 & 0.9216817 & 0.1453924 & 0.114773 & 0.004153 \\
          & 0.168724 & 0.5690233 & 0.09363364 & 0.005876 & 0.000693 \\
          \hline
    1     & 0.54686 & 0.7273657 & 0.04681492 & 0.205392 & 6.67E-05 \\

    \hline
    \end{tabular}\\[1ex]

    \caption{CIR model calibrated parameters on 08/09/2014}\label{OISPar:20140908}%

    \begin{tabular}{|c|ccccc|}
     \hline
    Factor & $y(0)$    & $\kappa$ & $\theta$ & $\sigma$ & $a$ \\
     \hline
    \multirow{3}[2]{*}{3} & 0.345831 & 0.038334 & 0.891026 & 0.024648 & 1.25E-06 \\
          & 0.127899 & 0.391148 & 0.90416 & 0.184931 & 3.78E-06 \\
          & 0.926643 & 0.694078 & 0.06601 & 0.160443 & 0.000129 \\
   \hline
    1     & 0.660211 & 0.49086 & 0.602925 & 0.318304 & 0.001602 \\
     \hline
    \end{tabular}%

    \caption{CIR model calibrated parameters on 18/06/2015}\label{OISPar:20150618}

    \begin{tabular}{|c|ccccc|}
     \hline
    Factor & $y(0)$    & $\kappa$ & $\theta$ & $\sigma$ & $a$ \\
    \hline
    \multirow{3}[2]{*}{3} & 0.241573 & 0.723951 & 0.262249 & 0.614391 & 0.006614 \\
          & 0.974125 & 0.435349 & 0.016192 & 0.022887 & 0.001613 \\
          & 0.155927 & 0.531023 & 0.088924 & 0.08861 & 0.006218 \\
     \hline
    1     & 0.002608 & 0.674544 & 0.051723 & 0.253174 & 0.01166 \\
   \hline
    \end{tabular}%

    \caption{CIR model calibrated parameters on 20/04/2016}\label{OISPar:20160420}

     \begin{tabular}{|c|ccccc|}
     \hline
    Factor & $y(0)$    & $\kappa$ & $\theta$ & $\sigma$ & $a$ \\
    \hline
    \multirow{3}[2]{*}{3} & 0.342 & 0.7045459 & 0.2785814 & 0.254248 & 0.003887 \\
    3     & 0.534943 & 0.6002792 & 0.966547 & 0.333788 & 0.001546 \\
          & 0.673438 & 0.4692102 & 0.1112721 & 0.206185 & 0.001157 \\
    \midrule
    1     & 0.340134 & 0.2156976 & 0.1746632 & 0.005657 & 9.40E-03 \\
   \hline
    \end{tabular}%

    \caption{CIR model calibrated parameters on 22/03/2017}\label{OISPar:201703122}

    \begin{tabular}{|c|ccccc|}
     \hline
    Factor & $y(0)$    & $\kappa$ & $\theta$ & $\sigma$ & $a$ \\
    \hline
    \multirow{3}[2]{*}{3} & 0.365879 & 0.2461593 & 0.1192502 & 0.11067 & 0.000127 \\
    \multicolumn{1}{|c|}{\multirow{2}[1]{*}{3}} & 0.028448 & 0.2058425 & 0.9971775 & 0.177739 & 0.000251 \\
          & 0.095642 & 0.9477889 & 0.4641101 & 0.312365 & 0.003171 \\
    \hline
    1     & 0.773084 & 0.2608761 & 0.7980566 & 0.264573 & 3.34E-03 \\
   \hline
    \end{tabular}%

    \caption{CIR model calibrated parameters on 31/10/2017}\label{OISPar:20171031}
\end{table}

\begin{table}[t]
\centering
\begin{subtable}{.4\textwidth}
\centering
    \begin{tabular}{|ccc|}
     \hline
    $T$     & $a_0^{(1)}$ & $a_0^{(3)}$ \\
     \hline
    0.5   & 1.03E-03 & 1.00E-03 \\
    1     & 0.001164 & 0.001524 \\
    2     & 0.001251 & 0.000397 \\
    3     & 0.001756 & 0.000795 \\
    4     & 0.000284 & 0.000684 \\
    5     & 0.00174 & 0.000206 \\
    6     & 0.002237 & 0.000958 \\
    8     & 0.004769 & 0.005657 \\
    9     & 0.017322 & 0.010415 \\
    10    & 0.023314 & 0.027489 \\
    \bottomrule
    \end{tabular}
\caption{01 January 2013}
\end{subtable}%
\begin{subtable}{.4\textwidth}
\centering
    \begin{tabular}{|rrr|}
     \hline
    $T$ &$a_0^{(1)}$ & $a_0^{(3)}$ \\
   \hline

    0.5   & 0.000923 & 0.000655 \\
    1     & 0.000871 & 0.000591 \\
    2     & 0.001037 & 0.000424 \\
    3     & 0.000964 & 0.000461 \\
    4     & 0.002122 & 0.001435 \\
    5     & 0.004211 & 0.003586 \\
    6     & 0.011486 & 0.011333 \\
    8     & 0.034919 & 0.032802 \\
    9     & 0.033426 & 0.031482 \\
    10    & 0.040708 & 0.041732 \\
     \hline
    \end{tabular}
\caption{08 September 2014}
\end{subtable}\vspace*{1.5em}

\begin{subtable}{.4\textwidth}
\centering
    \begin{tabular}{|rrr|}
     \hline
    $T$ &$a_0^{(1)}$ & $a_0^{(3)}$ \\
   \hline
    0.5   & 3.463E-04  & 5.09E-04   \\
    1     & 0.000807 & 0.000817 \\
    2     & 0.000901 & 0.001029 \\
    3     & 0.001155 & 0.000885 \\
    4     & 0.003721 & 0.003477 \\
    5     & 0.005422 & 0.005166\\
    6     & 0.016119 & 0.014454 \\
    8     & 0.022774 & 0.024241 \\
    9     & 0.035287 & 0.030412 \\
    10    & 0.037843 & 0.037339 \\
     \hline
    \end{tabular}
\caption{18 June 2015}
\end{subtable}
\begin{subtable}{.4\textwidth}
\centering
     \begin{tabular}{|ccc|}
    \hline
    $T$ &$a_0^{(1)}$ & $a_0^{(3)}$ \\
    \hline
    0.5   & 3.67E-03  & 0.000964 \\
    1     & 0.00399 & 0.000465 \\
    2     & 0.003692 & 0.000876 \\
    3     & 0.003975 & 0.001434 \\
    4     & 0.00521 & 0.0028 \\
    5     & 0.00534 & 0.003874 \\
    6     & 0.007987 & 0.00584 \\
    8     & 0.010661 & 0.008865 \\
    9     & 0.01513 & 0.014107 \\
    10    & 0.018535 & 0.015489 \\
     \hline
    \end{tabular}%
\caption{20 April 2016}
\end{subtable}\vspace*{1.5em}

\begin{subtable}{.4\textwidth}
\centering
    \begin{tabular}{|rrr|}
     \hline
    $T$ &$a_0^{(1)}$ & $a_0^{(3)}$ \\
   \hline

    0.5   & 0.006667 & 0.006766 \\
    1     & 0.009529 & 0.00978 \\
    2     & 0.012614 & 0.012493 \\
    3     & 0.015748 & 0.015599 \\
    4     & 0.016677 & 0.017447 \\
    5     & 0.019813 & 0.017904 \\
    6     & 0.019809 & 0.019014 \\
    8     & 0.030165 & 0.029841 \\
    9     & 0.002889 & 0.001362 \\
    10    & 0.026624 & 0.023907 \\

     \hline
    \end{tabular}
\caption{22 march 2017}
\end{subtable}
\begin{subtable}{.4\textwidth}
\centering
     \begin{tabular}{|ccc|}
    \hline
    $T$ &$a_0^{(1)}$ & $a_0^{(3)}$ \\
    \hline

    0.5   & 0.010865 & 0.012902 \\
    1     & 0.013164 & 0.014817 \\
    2     & 0.014373 & 0.01626 \\
    3     & 0.016675 & 0.016998 \\
    4     & 0.016015 & 0.019136 \\
    5     & 0.018735 & 0.017861 \\
    6     & 0.018916 & 0.021131 \\
    8     & 0.029504 & 0.029561 \\
    9     & 0.003794 & 0.00262 \\
    10    & 0.021191 & 0.022741 \\

     \hline
    \end{tabular}%
\caption{31 October 2017}
\end{subtable}\vspace*{1.5em}

\caption{Time-dependent parameter calibrations}\label{tab:a0}

\end{table}

The fitted parameters are shown in Tables \ref{OISPar:20130101} to \ref{tab:a0}.
\begin{figure}[t]
  \centering
    \begin{subfigure}[b]{0.53\textwidth}
        \includegraphics[width=\textwidth]{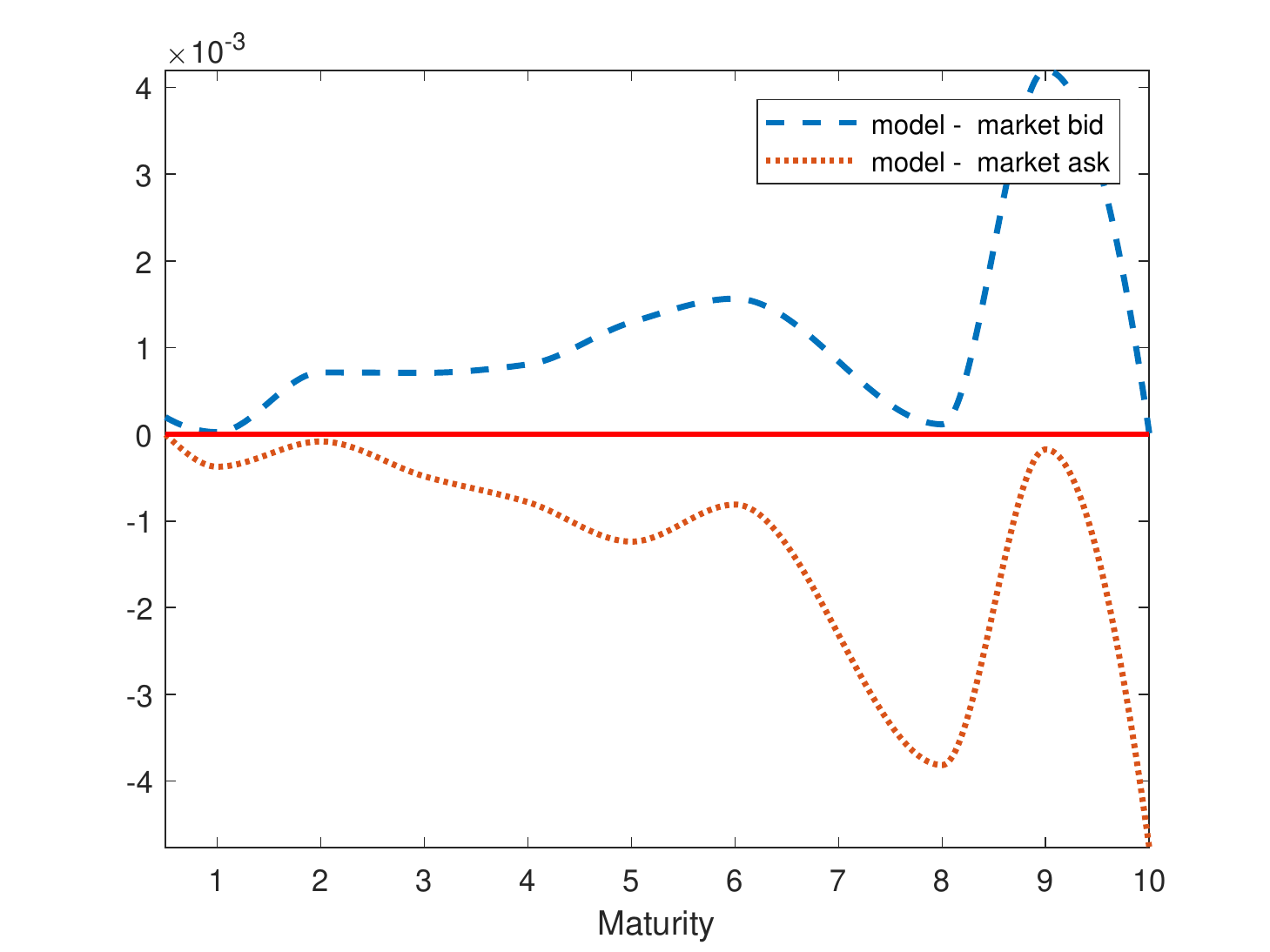}
        \caption{1 January 2013}
        \label{fig:ua}
    \end{subfigure}
    ~
    \begin{subfigure}[b]{0.53\textwidth}
        \includegraphics[width=\textwidth]{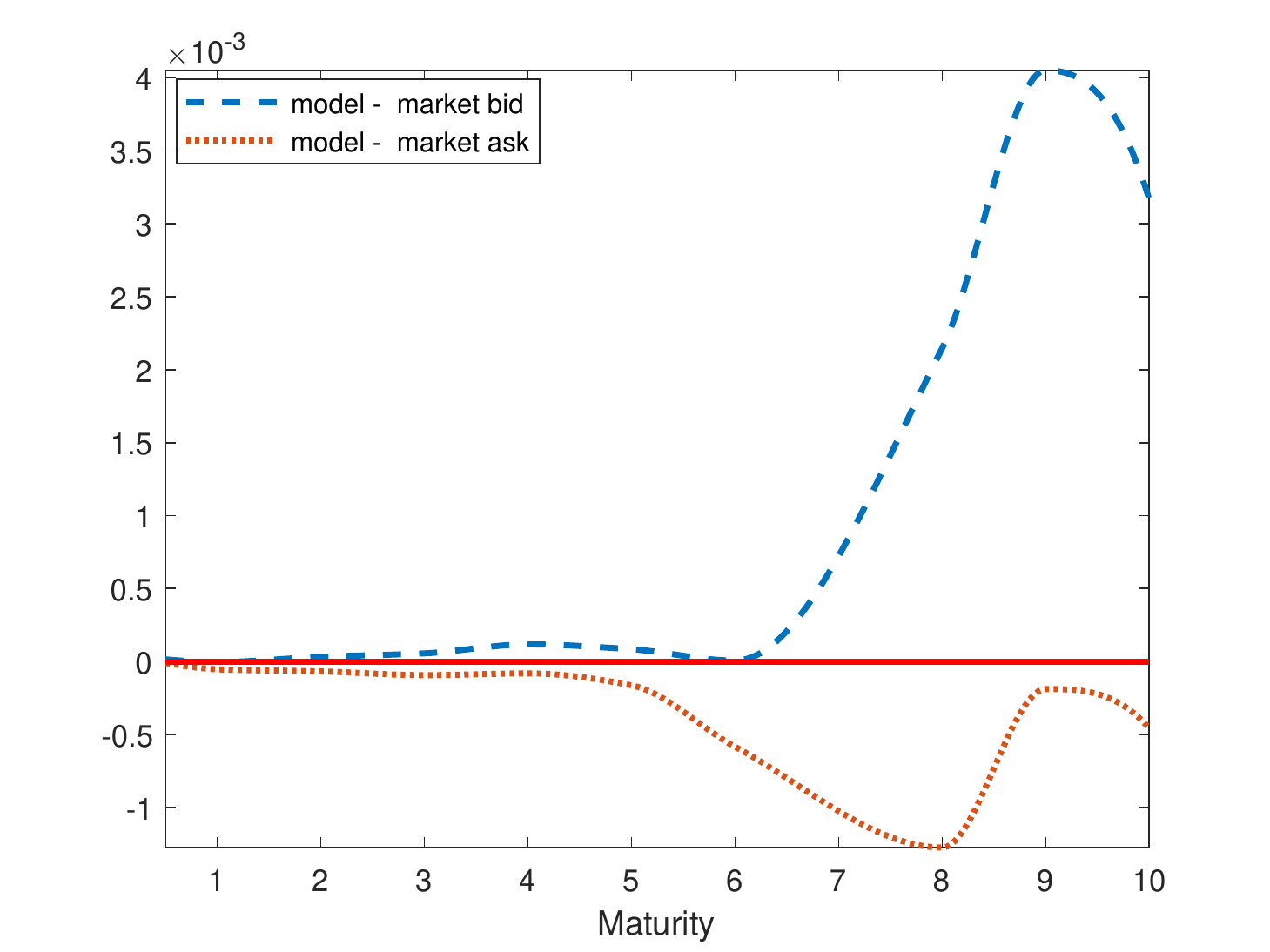}
       \caption{08 September 2014}
        \label{figxajun}
    \end{subfigure}\vspace*{1em}

 \begin{subfigure}[b]{0.53\textwidth}
        \includegraphics[width=\textwidth]{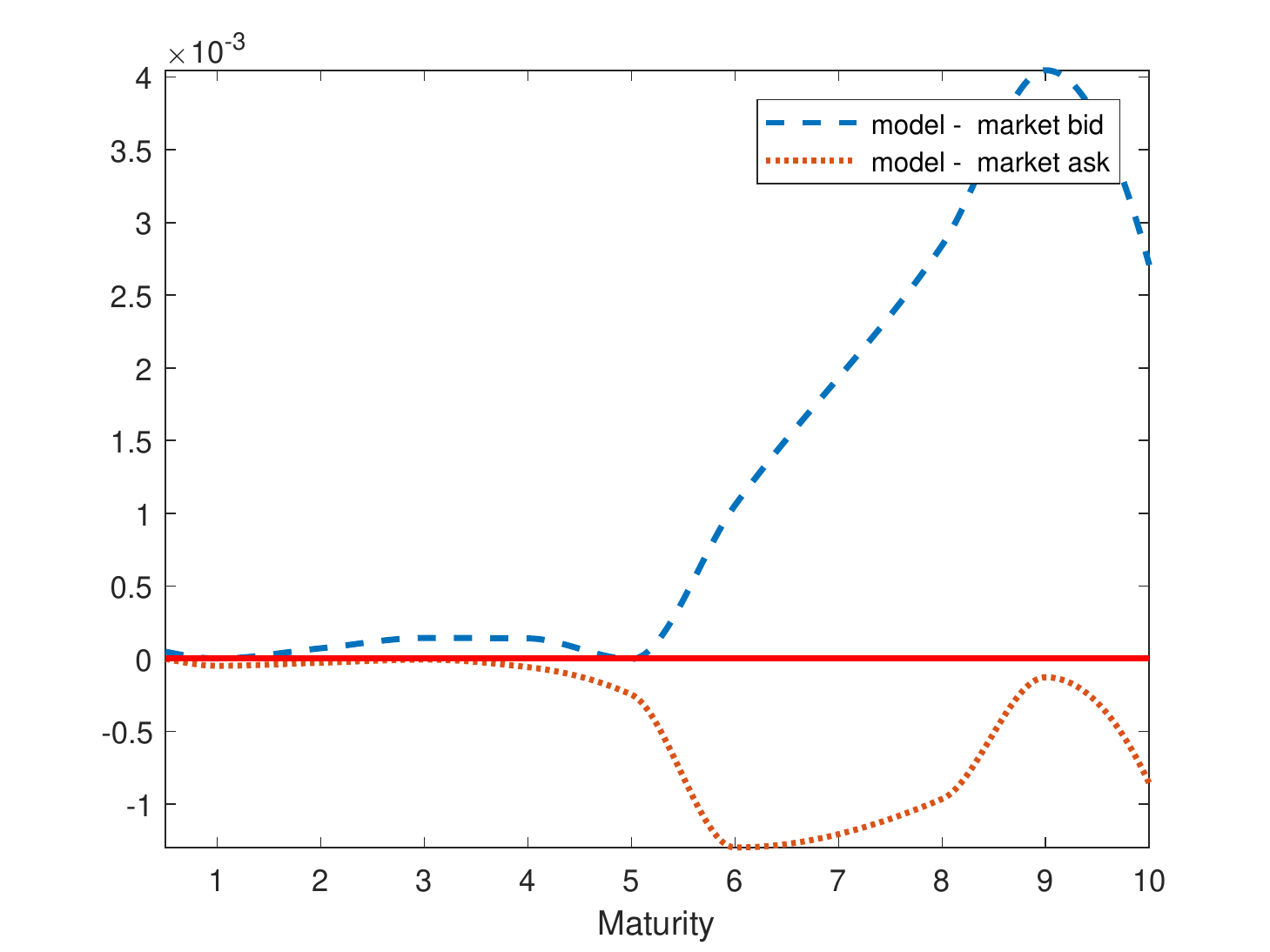}
       \caption{18 June 2015}
        \label{figxajun}
    \end{subfigure}
    ~
     \begin{subfigure}[b]{0.53\textwidth}
        \includegraphics[width=\textwidth]{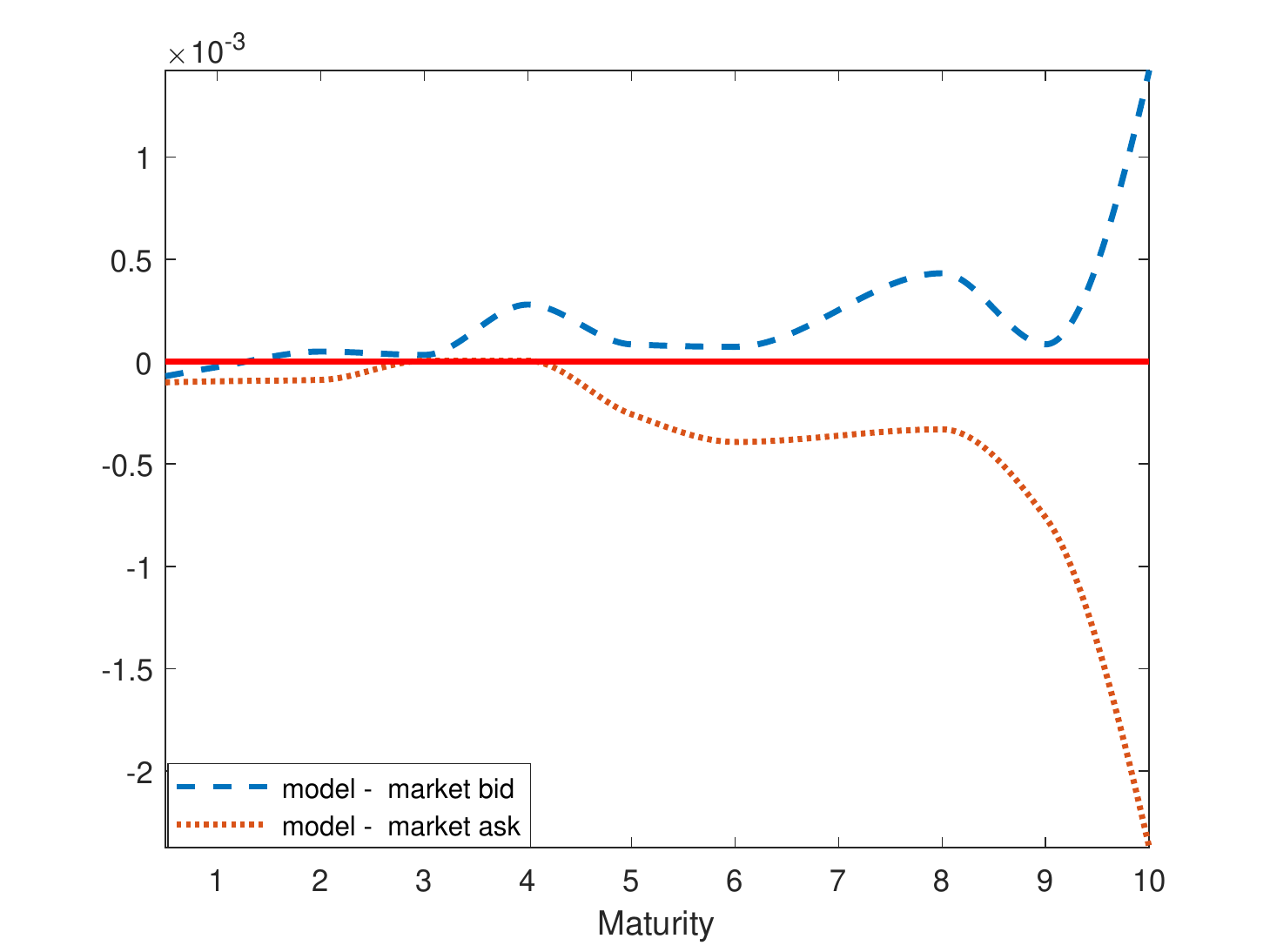}
       \caption{20 April 2016}
        \label{figxaapr}
    \end{subfigure}\vspace*{1em}

     \begin{subfigure}[b]{0.53\textwidth}
        \includegraphics[width=\textwidth]{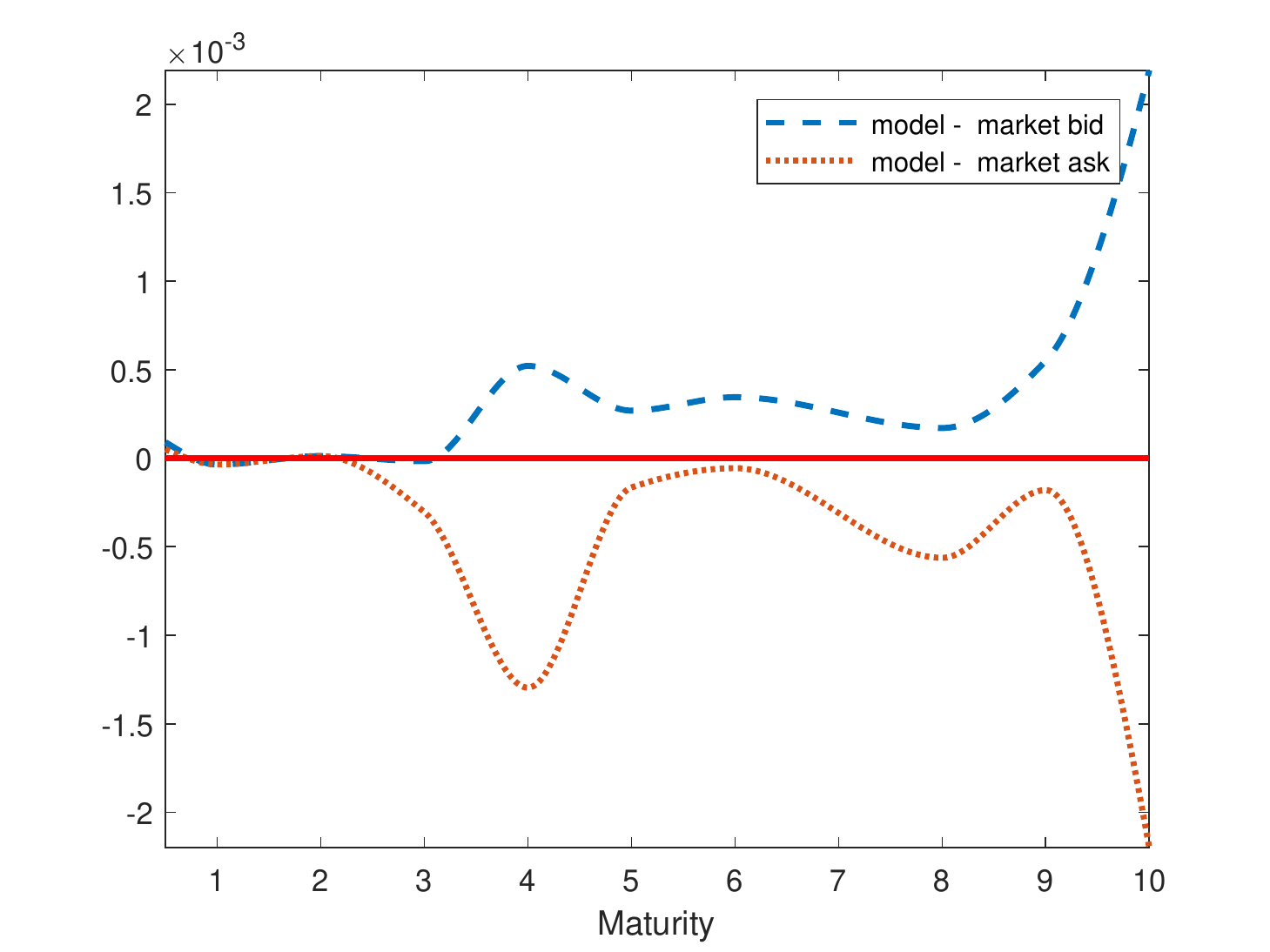}
       \caption{22 March 2017}
        \label{figx}
    \end{subfigure}
    ~
     \begin{subfigure}[b]{0.53\textwidth}
        \includegraphics[width=\textwidth]{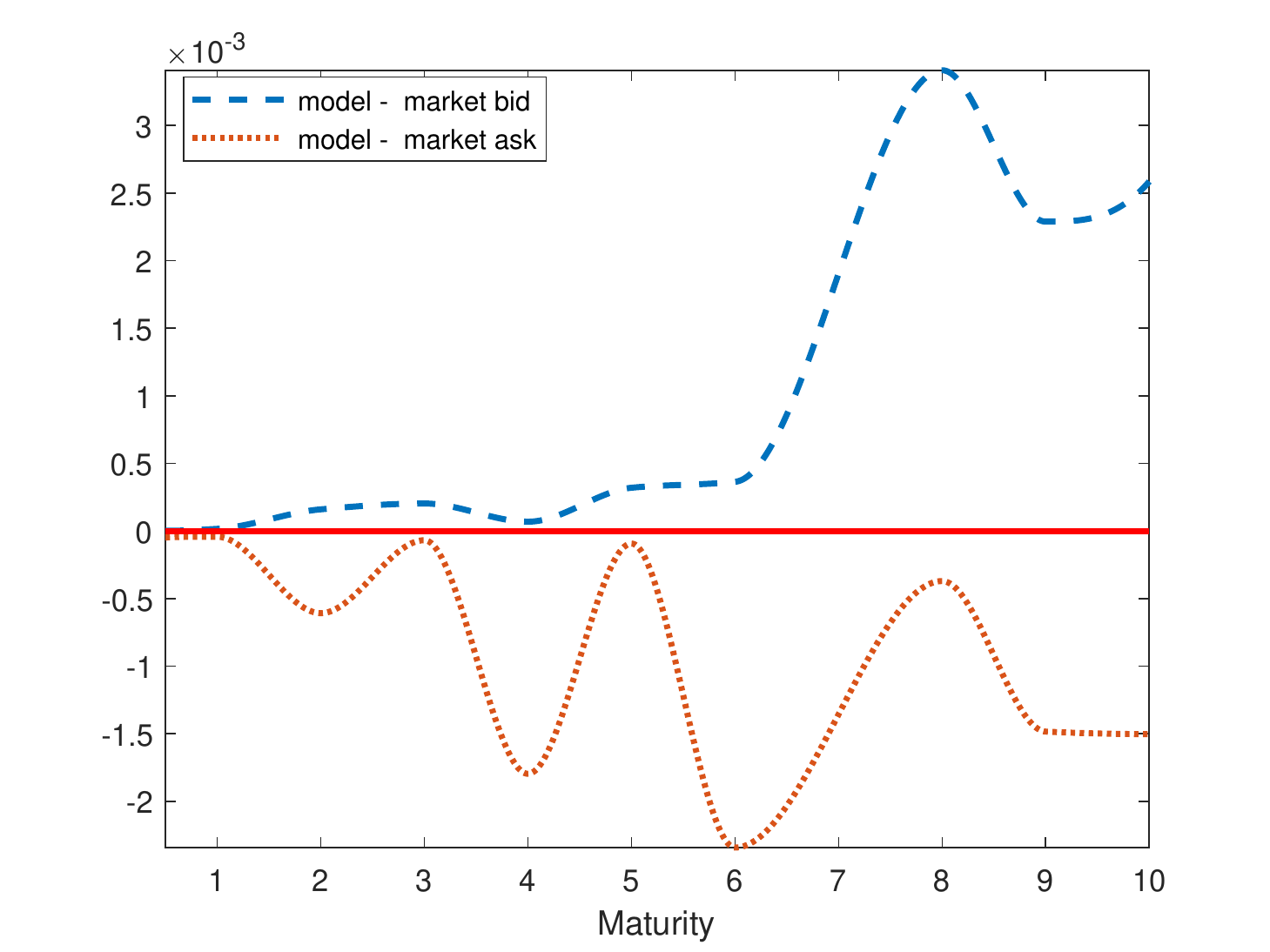}
       \caption{31 October 2017}
        \label{figxa}
    \end{subfigure}
\caption{Three--factor model fit to OIS discount factors (based on data from Bloomberg)}\label{OIS3Factor}
\end{figure}

\subsection{Fitting to OIS, vanilla and basis swaps}\label{OISswapcal}
The calibration condition for each vanilla swap is given by equation (\ref{swapcond}), and for each basis swap we obtain a condition as in (\ref{basisswapcond1}) or (\ref{basisswapcond2}). Specifically, in USD we have vanilla swaps exchanging a floating leg indexed to three--month LIBOR paid quarterly against a fixed leg paid semi-annually. We combine this with basis swaps for three months vs. six months and one months vs. three months, giving us three calibration conditions for each maturity (again, we restrict ourselves to maturities out to ten years for which data is available from Bloomberg, i.e. maturities of six months, 1, 2, 3, 4, 5, 6, 8, 9 and 10 years).

\begin{figure}[t]
  \centering
    \begin{subfigure}[b]{0.53\textwidth}
        \includegraphics[width=\textwidth]{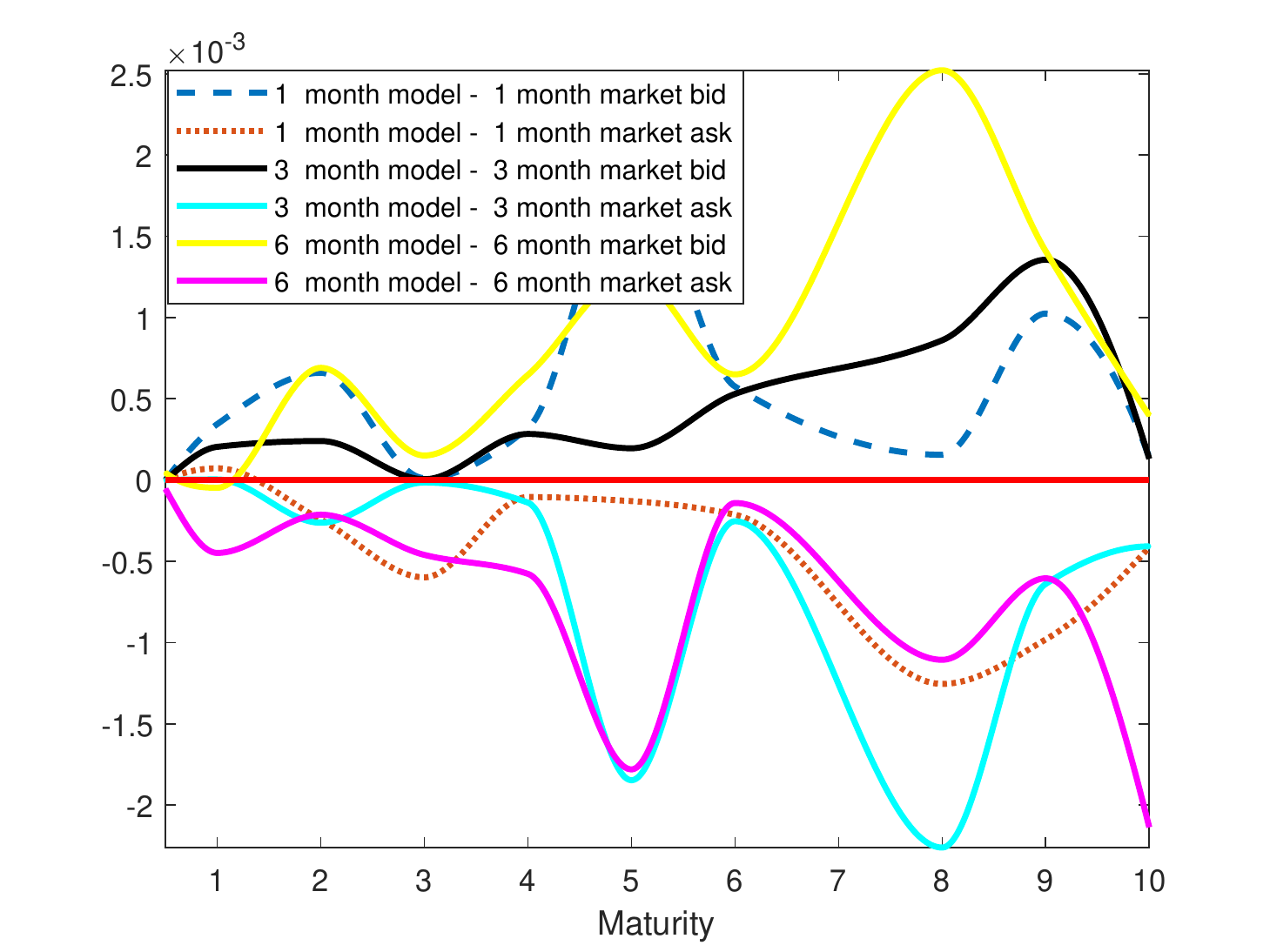}
        \caption{1 January 2013}
        \label{fig:ub}
    \end{subfigure}
    ~
    \begin{subfigure}[b]{0.53\textwidth}
        \includegraphics[width=\textwidth]{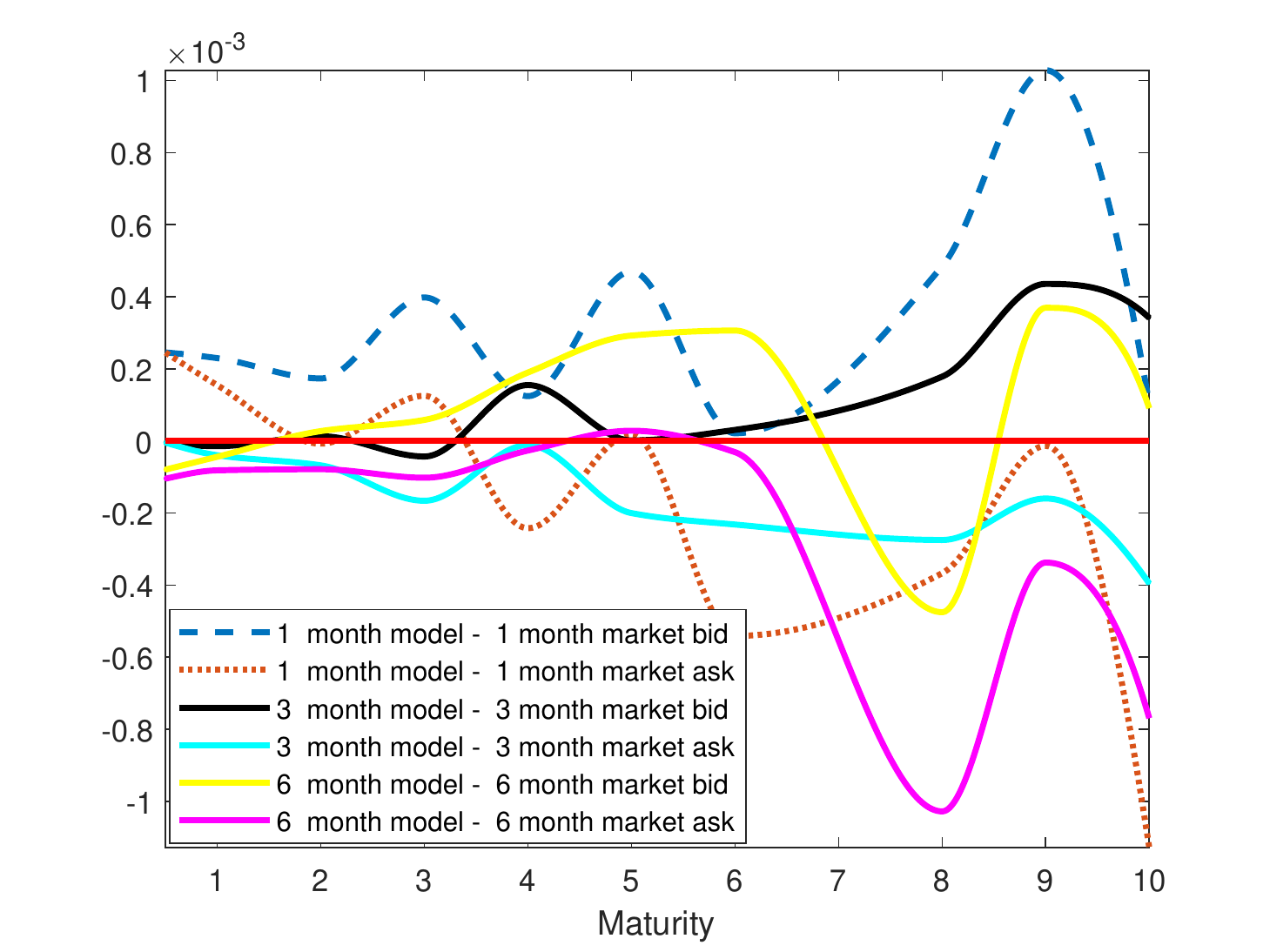}
       \caption{08 September 2014}
        \label{figxb}
    \end{subfigure}\vspace*{1em}

     \begin{subfigure}[b]{0.53\textwidth}
        \includegraphics[width=\textwidth]{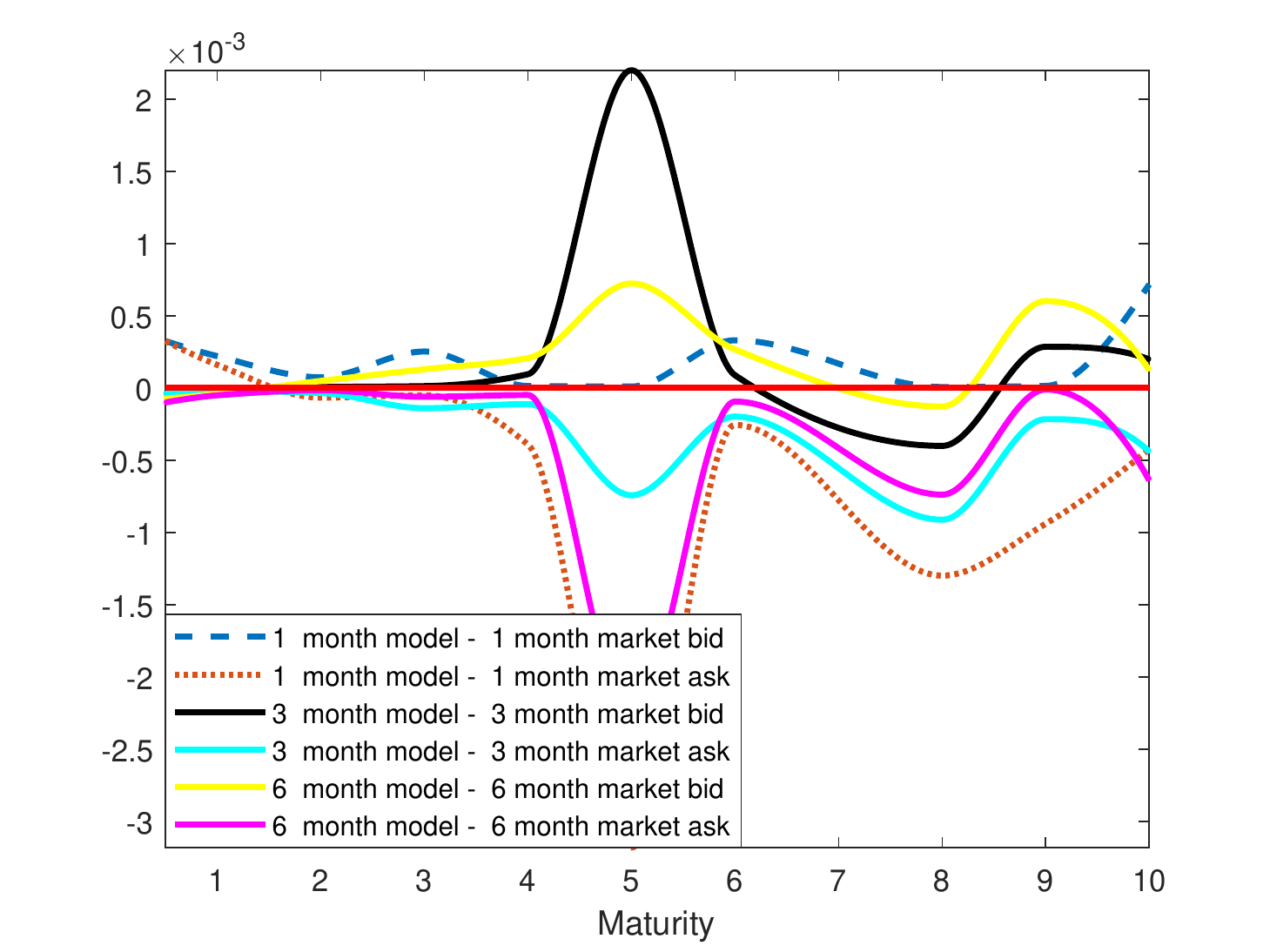}
       \caption{18 June 2015}
        \label{figxb}
    \end{subfigure}
      ~
    \begin{subfigure}[b]{0.53\textwidth}
        \includegraphics[width=\textwidth]{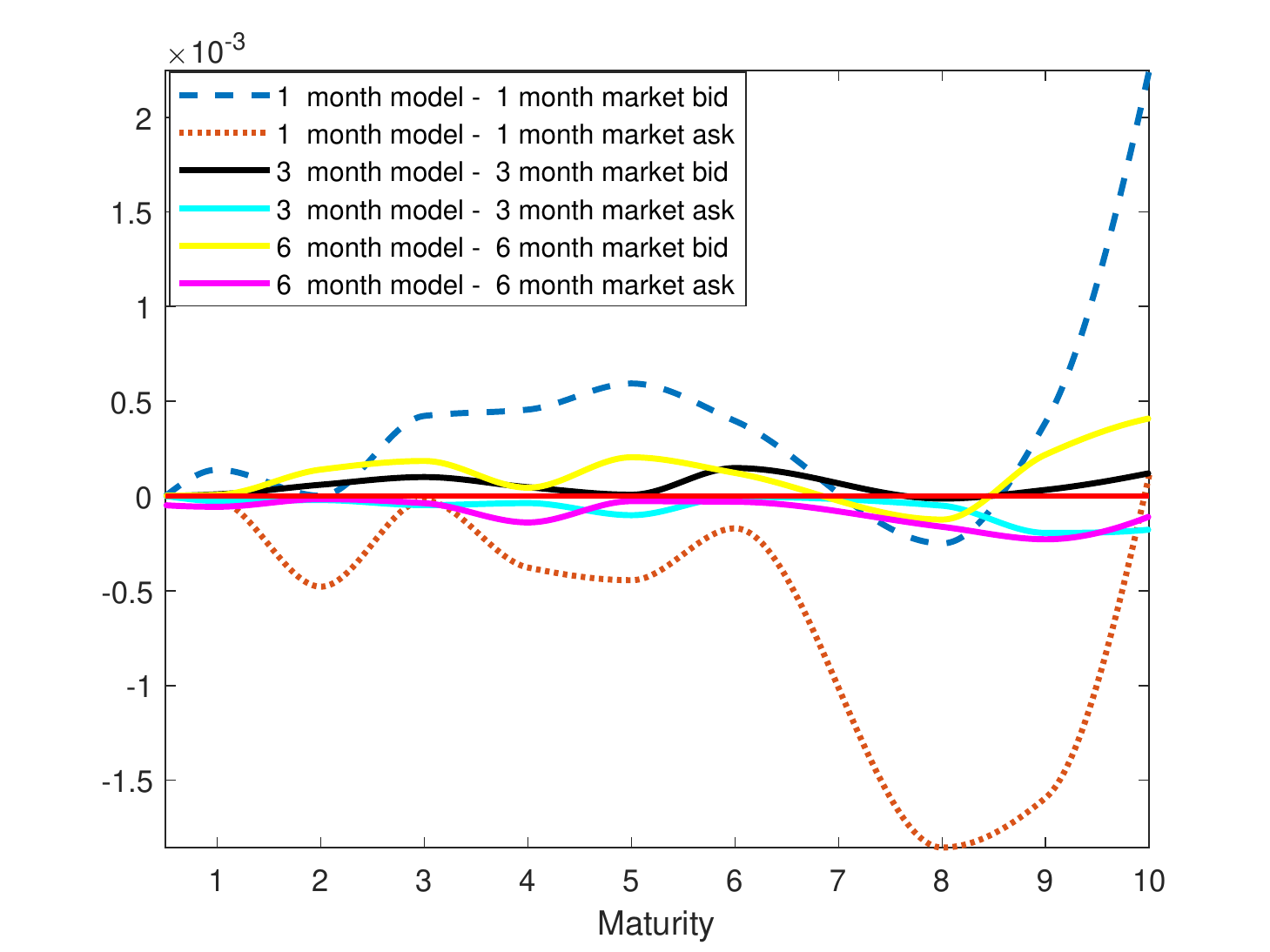}
       \caption{20 April 2016}
        \label{figxb}
    \end{subfigure}\vspace*{1em}

     \begin{subfigure}[b]{0.53\textwidth}
        \includegraphics[width=\textwidth]{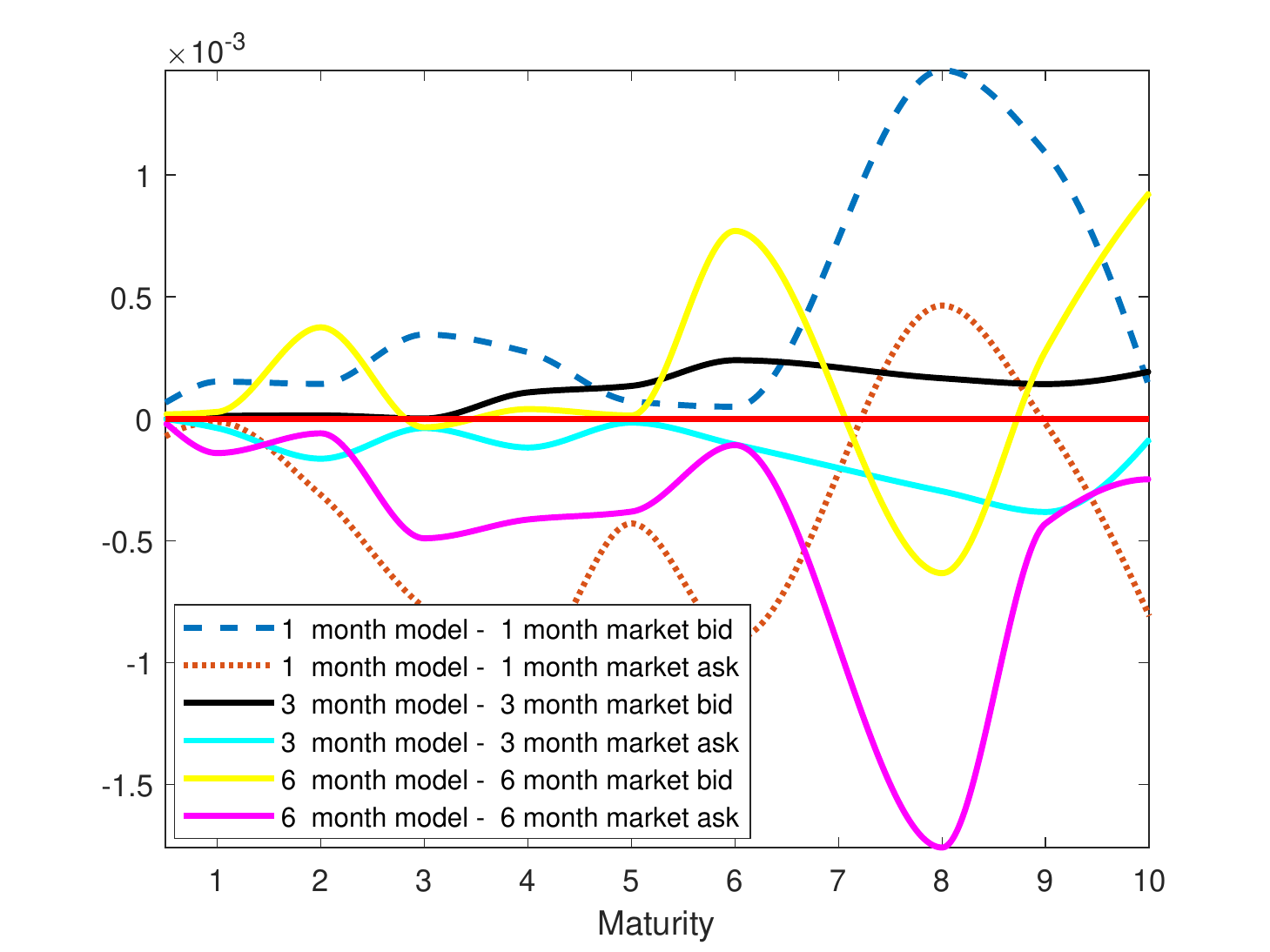}
       \caption{22 March 2017}
        \label{figx}
    \end{subfigure}
    ~
     \begin{subfigure}[b]{0.53\textwidth}
        \includegraphics[width=\textwidth]{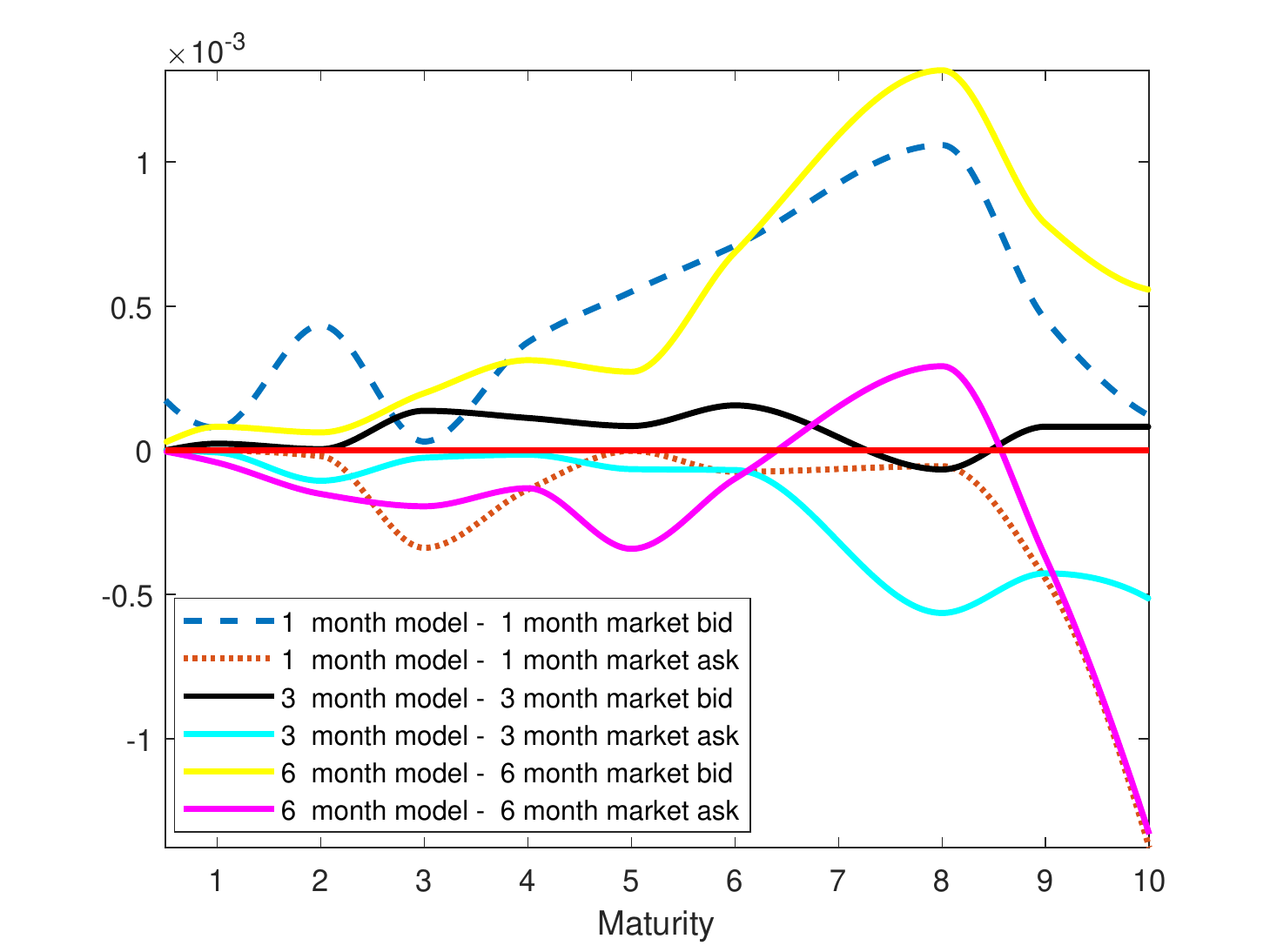}
       \caption{31 October 2017}
        \label{figxa}
    \end{subfigure}\vspace*{1em}

\caption{One--factor model fit to Basis Swaps (based on data from Bloomberg)}\label{basisconstpar1}
\end{figure}

Figure \ref{basisconstpar1} shows the fits of a one--factor model to vanilla and basis swap data obtained for 1 January 2013, 8 September 2014, 18 June 2015, 20 April 2016, 22 March 2017 and 31 October 2017, respectively. Table \ref{basispartbla} gives the corresponding model parameters.\footnote{Given that there is insufficient information contained on vanilla and basis swaps to calibrate the loss--in--default fraction $q$ ($q$ is essentially just a scaling parameter), we fix $q=0.6$, i.e. a default recovery rate of 40\%.}

As the fit of the one--factor model to the market is no longer perfect (though close) once vanilla and basis swap data are taken into account, we expand the number of factors to three ($d=3$), modifying the staged procedure to further facilitate the non-linear optimisation involved in the calibration: We first fit a one--factor model to OIS data as described in Section \ref{OISfit}, and then keep those parameters fixed in the three-factor model, setting $a_2=a_3=0$ to maintain the OIS calibration, and then fit the initial $y_2(0)$, $y_3(0)$ and constant parameters $d_0$, $b_1,b_2,b_3,c_1,c_2,c_3$, $\kappa_2$, $\kappa_3$, $\theta_2$, $\theta_3$ and $\sigma_2$, $\sigma_3$ in such a way as to match the swap calibration conditions on a given day as closely as possible.

\begin{table}[t]
\centering
\begin{subtable}{1.0\textwidth}
\centering
\begin{tabular}{|cccccccc|}
     \hline
    \multicolumn{1}{|c|}{$y(0)$} & \multicolumn{1}{c|}{a} & \multicolumn{1}{c|}{b} & \multicolumn{1}{c|}{c} & \multicolumn{1}{c|}{q} & \multicolumn{1}{c|}{$\sigma$} & \multicolumn{1}{c|}{$\kappa$} & $\theta$ \\
    \hline
    0.831418 & 0.000517 & 1.04E-05 & 0.000108 &       & 0.22479 & 0.278658 & 0.715343 \\
    0.094204 & 0     & 1.62E-05 & 0.000749 & 0.6 & 0.220092 & 0.352275 & 0.383409 \\
    0.100625 & 0     & 2.14E-05 & 0.001809 &       & 0.343971 & 0.334376 & 0.797952 \\
     \hline
    \end{tabular}%

\caption{01 January 2013}
\end{subtable}

\begin{subtable}{1.0\textwidth}
\centering
\begin{tabular}{|cccccccc|}
     \hline
    \multicolumn{1}{|c|}{$y(0)$} & \multicolumn{1}{c|}{a} & \multicolumn{1}{c|}{b} & \multicolumn{1}{c|}{c} & \multicolumn{1}{c|}{q} & \multicolumn{1}{c|}{$\sigma$} & \multicolumn{1}{c|}{$\kappa$} & $\theta$ \\
    \hline
0.54686 & 6.67E-05 & 0.000374 & 0.002578 &       & 0.205392 & 0.7273657 & 0.04681492 \\
    0.055658 & 0     & 0.003208 & 0.003004 & 0.6   & 0.355417 & 0.7175086 & 0.09971208 \\
    0.184441 & 0     & 0.000472 & 0.000216 &       & 0.568733 & 0.9197932 & 0.3863628 \\
     \hline
    \end{tabular}%

\caption{08 September 2014}
\end{subtable}

\begin{subtable}{1.0\textwidth}

\begin{tabular}{|cccccccc|}
    \hline
    \multicolumn{1}{|c|}{$y(0)$} & \multicolumn{1}{c|}{a} & \multicolumn{1}{c|}{b} & \multicolumn{1}{c|}{c} & \multicolumn{1}{c|}{q} & \multicolumn{1}{c|}{$\sigma$} & \multicolumn{1}{c|}{$\kappa$} & $\theta$ \\
     \hline
     {0.483062} & 0.001577 & 1.53E-03 & 0.004607 &       & \multicolumn{1}{r}{0.29234} & \multicolumn{1}{r}{0.94945} & {0.326688} \\
    0.100874 & 0     & 4.29E-05 & 0.00285 & \multicolumn{1}{c}{0.6} & 0.500709 & 0.780582 & 0.187569 \\
    0.326835 & 0     & 2.51E-05 & 0.000325 &       & 0.124841 & 0.898274 & 0.363523 \\
    \hline
    \end{tabular}%
\caption{18 June 2015}
\end{subtable}%

\begin{subtable}{1.0\textwidth}

\begin{tabular}{|cccccccc|}
    \hline
    \multicolumn{1}{|c|}{$y(0)$} & \multicolumn{1}{c|}{a} & \multicolumn{1}{c|}{b} & \multicolumn{1}{c|}{c} & \multicolumn{1}{c|}{q} & \multicolumn{1}{c|}{$\sigma$} & \multicolumn{1}{c|}{$\kappa$} & $\theta$ \\
    \hline
    0.002608 & 0.01166 & 0.00455 & 0.016407 &       & 0.253174 & 0.674544 & 0.051723 \\
    0.074159 & 0     & 0.001414 & 0.002822 & 0.6   & 0.008161 & 0.031635 & 0.156863 \\
    0.009814 & 0     & 0.002227 & 0.004898 &       & 0.377 & 0.261372 & 0.431784 \\
    \hline
    \end{tabular}%
\caption{20 April 2016}
\end{subtable}

\begin{subtable}{1.0\textwidth}

\begin{tabular}{|cccccccc|}
    \hline
    \multicolumn{1}{|c|}{$y(0)$} & \multicolumn{1}{c|}{a} & \multicolumn{1}{c|}{b} & \multicolumn{1}{c|}{c} & \multicolumn{1}{c|}{q} & \multicolumn{1}{c|}{$\sigma$} & \multicolumn{1}{c|}{$\kappa$} & $\theta$ \\
    \hline
   0.340134 & 9.40E-03 & 2.96E-05 & 9.37E-05 &       & 0.0056574 & 0.2156976 & 0.174663 \\
    0.000873 & 0     & 0.002236 & 0.001796 & 0.6   & 0.04284526 & 0.1823359 & 0.014394 \\
    0.000873 & 0     & 0.000662 & 0.000394 &       & 0.03652793 & 0.9609686 & 0.031956 \\

    \hline
    \end{tabular}%
\caption{22 March 2017}
\end{subtable}

\begin{subtable}{1.0\textwidth}

\begin{tabular}{|cccccccc|}
    \hline
    \multicolumn{1}{|c|}{$y(0)$} & \multicolumn{1}{c|}{a} & \multicolumn{1}{c|}{b} & \multicolumn{1}{c|}{c} & \multicolumn{1}{c|}{q} & \multicolumn{1}{c|}{$\sigma$} & \multicolumn{1}{c|}{$\kappa$} & $\theta$ \\
    \hline
   0.773084 & 3.34E-03 & 5.39E-05 & 3.72E-05 &       & 0.264573 & 0.2608761 & 0.7980566 \\
    0.013896 & 0     & 0.113603 & 0.008913 & 0.6   & 0.004227 & 0.3975118 & 0.0002119 \\
    0.065454 & 0     & 7.94E-05 & 4.09E-06 &       & 0.403512 & 0.9037867 & 0.8058104 \\
    \hline
    \end{tabular}%
\caption{31 October 2017}
\end{subtable}
\caption{Basis model parameters}\label{basispartbla}
\end{table}

\begin{figure}[t]
  \centering
    \begin{subfigure}[b]{0.53\textwidth}
        \includegraphics[scale=0.45]{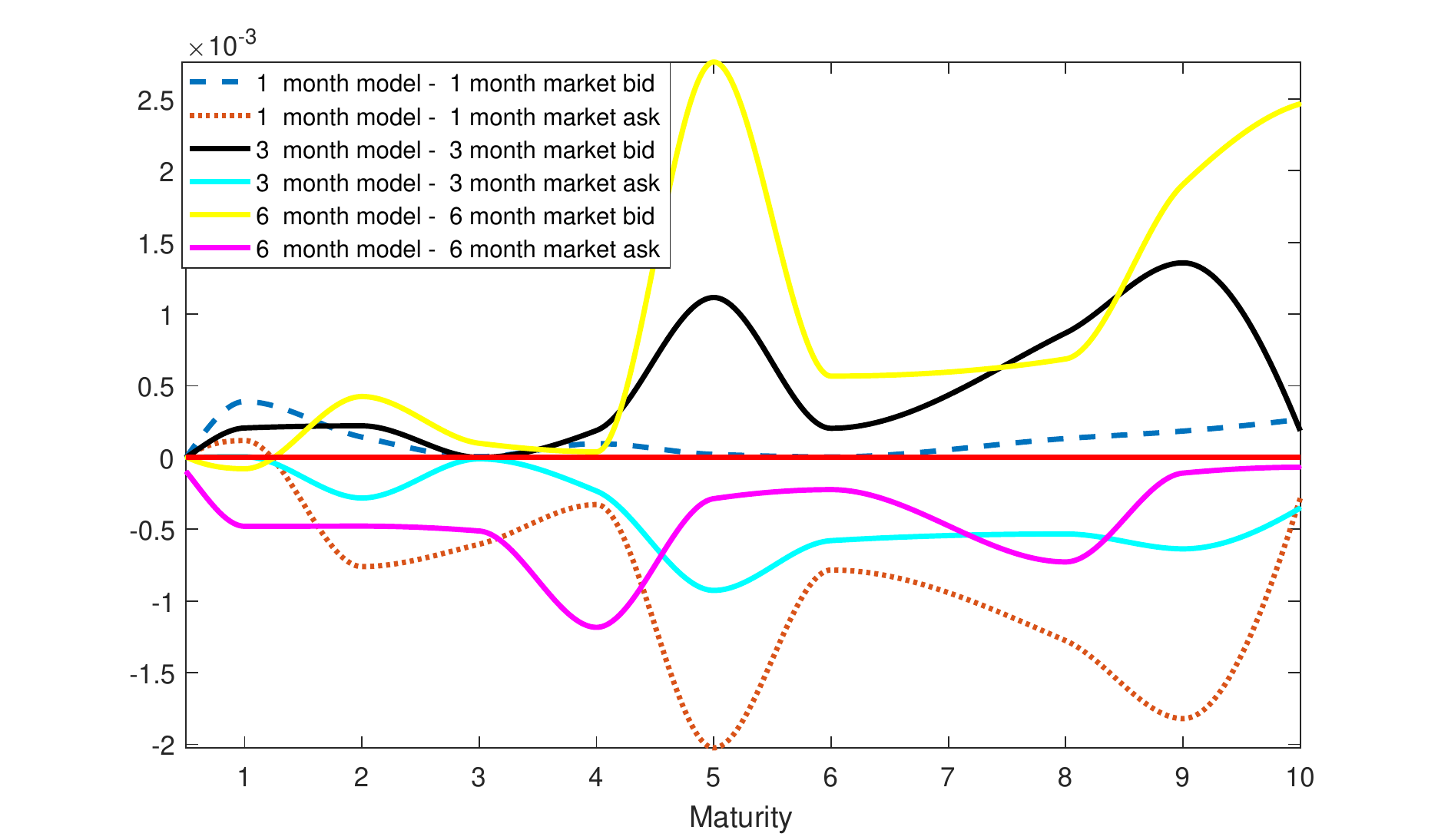}
        \caption{1 January 2013}
        \label{fig:ud}
    \end{subfigure}
    ~
        \begin{subfigure}[b]{0.53\textwidth}
        \includegraphics[scale=0.45]{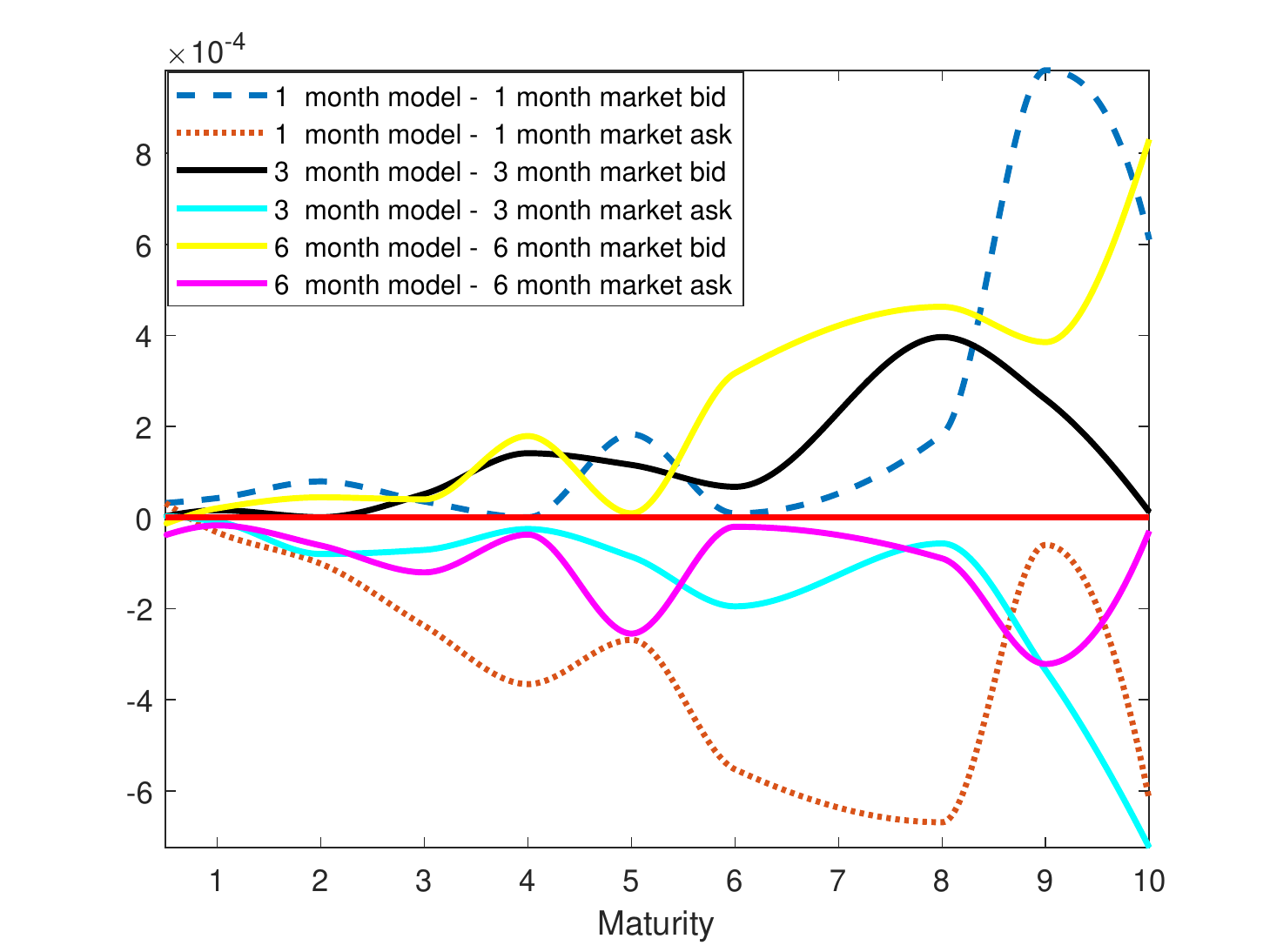}
       \caption{08 September 2014}
        \label{figxd}
    \end{subfigure}

    \begin{subfigure}[b]{0.53\textwidth}
        \includegraphics[scale=0.45]{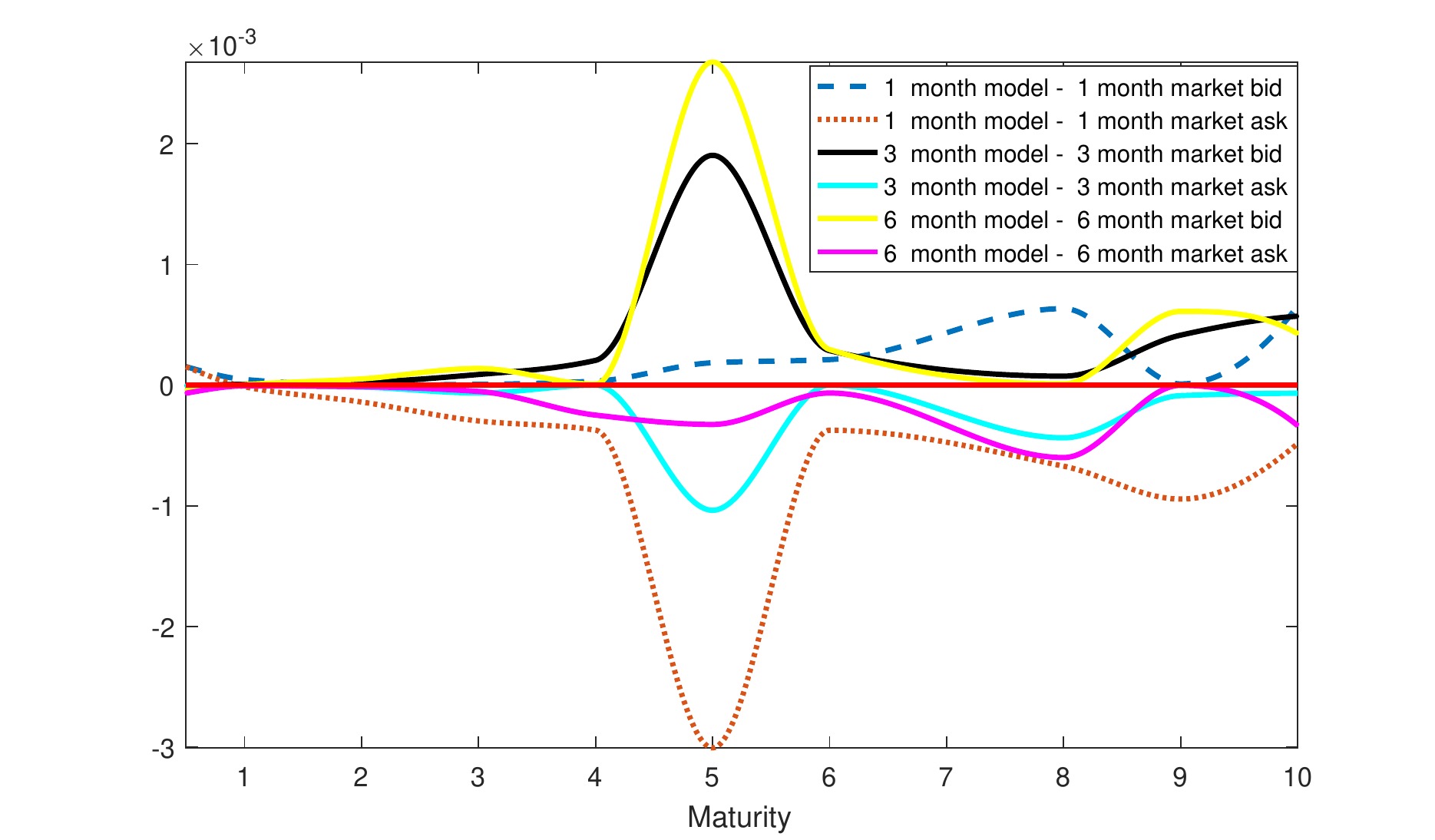}
       \caption{18 June 2015}
        \label{figxd}
    \end{subfigure}
    ~
    \begin{subfigure}[b]{0.53\textwidth}
        \includegraphics[scale=0.45]{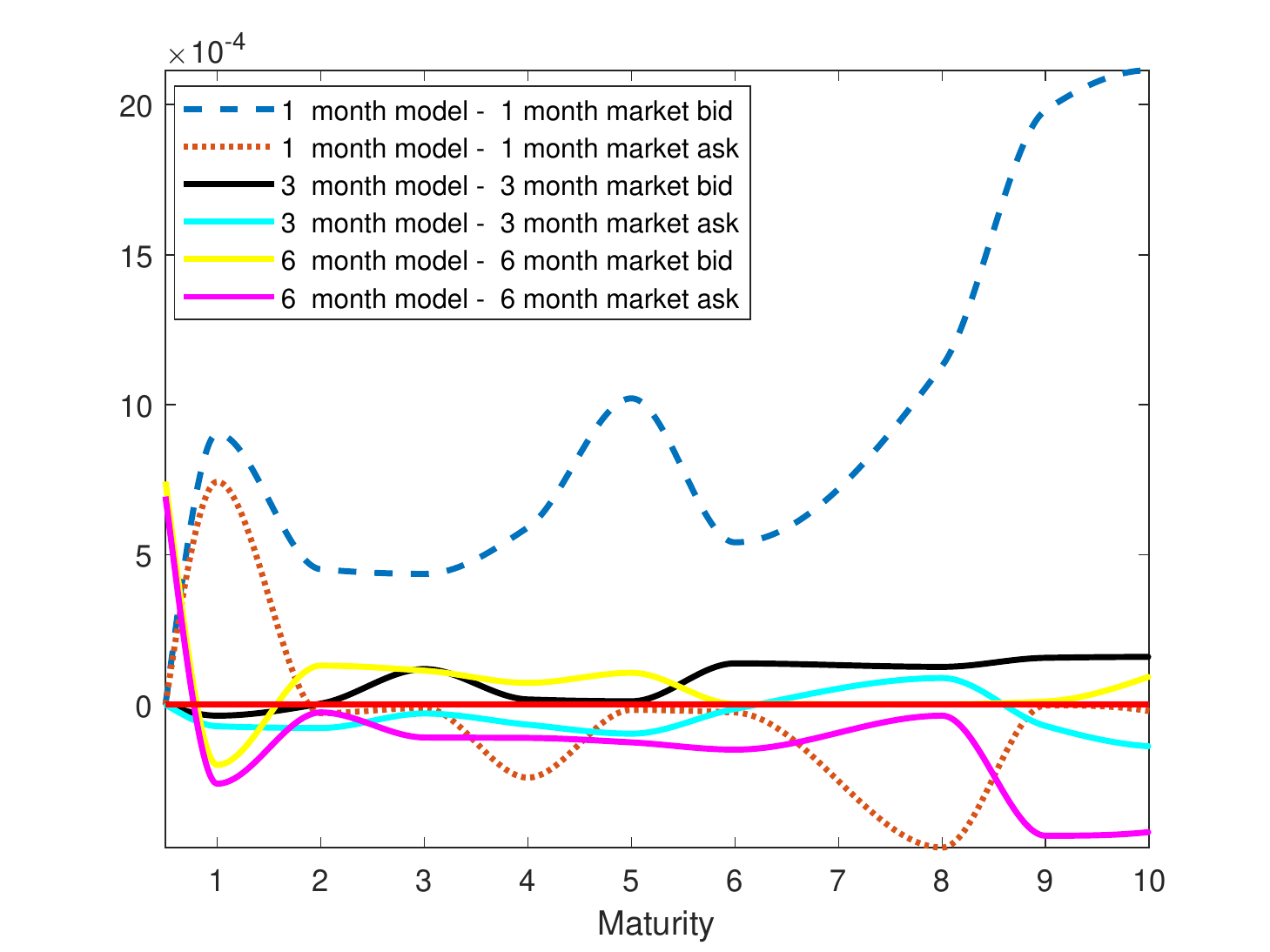}
       \caption{20 April 2016}
        \label{figxb}
    \end{subfigure}

     \begin{subfigure}[b]{0.53\textwidth}
        \includegraphics[width=\textwidth]{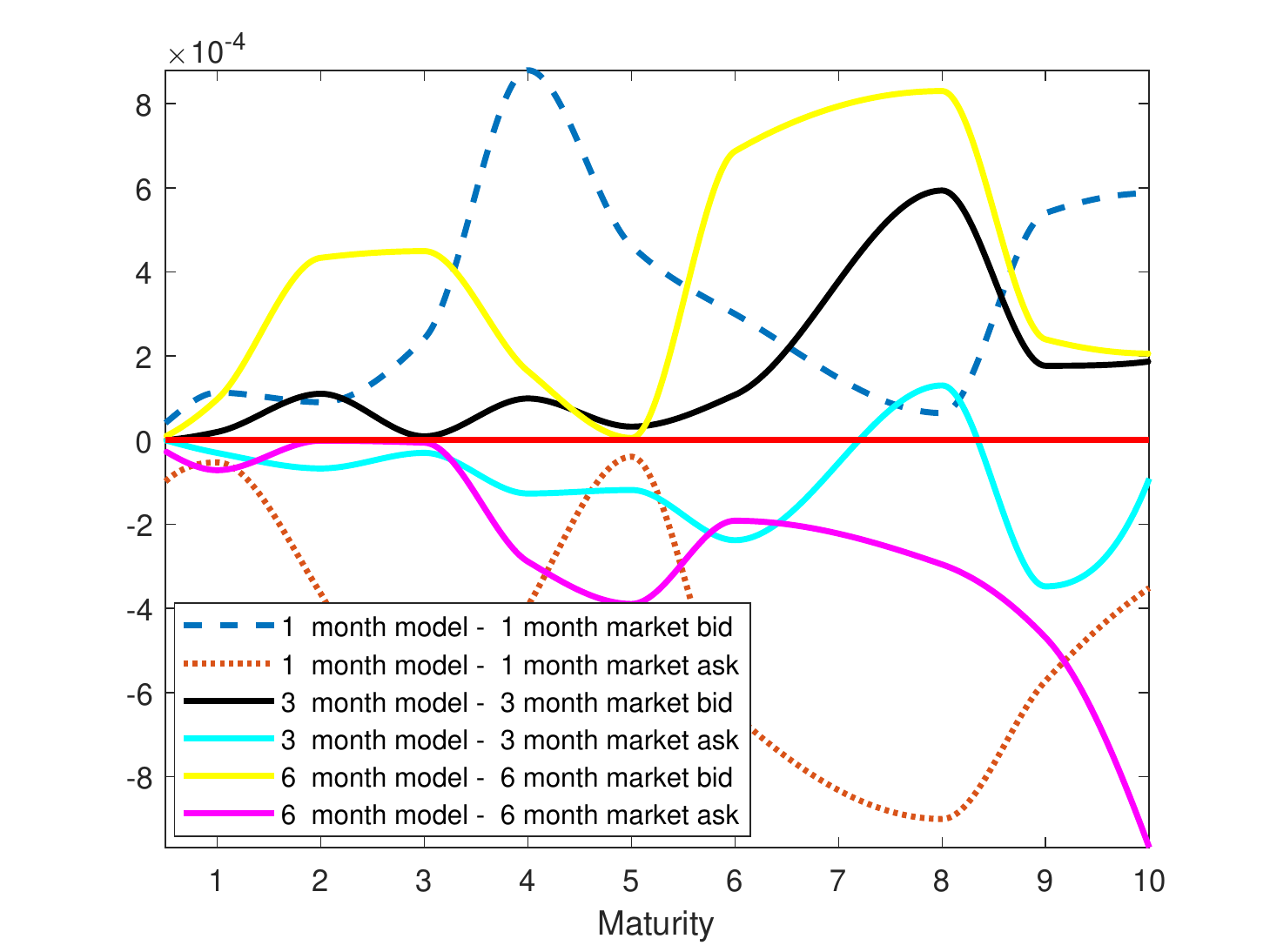}
       \caption{22 March 2017}
        \label{figx}
    \end{subfigure}
    ~
     \begin{subfigure}[b]{0.53\textwidth}
        \includegraphics[width=\textwidth]{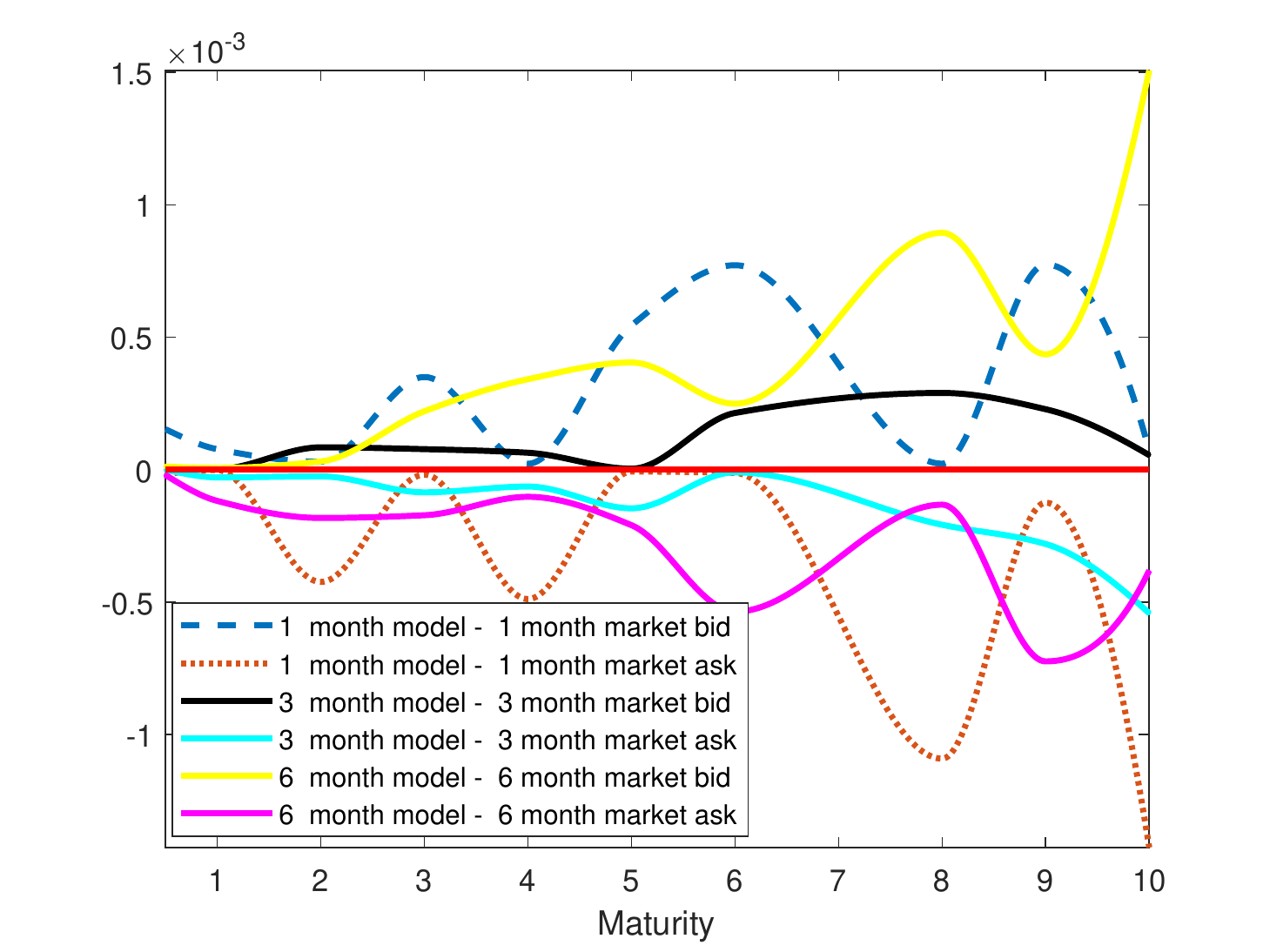}
       \caption{31 October 2017}
        \label{figxa}
    \end{subfigure}
\caption{Tenors fit to market data  (based on data from Bloomberg)}\label{basisfit}
\end{figure}

Lastly, in order to improve the fit further, we allow for time--dependent (piecewise constant) $d_0$. Figure \ref{basisfit} shows
the resulting fits for 1 January 2013, 8 September 2014, 18 June 2015, 20 April 2016, 22 March 2017 and 31 October 2017, respectively, and Table \ref{basispartbla} gives the corresponding model parameters.

From these results, we note that on a given day, the model can be simultaneously calibrated to all available OIS and vanilla/basis swap data, resulting in a fit between the bid and ask prices for all maturities, except for small discrepancies, mainly for maturities up to a year.


\section{Separating Funding Liquidity Risk and Credit Risk Using CDS}\label{CDScal}
This section aims to disentangle the components $\phi(s)$ and $\lambda(s)$ by including credit default swaps (CDS) in the calibration. The coupon leg of the CDS consists of a streams of payments while there is no default and the protection leg of the CDS only pays when default happens before maturity of the CDS. Since the reference rates for the vanilla and basis swaps are LIBORs, in our CDS calibration we focus on banks on the LIBOR panel. Because the payoff of a CDS is contingent on default of its reference entity, default must now be modelled explicitly, i.e. we also have to disentangle the default--free interest rate $r$ and the default intensity, rather than directly modelling the LHS of \ref{collateral}, $r_c$. To do this, one needs information on the systemic default risk component $\Lambda$ of $r_c$ as specified in (\ref{collateral}). However, identifying this from market information is not reliably possible, as Filipovi\'c and Trolle (\citeyear*{FilTro:2013}) also have found.\footnote{Typical candidates as proxies for a risk--free rate, i.e. government bonds or repo rates, often result in term structures which intersect the term structure implied by OIS, indicating that the spread between OIS and these rates is not a reliable proxy for $\Lambda$.} They fix $\Lambda$ at 5 basis points (bp), noting that ``reasonable variations'' in the value of $\Lambda$ do not change their results, a point which also applies in our context, and thus we follow this modelling choice.

%

\subsection{Setting up the calibration}
Let $\tau^j$ be the default time of the bank $j$ in the panel. Define a full filtration $\mathcal{F}^j_t=\mathcal{F}_t^{\tau^j}\vee \mathcal{H}_t$, where
\[\mathcal{F}_t^{\tau^j}=\sigma(\{\tau^j>u\}:u\leq t)\]
is a subfiltration generated by $\tau^j$ and $\mathcal{H}_t$ is a subfiltration that contains default--free information, also known as ``market filtration''. Then the default intensity $\hat\lambda_j(s)=\Lambda(s)+\lambda_j(s)$ of bank $j$ is $\mathcal{H}_t$-adapted, and $\lambda_j$ gives the departure of $\hat\lambda_j$ of bank $j$ from the panel's systemic $\Lambda$. We then take the $\lambda$ required by the roll--over risk pricing conditions derived in Section \ref{modelsection} to be the average over the $\lambda_j$ in the panel\footnote{We note that LIBOR is determined by taking the trimmed average of the corresponding simple--compounded rates quoted by the panel banks. Thus, for the sake of tractability, with this assumption we are making two approximations: Firstly, taking the relationship between $\lambda_j$ and the simple--compounded rate quoted by the $j$-th bank to be approximately linear, and secondly assuming that the trimmed outliers lie approximately symmetrically above and below the panel average (in which case the trimmed average is approximately equal to the untrimmed average).} and consequently the coefficients of the dynamics of $\lambda$ are the averages (due to linearity in coefficients) of the corresponding coefficients of the $\lambda_j$.

Analogously to (\ref{r_c})--(\ref{phi}), define $r_c$ as
\begin{eqnarray}
r_c(t)    &=& \alpha_0(t)+\langle \alpha , X(t) \rangle\label{r}
\end{eqnarray}
where as before the coefficients $\alpha_0(t), \alpha$ are calibrated to OIS discount factors via condition \ref{DOIScond}. Following Filipovi\'c and Trolle (\citeyear*{FilTro:2013}), $\Lambda(t)$ is fixed identically equal to 5 basis points, and we obtain the default--free rate $r$ via
\begin{equation}
r(t)=r_c(t)-\Lambda q
\end{equation}
The idiosyncratic default intensity associated with the $j$--th bank follows
\begin{eqnarray}
\hat{\lambda}_j(t) &=&\Lambda+ \lambda_j(t) =\hat{b}_0^j(t)+\langle \hat{b}^j , X(t) \rangle,\label{lambdaj} \end{eqnarray}
and the aim is to calibrate the coefficients of these dynamics to the corresponding credit default swaps. As noted above, $\lambda(t)$ is calculated as the average of $ \lambda_j(t)$. We can then write
$$r_c(t)+\lambda(t)q=d_0(t)+\langle d, X(t) \rangle,$$
Fixing $q=0.6$ as before and noting
\begin{eqnarray}
d_0(t) &=& a_0(t)+q\hat b_0(t)-q\Lambda\\
d_i &=& a_i+qb_i,
\end{eqnarray}
we are left with the freedom to determine the coefficients $c_0(t), c$ of the dynamics (\ref{phi}) of the ``funding liquidity'' component $\phi$ of roll--over risk. We obtain $c_0(t), c$ by calibrating to vanilla and basis swaps via conditions (\ref{swapcond}) to (\ref{basisswapcond2}).

\subsection{Pricing the CDS}
The time $t$-value of a zero--recovery zero coupon bond with notional $1$ issued by the $j$--th bank within the LIBOR panel is given by
\begin{eqnarray}\label{unsecure}
\begin{aligned}
B^j(t,T)&=\mathbb{E}^\mathbb{Q}\left[ e^{-\int_t^{T} r(s) ds} \mathbb{I}_{\tau^j>T}\mid \mathcal{F}_t\right]\\
&=\mathbb{E}^\mathbb{Q}\left[ e^{-\int_t^{T}r(s)+\hat{\lambda}_j(s) ds}\mid \mathcal{H}_t \right]\\
&=\mathbb{E}^\mathbb{Q}\left[ e^{-\int_t^{T}r_c(s)-(q\Lambda-\hat{\lambda}_j(s)) ds}\mid \mathcal{H}_t \right]\\
&=D^{\tiny \mbox{OIS}}(t,T)\mathbb{E}^{\mathbb{Q}^T}\left[ e^{-\int_t^{T}(\hat{\lambda}_j(s)-q\Lambda) ds}\mid \mathcal{H}_t \right].
\end{aligned}
\end{eqnarray}
In Appendix \ref{CDS} we develop the previous expression in terms of our affine specification.

Let us now consider the pricing of a  \emph{Credit Default Swap}, i.e.  a contractual agreement between the protection buyer and the protection seller, typically designed to provide protection from a credit event associated with a risky bond issued by the reference entity\footnote{We drop the subscript corresponding to the  reference entity  bank $j$ in order simplify the notation, when there is no ambiguity.}. CDS represent the most liquid credit derivative contracts by far.

\begin{itemize}
\item The \emph{protection buyer} pays rate $\mathcal{C}$ (the CDS spread) at times $T_{1}, \cdots,T_N=T$.
\item The \emph{protection seller} agrees to make a single protection payment called \emph{Loss Given Default} (LGD) (i.e. $L=1-\mathcal{R}$, where $\mathcal{R}$ is assumed to be a fixed cash recovery) in case the default event happens between $T_{0}$ and $T$.
\end{itemize}
Let $(T_k)_{k=1,\cdots, N}$ be the CDS spread payment dates and $\tau$ the default time. The  CDS discounted payoff from the perspective of a protection buyer, with unit notional and protection payment $1-\mathcal{R}$, is at time $t\leq T_1$

{\normalsize
\begin{align}\label{Payoff}
\mbox{CDS}(t,T)=&\underbrace{\mathbb{E}^\mathbb{Q}\left[e^{-\int_t^\tau r_c(s)ds }(1-\mathcal{R})\mathbb{I}_{\{\tau\leq T\}}\mid \mathcal{F}_t\right]}_{(1)}\\
&-\underbrace{\mathbb{E}^\mathbb{Q}\left[\sum_{k=1}^N e^{-\int_t^{T_k} r_c(s)ds }(T_k-T_{k-1})\mathcal{C}\mathbb{I}_{\{\tau> T_k\}}\mid \mathcal{F}_t\right]}_{(2)}\\ \nonumber
&-\underbrace{\mathbb{E}^\mathbb{Q}\left[\sum_{k=1}^N e^{-\int_t^{\tau} r_c(s)ds }(\tau-T_{k-1})\mathcal{C}\mathbb{I}_{\{\tau\in [T_{k-1},T_k]\}}\mid \mathcal{F}_t\right]}_{(3)},
\end{align}}
where:
\begin{itemize}
\item[(1)] Protection leg payment,  price of any proceeds from default before $T$;
\item[(2)] No-default spread payments;
\item[(3)] Payments of accrued spread interest at default.
\end{itemize}
We proceed to evaluate each of these components.
\begin{align}\label{Premium}
\mathbb{E}^\mathbb{Q}\left[e^{-\int_t^{T} r_c(s)ds }\mathbb{I}_{\{\tau> T\}}\mid \mathcal{F}_t\right]&= D^{\tiny \mbox{OIS}}(t,T)\mathbb{E}^{\mathbb{Q}^{T}}\left[\mathbb{I}_{\{\tau> T\}}\mid \mathcal{F}_t\right]\\ \nonumber
&=D^{\tiny \mbox{OIS}}(t,T)\mathbb{E}^{\mathbb{Q}^{T}}\left[e^{-\int_t^T \hat{\lambda}_j(s)ds}\mid \mathcal{F}_t\right]\\ \nonumber
&:=D^{\tiny \mbox{OIS}}(t,T)Z(t,T).
\end{align}
Here the dynamics of $\hat{\lambda}_j(s)$ have to be adjusted in accordance with $T$-forward measure.
It is easy to see that
\begin{align}
\frac{dZ(t,u)}{du}=-\mathbb{E}^{\mathbb{Q}^{T}}\left[e^{-\int_t^u \hat{\lambda}_j(s)ds}\hat{\lambda}_j(u)\mid \mathcal{F}_t\right],
\end{align}
where the change of order of differentiation and expectation is justified by the dominated convergence theorem.

Consider a mesh $\{t_p\}$ on the interval $[T_{k-1}, T_k]$,
\begin{align}\label{Interest}
&\mathbb{E}^\mathbb{Q}\left[\sum_{k=1}^N e^{-\int_t^{\tau} r_c(s)ds } \mathcal{C}(\tau-T_{k-1})\mathbb{I}_{\{\tau\in [T_{k-1},T_k]\}}\mid \mathcal{F}_t\right]\\
&=\mathbb{E}^\mathbb{Q}\left[\sum_{k=1}^N e^{-\int_t^{\tau} r_c(s)ds } \mathcal{C}(\tau-T_{k-1})\mathbb{I}_{\{\tau\in [T_{k-1},T_k]\}}\sum_p \mathbb{I}_{\{\tau\in [t_p,t_{p+1}]\}}\mid \mathcal{F}_t\right]\nonumber\\ \nonumber
&=\sum_{k=1}^N \mathcal{C}\sum_p D^{\tiny \mbox{OIS}}(t,t_p)\mathbb{E}^{\mathbb{Q}^{T}}\left[(t_p-T_{k-1})\mathbb{I}_{\{\tau\in [t_p,t_{p+1}]\}}\mid \mathcal{F}_t\right]\\ \nonumber
&=\sum_{k=1}^N \mathcal{C}\sum_p D^{\tiny \mbox{OIS}}(t,t_p)(t_p-T_{k-1})\mathbb{E}^{\mathbb{Q}^{T}}\left[\mathbb{I}_{\{\tau>t_p\}}-\mathbb{I}_{\{\tau>t_{p+1}\}}\mid \mathcal{F}_t\right]\\ \nonumber
&=\sum_{k=1}^N \mathcal{C}\sum_p D^{\tiny \mbox{OIS}}(t,t_p)(t_p-T_{k-1})\left[Z(t,t_p)-Z(t,t_{p+1})\right]\\ \nonumber
&\approx \mathcal{C} \sum_{k=1}^N \int_{T_{k-1}}^{T_k} D^{\tiny \mbox{OIS}}(t,u)(u-T_{k-1})dZ(t,u)\\ \nonumber
&=-\mathcal{C} \sum_{k=1}^N \int_{T_{k-1}}^{T_k} D^{\tiny \mbox{OIS}}(t,u)(u-T_{k-1})\mathbb{E}^{\mathbb{Q}^{T}}\left[e^{-\int_t^u \hat{\lambda}_j(s)ds}\hat{\lambda}_j(u)\mid \mathcal{F}_t\right] du.
\end{align}
Finally, the protection leg can be decomposed as follows:
\begin{align}\label{Protection}
&\mathbb{E}^\mathbb{Q}\left[e^{-\int_t^{\tau} r_c(s)ds } (1-\mathcal{R})\mathbb{I}_{\{\tau<T\}}\mid \mathcal{F}_t\right]\\
&=\mathbb{E}^\mathbb{Q}\left[e^{-\int_t^{\tau} r_c(s)ds }(\mathbb{I}_{\{\tau<T]\}}(1-\mathcal{R})\sum_p \mathbb{I}_{\{\tau\in [t_p,t_{p+1}]\}}\mid \mathcal{F}_t\right]\nonumber\\ \nonumber
&=(1-\mathcal{R})\sum_p D^{\tiny \mbox{OIS}}(t,t_p)\mathbb{E}^{\mathbb{Q}^{T}}\left[\mathbb{I}_{\{\tau\in [t_p,t_{p+1}]\}}\mid \mathcal{F}_t\right]\\ \nonumber
&=(1-\mathcal{R})\sum_p D^{\tiny \mbox{OIS}}(t,t_p)\mathbb{E}^{\mathbb{Q}^{T}}\left[\mathbb{I}_{\{\tau>t_p\}}-\mathbb{I}_{\{\tau>t_{p+1}\}}\mid \mathcal{F}_t\right]\\ \nonumber
&=(1-\mathcal{R})\sum_p D^{\tiny \mbox{OIS}}(t,t_p)\left[Z(t,t_p)-Z(t,t_{p+1})\right]\\ \nonumber
&\approx (1-\mathcal{R})\int_{t}^{T} D^{\tiny \mbox{OIS}}(t,u) dZ(t,u)\\ \nonumber
&= -(1-\mathcal{R})\int_{t}^{T} D^{\tiny \mbox{OIS}}(t,u)\mathbb{E}^{\mathbb{Q}^{T}}\left[e^{-\int_t^u \hat{\lambda}_j(s)ds}\hat{\lambda}_j(u)\mid \mathcal{F}_t\right]du.
\end{align}
Putting it all together, the expression of the price of a CDS written on the $j$--th bank in the panel can now be written as
\begin{align}
\mbox{CDS}(t,T)=&-(1-\mathcal{R})\int_{t}^{T} D^{\tiny \mbox{OIS}}(t,u)\mathbb{E}^{\mathbb{Q}^{T}}\left[e^{-\int_t^u \hat{\lambda}_j(s)ds}\hat{\lambda}_j(u)\mid \mathcal{F}_t\right]du\label{CDScond}\\ \nonumber
&-\mathcal{C}\sum_{k=1}^N  (T_k-T_{k-1})D^{\tiny \mbox{OIS}}(t,T_k)Z(t,T_k)\\
&-\mathcal{C}\sum_{k=1}^N\int_{T_{k-1}}^{T_k} D^{\tiny \mbox{OIS}}(t,u)(u-T_{k-1})\mathbb{E}^{\mathbb{Q}^{T}}\left[e^{-\int_t^u \hat{\lambda}_j(s)ds}\hat{\lambda}_j(u)\mid \mathcal{F}_t\right]du.\nonumber
\end{align}
The market--quoted CDS spread $s_t$ at time $t$  is the value of $\mathcal{C}$ such that $\mbox{CDS}(t,T)=0$. Thus the CDS calibration problem is to determine the coefficients of the dynamics of $\hat\lambda_j$ such that the RHS of (\ref{CDScond}) equals zero when the market--quoted CDS spread for the $j$--th bank is substituted for $\mathcal{C}$, for all maturities $T$ for which market--quoted CDS spreads for the $j$--th bank are available.

\begin{table}[t]
\begin{center}
\begin{tabular}{|c|c|}
\toprule
\textbf{TICKER}	&	\textbf{LONGNAME}	\\
\midrule
\hline

BACORP	&	Bank of America Corporation	\\
MUFJ-BTMUFJ	&	The Bank of Tokyo--Mitsubishi UFJ, Ltd.	\\
BACR-Bank	&	BARCLAYS BANK PLC	\\
C	&	Citigroup Inc.	\\
ACAFP	&	CREDIT AGRICOLE SA	\\
CSGAG	&	Credit Suisse Group AG	\\
DB	&	DEUTSCHE BANK AKTIENGESELLSCHAFT	\\
HSBC	&	HSBC HOLDINGS plc	\\
JPM	&	JPMorgan Chase \& Co.	\\
LBGP	&	LLOYDS BANKING GROUP PLC	\\
COOERAB	&	Cooeperatieve Rabobank U.A.	\\
RY	&	Royal Bank of Canada	\\
SOCGEN	&	SOCIETE GENERALE	\\
SUMIBK-Bank	&	Sumitomo Mitsui Banking Corporation	\\
NORBK	&	The Norinchukin Bank	\\
RBOS-RBOSplc	&	The Royal Bank of Scotland public limited company	\\
UBS	&	UBS AG	\\
\bottomrule
\end{tabular}
\end{center}
\caption{LIBOR panel banks}\label{panel}
\end{table}
\subsection{Calibration procedure and numerical results}
The calibration procedure outlined above thus proceeds in three steps:
\begin{enumerate}
\item Calibrate the coefficients $\alpha_0(t), \alpha$ of the dynamics of $r_c$ to the term structure of discount factors implied by OIS. This is identical to the first step of the calibration described in Section \ref{OISswapcal}, so we do not present the results of this step here.
\item Calibrate the coefficients $\hat b_0^j(t), \hat b^j$ of the dynamics of the idiosyncratic default intensity $\hat\lambda_j$ to the term structure of CDS spreads of bank $j$. Taking the average over all $j$ of each coefficient in $\hat b_0^j(t), \hat b^j$ and subtracting $\Lambda=5$ bp from $\hat b_0(t)$ gives the coefficients $b_0(t), b$ of the dynamics of $\lambda$.
\item Calibrate the coefficients $c_0(t), c$ of the dynamics of the ``funding liquidity'' component $\phi$ of roll--over risk simultaneously to all available vanilla and basis swaps.
\end{enumerate}
For the reasons stated in Section \ref{OISfit}, we restrict ourselves to maturities up to ten years. We report the calibration results for six exemplary dates, 1 January 2013, 8 September 2014, 18 June 2015, 20 April 2016, 22 March 2017 and 31 October 2017. Results for other days in our data set are qualitatively similar. For IRS, OIS and basis swaps here we use the same instruments as in Section \ref{calsec}. The USD LIBOR panel is made up of the 17 banks listed in Table \ref{panel}. Where CDS quotes (sourced from Markit) were not available on a given day for a given panel bank (e.g. for Lloyds Banking Group, Rabobank and Royal Bank of Canada), we have dropped these from our calibration. For the remaining panel banks, we used all available maturities in the CDS calibration, i.e. 0.5, 1, 2, 3, 4, 5, 7 and 10 years.

The calibrated coefficients $b$ for each day, in the one--factor and in the three--factor model, are given in Tables \ref{CDScoeff}--\ref{CDScoeff1}. Both the one--factor and three--factor model fit the Markit quotes to well within one basis point, as Tables \ref{CDSfit130101} to \ref{CDSfit171031} show.

Moving to Step 3 in the calibration procedure, as in Section \ref{OISswapcal} we plot the bid and ask values for each tenor relative to the value produced by the calibrated model. The results in this step differ from those reported in Section \ref{OISswapcal}, because now part of the roll--over risk (the part which a borrower faces due to the possibility of having to pay a higher credit spread than the LIBOR panel average in the future) is already determined by the credit spread dynamics calibrated in Step 2. Nevertheless, the one--factor model still fits the market reasonably well, as Figure \ref{CDSbasis1factor} shows --- and this could be improved further by a staged fit of a three--factor model in the same manner as in Section \ref{calsec}.

Thus we see that the model, even in its one--factor version, can be calibrated simultaneously to market data on a given day for OIS, CDS, interest rate and basis swaps. This is achieved in the usual fashion\footnote{This approach was pioneered in interest--rate term structure modelling by \citeasnoun{hull_pricing_1990}, and developed in the fashion used here by \citeasnoun{brigo_deterministic_1998} (it is called ``the CIR++ Model'' in \citeasnoun{BM2006}).} of fitting interest rate term structure models to the market, using ``term structures'' of time--dependent coefficients $\alpha_0(t)$, $b_0(t)$ and $c_0(t)$. In this sense, our approach extends the interest rate term structure modelling framework to roll--over risk in a relatively straightforward and consistent fashion. However, when calibrating simultaneously to all terms (1--month, 3--month and 6--month), the calibrated $c_0(t)$ are not very smooth, see Figure \ref{c0plot}, which may be due to market frictions not captured by the model, or an indication that the market across tenors has yet to mature fully.\footnote{In a simpler model of roll--over risk, \citeasnoun{ChaSch:2015} show how the ability of that model to fit the market for basis swaps has improved in the years since the tenor basis became an economically significant phenomenon as a result of the financial crisis of 2007/8.}

\begin{table}[t]

  \begin{subtable}{.5\textwidth}
  \centering
  \caption{One-factor CIR}
  \scalebox{0.45}{
\begin{tabular}{|l|rrrrrrrrrrrrrr|r|}
    \toprule
    \multicolumn{1}{|r|}{1/01/2013} & \multicolumn{1}{l}{BACORP} & \multicolumn{1}{l}{MUFJ} & \multicolumn{1}{l}{BACR} & \multicolumn{1}{l}{C} & \multicolumn{1}{l}{ACAFP} & \multicolumn{1}{l}{CSGAG} & \multicolumn{1}{l}{DB} & \multicolumn{1}{l}{HSBC} & \multicolumn{1}{l}{JPM} & \multicolumn{1}{l}{SOCGEN} & \multicolumn{1}{l}{SUMIBK-Bank} & \multicolumn{1}{l}{NORBK} & \multicolumn{1}{l}{RBOS-RBOSplc} & \multicolumn{1}{l|}{UBS} & \multicolumn{1}{l|}{\textbf{Average}} \\
    \midrule
    $\hat{b}$     & 0.001279 & 0.002276 & 0.001897 & 0.000931 & 0.001208 & 0.000939 & 0.000342 & 0.002265 & 0.000793 & 0.001322 & 0.000408 & 0.000962 & 0.000146 & 0.000038 & 0.001058 \\
    $\hat{b}_0(1)$   & 0.000451 & 0.000394 & 0.000264 & 0.000561 & 0.000204 & 0.000105 & 0.000248 & 0.000235 & 0.000450 & 0.000596 & 0.001114 & 0.000518 & 0.000371 & 0.000385 & 0.000421 \\
    $\hat{b}_0(2)$   & 0.000785 & 0.000539 & 0.000453 & 0.000984 & 0.000522 & 0.000140 & 0.000322 & 0.000399 & 0.000601 & 0.000980 & 0.001407 & 0.000835 & 0.000682 & 0.000428 & 0.000648 \\
    $\hat{b}_0(3)$   & 0.001971 & 0.001188 & 0.001480 & 0.002050 & 0.001969 & 0.000678 & 0.000996 & 0.000981 & 0.001281 & 0.002649 & 0.001977 & 0.001345 & 0.002040 & 0.001048 & 0.001547 \\
    $\hat{b}_0(4)$   & 0.003457 & 0.001764 & 0.002678 & 0.003368 & 0.003582 & 0.001676 & 0.001960 & 0.001544 & 0.002183 & 0.004439 & 0.002710 & 0.001918 & 0.003829 & 0.001930 & 0.002646 \\
    $\hat{b}_0(5)$   & 0.004546 & 0.002563 & 0.003943 & 0.004610 & 0.005246 & 0.002705 & 0.003067 & 0.002228 & 0.002931 & 0.006065 & 0.003364 & 0.002576 & 0.005546 & 0.003061 & 0.003746 \\
    $\hat{b}_0(6)$   & 0.005693 & 0.003408 & 0.005556 & 0.005658 & 0.007040 & 0.003827 & 0.004279 & 0.003001 & 0.003742 & 0.008115 & 0.004086 & 0.003254 & 0.007450 & 0.004198 & 0.004951 \\
\bottomrule
    \end{tabular}%
    }
  \label{tab:addlabel}%
  \end{subtable}

  \begin{subtable}{.5\textwidth}
  \centering
  \caption{Three-factor CIR}
    \scalebox{0.45}{
    \begin{tabular}{|l|rrrrrrrrrrrrrr|r|}
    \toprule
    \multicolumn{1}{|r|}{1/01/2013} & \multicolumn{1}{l}{BACORP} & \multicolumn{1}{l}{MUFJ} & \multicolumn{1}{l}{BACR} & \multicolumn{1}{l}{C} & \multicolumn{1}{l}{ACAFP} & \multicolumn{1}{l}{CSGAG} & \multicolumn{1}{l}{DB} & \multicolumn{1}{l}{HSBC} & \multicolumn{1}{l}{JPM} & \multicolumn{1}{l}{SOCGEN} & \multicolumn{1}{l}{SUMIBK-Bank} & \multicolumn{1}{l}{NORBK} & \multicolumn{1}{l}{RBOS-RBOSplc} & \multicolumn{1}{l|}{UBS} & \multicolumn{1}{l|}{\textbf{Average}} \\
    \midrule
    $\hat{b}_1$    & 0.000256 & 0.001455 & 0.000242 & 0.000047 & 0.000185 & 0.000083 & 0.000168 & 0.000571 & 0.000242 & 0.000448 & 0.000001 & 0.000726 & 0.000310 & 0.000094 & 0.000345 \\
  $\hat{b}_2$    & 0.000067 & 0.000218 & 0.000179 & 0.000252 & 0.000295 & 0.000219 & 0.000012 & 0.000011 & 0.000003 & 0.000035 & 0.000246 & 0.000054 & 0.000317 & 0.000066 & 0.000141 \\
   $\hat{b}_3$     & 0.000681 & 0.000178 & 0.001339 & 0.000303 & 0.000801 & 0.000431 & 0.000635 & 0.000091 & 0.000115 & 0.000977 & 0.002151 & 0.000331 & 0.000175 & 0.000313 & 0.000609 \\
  $\hat{b}_0(1)$   & 0.000498 & 0.000354 & 0.000224 & 0.000539 & 0.000043 & 0.000061 & 0.000106 & 0.000669 & 0.000567 & 0.000532 & 0.000446 & 0.000422 & 0.000015 & 0.000226 & 0.000336 \\
   $\hat{b}_0(2)$   & 0.000820 & 0.000524 & 0.000391 & 0.000970 & 0.000358 & 0.000096 & 0.000168 & 0.000834 & 0.000712 & 0.000902 & 0.000697 & 0.000751 & 0.000346 & 0.000268 & 0.000560 \\
   $\hat{b}_0(3)$   & 0.001980 & 0.001214 & 0.001372 & 0.002029 & 0.001788 & 0.000631 & 0.000811 & 0.001418 & 0.001395 & 0.002529 & 0.001197 & 0.001252 & 0.001724 & 0.000880 & 0.001444 \\
   $\hat{b}_0(4)$   & 0.003448 & 0.001812 & 0.002537 & 0.003352 & 0.003390 & 0.001622 & 0.001762 & 0.001984 & 0.002281 & 0.004299 & 0.001865 & 0.001831 & 0.003542 & 0.001755 & 0.002534 \\
\bottomrule
    \end{tabular}%
    }
  \end{subtable}
   \caption{Model calibration to CDS data for USD LIBOR panel banks on 01/01/2013}\label{CDScoeff}\vspace*{1ex}

  \begin{subtable}{.5\textwidth}
  \centering
  \caption{One-factor CIR}
  \scalebox{0.45}{
    \begin{tabular}{|l|rrrrrrrrrrrrrr|r|}
    \toprule
    \multicolumn{1}{|r|}{1/01/2013} & \multicolumn{1}{l}{BACORP} & \multicolumn{1}{l}{MUFJ} & \multicolumn{1}{l}{BACR} & \multicolumn{1}{l}{C} & \multicolumn{1}{l}{ACAFP} & \multicolumn{1}{l}{CSGAG} & \multicolumn{1}{l}{DB} & \multicolumn{1}{l}{HSBC} & \multicolumn{1}{l}{JPM} & \multicolumn{1}{l}{SOCGEN} & \multicolumn{1}{l}{SUMIBK-Bank} & \multicolumn{1}{l}{NORBK} & \multicolumn{1}{l}{RBOS-RBOSplc} & \multicolumn{1}{l|}{UBS} & \multicolumn{1}{l|}{\textbf{Average}} \\
    \midrule

    $\hat{b}$     & 0.008792 & 0.002785 & 0.001039 & 0.009392 & 0.013380 & 0.006849 & 0.015682 & 0.023805 & 0.010712 & 0.013942 & 0.006820 & 0.008075 & 0.008763 & 0.003894 & 0.009566 \\
     $\hat{b}_0(1)$   & 0.001035 & 0.000538 & 0.000715 & 0.000901 & 0.000230 & 0.000350 & 0.000274 & 0.000300 & 0.000561 & 0.000169 & 0.000484 & 0.000168 & 0.001169 & 0.000326 & 0.000516 \\
     $\hat{b}_0(2)$  & 0.001246 & 0.000695 & 0.000931 & 0.001247 & 0.000295 & 0.000473 & 0.000525 & 0.000331 & 0.001044 & 0.000372 & 0.000595 & 0.000343 & 0.001516 & 0.000440 & 0.000718 \\
    $\hat{b}_0(3)$  & 0.001770 & 0.001030 & 0.001462 & 0.001779 & 0.000546 & 0.000750 & 0.000970 & 0.000568 & 0.001417 & 0.000951 & 0.000895 & 0.000757 & 0.002263 & 0.000834 & 0.001142 \\
    $\hat{b}_0(4)$  & 0.002295 & 0.001453 & 0.002086 & 0.002333 & 0.000813 & 0.001060 & 0.001406 & 0.000840 & 0.001734 & 0.001560 & 0.001252 & 0.001181 & 0.003087 & 0.001233 & 0.001595 \\
    $\hat{b}_0(5)$   & 0.003005 & 0.002034 & 0.002802 & 0.003080 & 0.001156 & 0.001376 & 0.002038 & 0.001242 & 0.002225 & 0.002293 & 0.001785 & 0.001779 & 0.003974 & 0.001728 & 0.002180 \\
    $\hat{b}_0(6)$   & 0.003725 & 0.002668 & 0.003590 & 0.003752 & 0.001554 & 0.001717 & 0.002653 & 0.001670 & 0.002806 & 0.003021 & 0.002385 & 0.002308 & 0.005026 & 0.002218 & 0.002792 \\
\bottomrule
    \end{tabular}%
    }
  \label{tab:addlabel}%
  \end{subtable}

  \begin{subtable}{.5\textwidth}
  \centering
  \caption{Three-factor CIR}
    \scalebox{0.45}{
    \begin{tabular}{|l|rrrrrrrrrrrrrr|r|}
    \toprule
    \multicolumn{1}{|r|}{1/01/2013} & \multicolumn{1}{l}{BACORP} & \multicolumn{1}{l}{MUFJ} & \multicolumn{1}{l}{BACR} & \multicolumn{1}{l}{C} & \multicolumn{1}{l}{ACAFP} & \multicolumn{1}{l}{CSGAG} & \multicolumn{1}{l}{DB} & \multicolumn{1}{l}{HSBC} & \multicolumn{1}{l}{JPM} & \multicolumn{1}{l}{SOCGEN} & \multicolumn{1}{l}{SUMIBK-Bank} & \multicolumn{1}{l}{NORBK} & \multicolumn{1}{l}{RBOS-RBOSplc} & \multicolumn{1}{l|}{UBS} & \multicolumn{1}{l|}{\textbf{Average}} \\
    \midrule

   $\hat{b}_1$  & 0.000139 & 0.000219 & 0.000003 & 0.000047 & 0.000026 & 0.000017 & 0.000229 & 0.000026 & 0.000092 & 0.000243 & 0.000146 & 0.000005 & 0.000422 & 0.000047 & 0.000119 \\
    $\hat{b}_2$    & 0.000424 & 0.000046 & 0.000491 & 0.000631 & 0.000066 & 0.000139 & 0.000097 & 0.000018 & 0.000014 & 0.000004 & 0.000004 & 0.000022 & 0.000188 & 0.000059 & 0.000157 \\
    $\hat{b}_3$    & 0.003722 & 0.001874 & 0.001572 & 0.000025 & 0.000397 & 0.000283 & 0.000939 & 0.000075 & 0.000671 & 0.001618 & 0.000147 & 0.000115 & 0.001625 & 0.000925 & 0.000999 \\
     $\hat{b}_0(1)$  & 0.000424 & 0.000257 & 0.000156 & 0.000528 & 0.000417 & 0.000344 & 0.000318 & 0.000793 & 0.000671 & 0.000182 & 0.000537 & 0.000549 & 0.000806 & 0.000246 & 0.000445 \\
    $\hat{b}_0(2)$  & 0.000516 & 0.000337 & 0.000305 & 0.000918 & 0.000537 & 0.000489 & 0.000602 & 0.000955 & 0.001184 & 0.000381 & 0.000665 & 0.000697 & 0.001105 & 0.000338 & 0.000645 \\
     $\hat{b}_0(3)$   & 0.000904 & 0.000568 & 0.000747 & 0.001537 & 0.000911 & 0.000829 & 0.001140 & 0.001445 & 0.001624 & 0.000993 & 0.001027 & 0.001019 & 0.001807 & 0.000718 & 0.001091 \\
     $\hat{b}_0(4)$   & 0.001347 & 0.000934 & 0.001316 & 0.002183 & 0.001294 & 0.001194 & 0.001689 & 0.001967 & 0.002019 & 0.001662 & 0.001441 & 0.001385 & 0.002643 & 0.001103 & 0.001584 \\
\bottomrule
    \end{tabular}%
    }
  \end{subtable}
   \caption{Model calibration to CDS data for USD LIBOR panel banks on 08/09/2014}\label{CDScoeffa}\vspace*{1ex}

  \begin{subtable}{.5\textwidth}
  \centering
  \caption{One-factor CIR}
  \scalebox{0.42}{
    \begin{tabular}{|l|rrrrrrrrrrrrrrr|r|}
    \toprule
    \multicolumn{1}{|c|}{18/06/2015} & \multicolumn{1}{c}{BACORP} & \multicolumn{1}{c}{MUFJ} & \multicolumn{1}{c}{BACR} & \multicolumn{1}{c}{C} & \multicolumn{1}{c}{ACAFP} & \multicolumn{1}{c}{CSGAG} & \multicolumn{1}{c}{DB} & \multicolumn{1}{c}{HSBC} & \multicolumn{1}{c}{JPM} & \multicolumn{1}{c}{RY} & \multicolumn{1}{c}{SOCGEN} & \multicolumn{1}{c}{SUMIBK-Bank} & \multicolumn{1}{c}{NORBK} & \multicolumn{1}{c}{RBOS-RBOSplc} & \multicolumn{1}{c|}{UBS} & \multicolumn{1}{l|}{\textbf{Average}} \\
    \midrule

   $\hat{b}$     & 0.002689 & 0.002462 & 0.006484 & 0.007259 & 0.004160 & 0.000311 & 0.000344 & 0.002964 & 0.003728 & 0.008462 & 0.001090 & 0.004187 & 0.003226 & 0.004689 & 0.001667 & 0.003582 \\
    $\hat{b}_0(1)$   & 0.000599 & 0.000399 & 0.001181 & 0.000340 & 0.001414 & 0.001830 & 0.002049 & 0.001224 & 0.000480 & 0.000134 & 0.001655 & 0.000146 & 0.000181 & 0.001456 & 0.001279 & 0.000958 \\
    $\hat{b}_0(2)$   & 0.000776 & 0.000423 & 0.001253 & 0.000479 & 0.001452 & 0.002063 & 0.002338 & 0.001334 & 0.000662 & 0.000261 & 0.001855 & 0.000145 & 0.000255 & 0.001566 & 0.001427 & 0.001086 \\
    $\hat{b}_0(3)$   & 0.001065 & 0.000616 & 0.001374 & 0.000695 & 0.001741 & 0.002412 & 0.002729 & 0.001550 & 0.000895 & 0.000318 & 0.002333 & 0.000269 & 0.000360 & 0.001794 & 0.001715 & 0.001324 \\
    $\hat{b}_0(4)$   & 0.001503 & 0.000827 & 0.001592 & 0.001139 & 0.002106 & 0.002804 & 0.003171 & 0.001846 & 0.001157 & 0.000295 & 0.002860 & 0.000487 & 0.000621 & 0.002083 & 0.002029 & 0.001635 \\
    $\hat{b}_0(5)$   & 0.001954 & 0.001235 & 0.001973 & 0.001608 & 0.002559 & 0.003213 & 0.003698 & 0.002269 & 0.001598 & 0.000663 & 0.003460 & 0.000930 & 0.000950 & 0.002500 & 0.002359 & 0.002065 \\
    $\hat{b}_0(6)$   & 0.002611 & 0.001646 & 0.002425 & 0.002293 & 0.003122 & 0.003717 & 0.004315 & 0.002737 & 0.002245 & 0.001018 & 0.004131 & 0.001395 & 0.001505 & 0.003013 & 0.002744 & 0.002595 \\
\bottomrule
    \end{tabular}%
    }
  \label{tab:addlabel}%
  \end{subtable}

  \begin{subtable}{.5\textwidth}

  \centering
  \caption{Three-factor CIR}
    \scalebox{0.42}{

    \begin{tabular}{|l|rrrrrrrrrrrrrrr|r|}
    \toprule
    \multicolumn{1}{|c|}{18/06/2015} & \multicolumn{1}{c}{BACORP} & \multicolumn{1}{c}{MUFJ} & \multicolumn{1}{c}{BACR} & \multicolumn{1}{c}{C} & \multicolumn{1}{c}{ACAFP} & \multicolumn{1}{c}{CSGAG} & \multicolumn{1}{c}{DB} & \multicolumn{1}{c}{HSBC} & \multicolumn{1}{c}{JPM} & \multicolumn{1}{c}{RY} & \multicolumn{1}{c}{SOCGEN} & \multicolumn{1}{c}{SUMIBK-Bank} & \multicolumn{1}{c}{NORBK} & \multicolumn{1}{c}{RBOS-RBOSplc} & \multicolumn{1}{c|}{UBS} & \multicolumn{1}{l|}{\textbf{Average}} \\
    \midrule
   $ \hat{b}_1$    & 0.000394 & 0.000404 & 0.000054 & 0.000636 & 0.000338 & 0.000874 & 0.000244 & 0.000470 & 0.000015 & 0.000357 & 0.000087 & 0.000043 & 0.000146 & 0.000171 & 0.000171 & 0.000294 \\
    $\hat{b}_2$     & 0.001152 & 0.000430 & 0.001729 & 0.001102 & 0.000793 & 0.001980 & 0.001974 & 0.000116 & 0.001333 & 0.001424 & 0.000994 & 0.001338 & 0.000162 & 0.000160 & 0.000160 & 0.000990 \\
    $\hat{b}_3$     & 0.003551 & 0.001544 & 0.003875 & 0.004396 & 0.006599 & 0.007963 & 0.002503 & 0.005821 & 0.003205 & 0.003017 & 0.001845 & 0.001316 & 0.001262 & 0.005766 & 0.005766 & 0.003895 \\
    $\hat{b}_0(1)$   & 0.000212 & 0.000179 & 0.001243 & 0.000078 & 0.000969 & 0.000571 & 0.001609 & 0.000700 & 0.000433 & 0.000228 & 0.001482 & 0.000255 & 0.000227 & 0.001261 & 0.001261 & 0.000714 \\
    $\hat{b}_0(2)$   & 0.000437 & 0.000268 & 0.001413 & 0.000383 & 0.001097 & 0.000787 & 0.001848 & 0.000900 & 0.000662 & 0.000526 & 0.001679 & 0.000310 & 0.000378 & 0.001484 & 0.001484 & 0.000910 \\
    $\hat{b}_0(3)$   & 0.000772 & 0.000548 & 0.001666 & 0.000834 & 0.001512 & 0.001108 & 0.002166 & 0.001256 & 0.000956 & 0.000825 & 0.002150 & 0.000507 & 0.000592 & 0.001881 & 0.001881 & 0.001243 \\
    $\hat{b}_0(4)$   & 0.001265 & 0.000817 & 0.001957 & 0.001425 & 0.001954 & 0.001478 & 0.002550 & 0.001634 & 0.001247 & 0.000942 & 0.002666 & 0.000765 & 0.000929 & 0.002275 & 0.002275 & 0.001612 \\
\bottomrule
    \end{tabular}%
    }
  \end{subtable}
   \caption{Model calibration to CDS data for USD LIBOR panel banks on 18/06/2015}
\end{table}%

\begin{table}[t]

  \begin{subtable}{.5\textwidth}
  \centering
  \caption{One-factor CIR}
  \scalebox{0.40}{
  \begin{tabular}{|l|rrrrrrrrrrrrrrrr|r|}
    \toprule
    \multicolumn{1}{|c|}{20/04/2016} & \multicolumn{1}{c}{BACORP} & \multicolumn{1}{c}{MUFJ} & \multicolumn{1}{c}{BACR} & \multicolumn{1}{c}{C} & \multicolumn{1}{c}{ACAFP} & \multicolumn{1}{c}{CSGAG} & \multicolumn{1}{c}{DB} & \multicolumn{1}{c}{HSBC} & \multicolumn{1}{c}{JPM} & \multicolumn{1}{c}{COOERAB} & \multicolumn{1}{c}{RY} & \multicolumn{1}{c}{SOCGEN} & \multicolumn{1}{c}{SUMIBK-Bank} & \multicolumn{1}{c}{NORBK} & \multicolumn{1}{c}{RBOS-RBOSplc} & \multicolumn{1}{c|}{UBS} & \multicolumn{1}{c|}{\textbf{Average}} \\
    \midrule
    $\hat{b}$     & 0.001937 & 0.002802 & 0.041052 & 0.009480 & 0.001848 & 0.007952 & 0.008248 & 0.024514 & 0.000950 & 0.002662 & 0.013762 & 0.001395 & 0.004749 & 0.008075 & 0.023204 & 0.009226 & 0.010116 \\
    $\hat{b}_0(1)$   & 0.001039 & 0.000603 & 0.000002 & 0.000514 & 0.000763 & 0.003029 & 0.003518 & 0.000595 & 0.000910 & 0.000796 & 0.000942 & 0.000765 & 0.000567 & 0.000168 & 0.001324 & 0.000462 & 0.001000 \\
    $\hat{b}_0(2)$   & 0.001396 & 0.000750 & 0.000382 & 0.000954 & 0.000920 & 0.003284 & 0.004251 & 0.000784 & 0.001142 & 0.000946 & 0.001437 & 0.000963 & 0.000662 & 0.000343 & 0.001701 & 0.000666 & 0.001286 \\
    $\hat{b}_0(3)$   & 0.001896 & 0.001204 & 0.000884 & 0.001491 & 0.001368 & 0.003681 & 0.004891 & 0.001243 & 0.001522 & 0.001284 & 0.001919 & 0.001478 & 0.001076 & 0.000757 & 0.002135 & 0.000979 & 0.001738 \\
    $\hat{b}_0(4)$   & 0.002319 & 0.001664 & 0.001331 & 0.001936 & 0.001824 & 0.004098 & 0.005497 & 0.001733 & 0.001794 & 0.001622 & 0.002323 & 0.001958 & 0.001595 & 0.001181 & 0.002563 & 0.001326 & 0.002173 \\
    $\hat{b}_0(5)$   & 0.002849 & 0.002375 & 0.001919 & 0.002482 & 0.002376 & 0.004678 & 0.006048 & 0.002255 & 0.002151 & 0.002017 & 0.003158 & 0.002501 & 0.002338 & 0.001779 & 0.003210 & 0.001732 & 0.002742 \\
\bottomrule

    \end{tabular}%
    }

  \end{subtable}

  \begin{subtable}{.5\textwidth}

  \centering
  \caption{Three-factor CIR}
    \scalebox{0.40}{

    \begin{tabular}{|l|rrrrrrrrrrrrrrrr|r|}
    \toprule
    \multicolumn{1}{|c|}{20/04/2016} & \multicolumn{1}{c}{BACORP} & \multicolumn{1}{c}{MUFJ} & \multicolumn{1}{c}{BACR} & \multicolumn{1}{c}{C} & \multicolumn{1}{c}{ACAFP} & \multicolumn{1}{c}{CSGAG} & \multicolumn{1}{c}{DB} & \multicolumn{1}{c}{HSBC} & \multicolumn{1}{c}{JPM} & \multicolumn{1}{c}{COOERAB} & \multicolumn{1}{c}{RY} & \multicolumn{1}{c}{SOCGEN} & \multicolumn{1}{c}{SUMIBK-Bank} & \multicolumn{1}{c}{NORBK} & \multicolumn{1}{c}{RBOS-RBOSplc} & \multicolumn{1}{c|}{UBS} & \multicolumn{1}{c|}{\textbf{Average}} \\
    \midrule
    $\hat{b}_1$     & 0.000182 & 0.000260 & 0.000041 & 0.001359 & 0.000930 & 0.000376 & 0.001689 & 0.000062 & 0.001972 & 0.000022 & 0.000395 & 0.000019 & 0.000027 & 0.000396 & 0.000826 & 0.000103 & 0.000541 \\
    $\hat{b}_2$     & 0.000102 & 0.000011 & 0.000027 & 0.000158 & 0.000152 & 0.000240 & 0.000439 & 0.000045 & 0.000185 & 0.000111 & 0.000116 & 0.000256 & 0.000034 & 0.000300 & 0.000001 & 0.000015 & 0.000137 \\
    $\hat{b}_3$     & 0.001660 & 0.000446 & 0.000344 & 0.002899 & 0.002952 & 0.002934 & 0.000871 & 0.002915 & 0.000414 & 0.000450 & 0.002615 & 0.000844 & 0.000757 & 0.000603 & 0.001812 & 0.000423 & 0.001434 \\
    $\hat{b}_0(1)$   & 0.000869 & 0.000669 & 0.002360 & 0.000403 & 0.000281 & 0.002938 & 0.003211 & 0.001713 & 0.000399 & 0.000803 & 0.001326 & 0.000528 & 0.000736 & 0.000239 & 0.002369 & 0.000933 & 0.001236 \\
    $\hat{b}_0(2)$   & 0.001143 & 0.000783 & 0.002694 & 0.000655 & 0.000273 & 0.003023 & 0.003830 & 0.001755 & 0.000529 & 0.000933 & 0.001685 & 0.000688 & 0.000799 & 0.000361 & 0.002621 & 0.001110 & 0.001430 \\
    $\hat{b}_0(3)$   & 0.001539 & 0.001190 & 0.003128 & 0.000937 & 0.000492 & 0.003257 & 0.004319 & 0.002021 & 0.000768 & 0.001243 & 0.001977 & 0.001154 & 0.001158 & 0.000709 & 0.002878 & 0.001382 & 0.001760 \\
\bottomrule
    \end{tabular}%
    }
  \end{subtable}
 \caption{Model calibration to CDS data for USD LIBOR panel banks on 20/04/2016}\label{CDScoeff1}\vspace*{1ex}

  \begin{subtable}{.5\textwidth}
  \centering
  \caption{One-factor CIR}
  \scalebox{0.40}{

  \begin{tabular}{|l|rrrrrrrrrrrrrrrr|r|}
    \toprule
    \multicolumn{1}{|c|}{22/03/2017} & \multicolumn{1}{c}{BACORP} & \multicolumn{1}{c}{MUFJ} & \multicolumn{1}{c}{BACR} & \multicolumn{1}{c}{C} & \multicolumn{1}{c}{ACAFP} & \multicolumn{1}{c}{CSGAG} & \multicolumn{1}{c}{DB} & \multicolumn{1}{c}{HSBC} & \multicolumn{1}{c}{JPM} & \multicolumn{1}{c}{COOERAB} & \multicolumn{1}{c}{RY} & \multicolumn{1}{c}{SOCGEN} & \multicolumn{1}{c}{SUMIBK-Bank} & \multicolumn{1}{c}{NORBK} & \multicolumn{1}{c}{RBOS-RBOSplc} & \multicolumn{1}{c|}{UBS} & \multicolumn{1}{c|}{\textbf{Average}} \\
    \midrule
     $\hat{b}$     & 0.003284 & 0.000319 & 0.000413 & 0.000399 & 0.005510 & 0.001082 & 0.002056 & 0.003603 & 0.001291 & 0.000879 & 0.000139 & 0.005769 & 0.002786 & 0.000710 & 0.002071 & 0.000237 & 0.001909 \\
     $\hat{b}_0(1)$  & 0.000119 & 0.000818 & 0.000907 & 0.001100 & 0.000956 & 0.001823 & 0.002333 & 0.000665 & 0.000623 & 0.001731 & 0.000876 & 0.001538 & 0.000011 & 0.000332 & 0.002196 & 0.002194 & 0.001139 \\
     $\hat{b}_0(2)$   & 0.000909 & 0.001020 & 0.001203 & 0.001542 & 0.001350 & 0.002289 & 0.002820 & 0.000991 & 0.001180 & 0.002146 & 0.001147 & 0.002141 & 0.000269 & 0.000585 & 0.002672 & 0.002533 & 0.001550 \\
     $\hat{b}_0(3)$  & 0.001416 & 0.001385 & 0.002060 & 0.002181 & 0.002542 & 0.003177 & 0.004640 & 0.002177 & 0.001638 & 0.003114 & 0.001252 & 0.003687 & 0.000709 & 0.001077 & 0.003755 & 0.003299 & 0.002382 \\
    $\hat{b}_0(4)$   & 0.002051 & 0.001738 & 0.002937 & 0.002912 & 0.003969 & 0.004095 & 0.006886 & 0.003286 & 0.002044 & 0.004103 & 0.001729 & 0.005078 & 0.001148 & 0.001606 & 0.005025 & 0.004402 & 0.003313 \\
     $\hat{b}_0(5)$  & 0.002658 & 0.002431 & 0.003823 & 0.003759 & 0.005243 & 0.004989 & 0.009616 & 0.004561 & 0.002625 & 0.005061 & 0.002553 & 0.006780 & 0.001938 & 0.002221 & 0.006556 & 0.005333 & 0.004384 \\
\bottomrule
    \end{tabular}%

    }

  \end{subtable}

  \begin{subtable}{.5\textwidth}

  \centering
  \caption{Three-factor CIR}
    \scalebox{0.40}{

    \begin{tabular}{|l|rrrrrrrrrrrrrrrr|r|}
    \toprule
    \multicolumn{1}{|c|}{22/03/2017} & \multicolumn{1}{c}{BACORP} & \multicolumn{1}{c}{MUFJ} & \multicolumn{1}{c}{BACR} & \multicolumn{1}{c}{C} & \multicolumn{1}{c}{ACAFP} & \multicolumn{1}{c}{CSGAG} & \multicolumn{1}{c}{DB} & \multicolumn{1}{c}{HSBC} & \multicolumn{1}{c}{JPM} & \multicolumn{1}{c}{COOERAB} & \multicolumn{1}{c}{RY} & \multicolumn{1}{c}{SOCGEN} & \multicolumn{1}{c}{SUMIBK-Bank} & \multicolumn{1}{c}{NORBK} & \multicolumn{1}{c}{RBOS-RBOSplc} & \multicolumn{1}{c|}{UBS} & \multicolumn{1}{c|}{\textbf{Average}} \\
    \midrule
     $\hat{b}_1$   & 0.000080 & 0.001565 & 0.000151 & 0.000851 & 0.002635 & 0.003525 & 0.006442 & 0.002449 & 0.001497 & 0.000150 & 0.000066 & 0.001210 & 0.002762 & 0.000268 & 0.004113 & 0.000302 & 0.001754 \\
     $\hat{b}_2$    & 0.001270 & 0.000169 & 0.001799 & 0.000705 & 0.000992 & 0.000509 & 0.001308 & 0.004017 & 0.000238 & 0.000918 & 0.000543 & 0.002878 & 0.000081 & 0.000083 & 0.000163 & 0.000093 & 0.000985 \\
    $\hat{b}_3$     & 0.000436 & 0.004939 & 0.002960 & 0.004991 & 0.010522 & 0.004465 & 0.010407 & 0.000646 & 0.009320 & 0.013901 & 0.008460 & 0.022052 & 0.001210 & 0.000366 & 0.013104 & 0.000718 & 0.006781 \\
    $\hat{b}_0(1)$  & 0.000630 & 0.000502 & 0.000155 & 0.000671 & 0.001701 & 0.001367 & 0.001288 & 0.000938 & 0.000487 & 0.001213 & 0.000439 & 0.001416 & 0.000495 & 0.000486 & 0.001897 & 0.002173 & 0.000991 \\
     $\hat{b}_0(2)$   & 0.001397 & 0.000616 & 0.000449 & 0.001066 & 0.001908 & 0.001658 & 0.001455 & 0.002075 & 0.000941 & 0.001572 & 0.000683 & 0.001865 & 0.000609 & 0.000725 & 0.002129 & 0.002493 & 0.001353 \\
    $\hat{b}_0(3)$  & 0.001870 & 0.000859 & 0.001299 & 0.001622 & 0.002823 & 0.002294 & 0.002837 & 0.003147 & 0.001241 & 0.002450 & 0.000737 & 0.003175 & 0.000843 & 0.001184 & 0.002868 & 0.003236 & 0.002030 \\
\bottomrule
    \end{tabular}%

    }
  \end{subtable}
 \caption{Model calibration to CDS data for USD LIBOR panel banks on 22/03/2017}\label{CDScoeff1}\vspace*{1ex}

  \begin{subtable}{.5\textwidth}
  \centering
  \caption{One-factor CIR}
  \scalebox{0.40}{

 \begin{tabular}{|l|rrrrrrrrrrrrrrrrr|r|}
    \toprule
    \multicolumn{1}{|c|}{31/10/2017} & \multicolumn{1}{c}{BACORP} & \multicolumn{1}{c}{MUFJ} & \multicolumn{1}{c}{BACR} & \multicolumn{1}{c}{C} & \multicolumn{1}{c}{ACAFP} & \multicolumn{1}{c}{CSGAG} & \multicolumn{1}{c}{DB} & \multicolumn{1}{c}{HSBC} & \multicolumn{1}{c}{JPM} & \multicolumn{1}{c}{LBGP} & \multicolumn{1}{c}{COOERAB} & \multicolumn{1}{c}{RY} & \multicolumn{1}{c}{SOCGEN} & \multicolumn{1}{c}{SUMIBK-Bank} & \multicolumn{1}{c}{NORBK} & \multicolumn{1}{c}{RBOS-RBOSplc} & \multicolumn{1}{c|}{UBS} & \multicolumn{1}{c|}{\textbf{Average}} \\
    \midrule
    $\hat{b}$     & 0.000631 & 0.000295 & 0.000039 & 0.000073 & 0.000125 & 0.000302 & 0.000683 & 0.000289 & 0.000103 & 0.000297 & 0.000572 & 0.000278 & 0.000773 & 0.000036 & 0.000001 & 0.000457 & 0.000295 & 0.000309 \\
   $\hat{b}_0(1)$   & 0.000304 & 0.000456 & 0.001016 & 0.000863 & 0.000636 & 0.000460 & 0.000714 & 0.000740 & 0.000715 & 0.002141 & 0.000172 & 0.000435 & 0.000471 & 0.000369 & 0.000548 & 0.000553 & 0.000456 & 0.000650 \\
     $\hat{b}_0(2)$  & 0.000821 & 0.000684 & 0.001215 & 0.001287 & 0.000840 & 0.000592 & 0.000994 & 0.000739 & 0.001171 & 0.001760 & 0.000323 & 0.000502 & 0.000617 & 0.000464 & 0.000775 & 0.000554 & 0.000684 & 0.000825 \\
    $\hat{b}_0(3)$ & 0.001234 & 0.001105 & 0.001732 & 0.001765 & 0.001504 & 0.001091 & 0.002294 & 0.001330 & 0.001545 & 0.002601 & 0.000773 & 0.000585 & 0.001424 & 0.000731 & 0.001229 & 0.001021 & 0.001105 & 0.001357 \\
  $\hat{b}_0(4)$ & 0.001638 & 0.001530 & 0.002300 & 0.002386 & 0.002175 & 0.001565 & 0.003729 & 0.002254 & 0.001877 & 0.003641 & 0.001260 & 0.000712 & 0.002146 & 0.001001 & 0.001724 & 0.001537 & 0.001530 & 0.001942 \\
    $\hat{b}_0(5)$  & 0.002218 & 0.002144 & 0.003011 & 0.002971 & 0.002996 & 0.002154 & 0.005130 & 0.002974 & 0.002448 & 0.004720 & 0.001842 & 0.001319 & 0.002907 & 0.001444 & 0.002342 & 0.002237 & 0.002144 & 0.002647 \\
\bottomrule
    \end{tabular}%
    }

  \end{subtable}

  \begin{subtable}{.5\textwidth}
  \centering
  \caption{Three-factor CIR}
    \scalebox{0.40}{
  \begin{tabular}{|l|rrrrrrrrrrrrrrrrr|r|}
    \toprule
    \multicolumn{1}{|c|}{31/10/2017} & \multicolumn{1}{c}{BACORP} & \multicolumn{1}{c}{MUFJ} & \multicolumn{1}{c}{BACR} & \multicolumn{1}{c}{C} & \multicolumn{1}{c}{ACAFP} & \multicolumn{1}{c}{CSGAG} & \multicolumn{1}{c}{DB} & \multicolumn{1}{c}{HSBC} & \multicolumn{1}{c}{JPM} & \multicolumn{1}{c}{LBGP} & \multicolumn{1}{c}{COOERAB} & \multicolumn{1}{c}{RY} & \multicolumn{1}{c}{SOCGEN} & \multicolumn{1}{c}{SUMIBK-Bank} & \multicolumn{1}{c}{NORBK} & \multicolumn{1}{c}{RBOS-RBOSplc} & \multicolumn{1}{c|}{UBS} & \multicolumn{1}{c|}{\textbf{Average}} \\
    \midrule
   $\hat{b}_1$   & 0.001297 & 0.000020 & 0.000000 & 0.000779 & 0.000455 & 0.000341 & 0.000734 & 0.000857 & 0.000505 & 0.002993 & 0.000295 & 0.000043 & 0.000031 & 0.000434 & 0.000294 & 0.000092 & 0.000384 & 0.000562 \\
    $\hat{b}_2$    & 0.000041 & 0.000204 & 0.000805 & 0.000254 & 0.000328 & 0.000004 & 0.000010 & 0.000039 & 0.000240 & 0.000777 & 0.000026 & 0.000006 & 0.000418 & 0.000343 & 0.000268 & 0.000287 & 0.000301 & 0.000256 \\
   $\hat{b}_3$  & 0.002418 & 0.000868 & 0.000029 & 0.000725 & 0.000733 & 0.000869 & 0.003653 & 0.000547 & 0.001749 & 0.002645 & 0.000231 & 0.000713 & 0.001506 & 0.000063 & 0.000087 & 0.000771 & 0.000867 & 0.001087 \\
    $\hat{b}_0(1)$ & 0.000033 & 0.000485 & 0.000726 & 0.000467 & 0.000358 & 0.000466 & 0.000520 & 0.000597 & 0.000305 & 0.000730 & 0.000478 & 0.000543 & 0.000699 & 0.000111 & 0.000334 & 0.000666 & 0.000264 & 0.000458 \\
    $\hat{b}_0(2)$   & 0.000436 & 0.000677 & 0.000942 & 0.000864 & 0.000536 & 0.000557 & 0.000624 & 0.000568 & 0.000685 & 0.000247 & 0.000620 & 0.000574 & 0.000786 & 0.000212 & 0.000566 & 0.000639 & 0.000165 & 0.000570 \\
     $\hat{b}_0(3)$   & 0.000656 & 0.001039 & 0.001506 & 0.001297 & 0.001148 & 0.000988 & 0.001631 & 0.001119 & 0.000930 & 0.000905 & 0.001051 & 0.000605 & 0.001493 & 0.000488 & 0.001024 & 0.001051 & 0.000466 & 0.001023 \\
\bottomrule
    \end{tabular}%
    }
  \end{subtable}
 \caption{Model calibration to CDS data for USD LIBOR panel banks on 31/10/2017}\label{CDScoeff1}
\end{table}%
\clearpage

{\scriptsize

\begin{table}[t]
  \centering

  \scalebox{0.65}{
    \begin{tabular}{|rl|cccccccc|}
    \hline
    \multicolumn{1}{|l}{LIBOR Panel} & Maturity & 0.5   & 1     & 2     & 3     & 4     & 5     & 7     & 10 \\
    \hline
    \multicolumn{1}{|l}{\textbf{BACORP}} & Market & 0.04948\% & 0.06922\% & 0.13891\% & 0.22322\% & 0.27845\% & 0.32756\% & 0.37376\% & 0.41302\% \\
\cline{3-10}          & 1-factor & 0.04945\% & 0.06920\% & 0.13898\% & 0.22321\% & 0.27829\% & 0.32760\% & 0.37385\% & 0.41284\% \\
          & 3-factor & 0.04948\% & 0.06922\% & 0.13888\% & 0.22327\% & 0.27840\% & 0.32730\% & 0.37385\% & 0.41312\% \\
    \hline
   \hline
    \multicolumn{1}{|l}{\textbf{MUFJ-BTMUFJ}} & Market & 0.06353\% & 0.07170\% & 0.10938\% & 0.14172\% & 0.18472\% & 0.22588\% & 0.27046\% & 0.30236\% \\
\cline{3-10}          & 1-factor & 0.06353\% & 0.07170\% & 0.10938\% & 0.14172\% & 0.18472\% & 0.22588\% & 0.27046\% & 0.30236\% \\
          & 3-factor & 0.06365\% & 0.07156\% & 0.10944\% & 0.14181\% & 0.18485\% & 0.22450\% & 0.27106\% & 0.30287\% \\
    \hline
    \hline

    \multicolumn{1}{|l}{\textbf{BACR-Bank}} & Market & 0.04903\% & 0.05996\% & 0.12029\% & 0.18846\% & 0.25487\% & 0.32914\% & 0.38240\% & 0.41731\% \\
\hline
    \end{tabular}%
    }
    \caption{Model calibration to USD LIBOR panel banks on 01/01/2013 }
  \label{CDSfit130101}%
\end{table}%

}

{\scriptsize

\begin{table}[t]
  \centering

  \scalebox{0.65}{
    \begin{tabular}{|rl|cccccccc|}
    \hline
    \multicolumn{1}{|l}{LIBOR Panel} & Maturity & 0.5   & 1     & 2     & 3     & 4     & 5     & 7     & 10 \\
    \hline
    \multicolumn{1}{|l}{\textbf{BACORP}} & Market & 0.000739 & 0.000895 & 0.001260 & 0.001607 & 0.002029 & 0.002402 & 0.002998 & 0.003592 \\
\cmidrule{3-10}          & 1-factor & 0.000739 & 0.000895 & 0.001259 & 0.001607 & 0.002028 & 0.002402 & 0.002997 & 0.003592 \\
          & 3-factor & 0.000739 & 0.000895 & 0.001266 & 0.001603 & 0.002026 & 0.002402 & 0.002992 & 0.003588 \\
    \hline
   \hline
    \multicolumn{1}{|l}{\textbf{MUFJ-BTMUFJ}} & Market & 0.000360 & 0.000464 & 0.000680 & 0.000945 & 0.001286 & 0.001630 & 0.002037 & 0.002269 \\
\cmidrule{3-10}          & 1-factor & 0.000360 & 0.000464 & 0.000681 & 0.000944 & 0.001286 & 0.001631 & 0.002036 & 0.002270 \\
          & 3-factor & 0.000361 & 0.000465 & 0.000681 & 0.000943 & 0.001290 & 0.001631 & 0.002027 & 0.002270 \\
    \hline
 \hline
    \multicolumn{1}{|l}{\textbf{BACR-Bank}} & Market & 0.000443 & 0.000576 & 0.000897 & 0.001264 & 0.001663 & 0.002062 & 0.002642 & 0.002999 \\
\hline
    \end{tabular}%
    }
    \caption{Model calibration to USD LIBOR panel banks on 08/09/2014 }
  \label{CDSfit140908}%
\end{table}%

}

\begin{table}[t]
  \centering

  \scalebox{0.65}{
    \begin{tabular}{|rl|cccccccc|}
    \toprule
    \multicolumn{1}{|l}{LIBOR Panel} & Maturity & 0.5   & 1     & 2     & 3     & 4     & 5     & 7     & 10 \\
    \midrule
    \multicolumn{1}{|l}{\textbf{BACORP}} & Market & 0.04813\% & 0.06260\% & 0.08497\% & 0.11369\% & 0.14062\% & 0.17582\% & 0.23483\% & 0.27894\% \\
\cmidrule{3-10}          & 1-factor & 0.04814\% & 0.06257\% & 0.08499\% & 0.11366\% & 0.14060\% & 0.17582\% & 0.23474\% & 0.27887\% \\
          & 3-factor & 0.04823\% & 0.06287\% & 0.08451\% & 0.11443\% & 0.14098\% & 0.17649\% & 0.23487\% & 0.27936\% \\
    \midrule
    \midrule
    \multicolumn{1}{|l}{\textbf{MUFJ-BTMUFJ}} & Market & 0.03513\% & 0.04003\% & 0.05647\% & 0.07192\% & 0.09724\% & 0.12071\% & 0.15429\% & 0.17695\% \\
\cmidrule{3-10}          & 1-factor & 0.03513\% & 0.04003\% & 0.05646\% & 0.07193\% & 0.09724\% & 0.12070\% & 0.15429\% & 0.17695\% \\
          & 3-factor & 0.03507\% & 0.04006\% & 0.05633\% & 0.07184\% & 0.09746\% & 0.12045\% & 0.15454\% & 0.17679\% \\
    \midrule
    \midrule
    \multicolumn{1}{|l}{\textbf{BACR-Bank}} & Market & 0.10031\% & 0.11363\% & 0.13322\% & 0.15291\% & 0.17759\% & 0.20189\% & 0.24070\% & 0.26574\% \\
\bottomrule
    \end{tabular}%
    }
    \caption{Model calibration to USD LIBOR panel banks on 18/06/2015 }
  \label{CDSfit150618}%
\end{table}%

\begin{table}[t]
  \centering

  \scalebox{0.65}{
    \begin{tabular}{|rl|cccccccc|}
    \toprule
    \multicolumn{1}{|l}{LIBOR Panel} & Maturity & 0.5   & 1     & 2     & 3     & 4     & 5     & 7     & 10 \\
    \midrule
    \multicolumn{1}{|l}{\textbf{BACORP}} & Market & 0.06923\% & 0.09045\% & 0.11964\% & 0.14343\% & 0.17157\% & 0.21044\% & 0.26665\% & 0.30669\% \\
\cmidrule{3-10}          & 1-factor & 0.06922\% & 0.09047\% & 0.11973\% & 0.14347\% & 0.17164\% & 0.21039\% & 0.26676\% & 0.30681\% \\
          & 3-factor & 0.06928\% & 0.09030\% & 0.11980\% & 0.14317\% & 0.17159\% & 0.21089\% & 0.26678\% & 0.30672\% \\
    \midrule
    \midrule
    \multicolumn{1}{|l}{\textbf{MUFJ-BTMUFJ}} & Market & 0.04612\% & 0.05479\% & 0.08152\% & 0.10809\% & 0.14762\% & 0.18391\% & 0.21351\% & 0.23720\% \\
\cmidrule{3-10}          & 1-factor & 0.04613\% & 0.05479\% & 0.08153\% & 0.10809\% & 0.14762\% & 0.18388\% & 0.21350\% & 0.23720\% \\
          & 3-factor & 0.04613\% & 0.05479\% & 0.08136\% & 0.10812\% & 0.14766\% & 0.18394\% & 0.21363\% & 0.23741\% \\
    \midrule
    \midrule
    \multicolumn{1}{|l}{\textbf{BACR-Bank}} & Market & 0.14546\% & 0.16644\% & 0.19214\% & 0.21354\% & 0.24068\% & 0.26680\% & 0.31066\% & 0.33777\% \\
\bottomrule
    \end{tabular}%
    }
     \caption{Model calibration to USD LIBOR panel banks on 20/04/2016 }
  \label{CDSfit160420}%
\end{table}%

\begin{table}[t]
  \centering

  \scalebox{0.65}{
    \begin{tabular}{|rl|cccccccc|}
    \toprule
    \multicolumn{1}{|l}{LIBOR Panel} & Maturity & 0.5   & 1     & 2     & 3     & 4     & 5     & 7     & 10 \\
    \midrule
 \multicolumn{1}{|l}{\textbf{BACORP}} & Market & 0.07252\% & 0.11828\% & 0.14502\% & 0.17841\% & 0.20794\% & 0.25230\% & 0.30848\% & 0.36520\% \\
\cmidrule{3-10}          & 1-factor & 0.07252\% & 0.11827\% & 0.14505\% & 0.17840\% & 0.20790\% & 0.25232\% & 0.30847\% & 0.36533\% \\
          & 3-factor & 0.07253\% & 0.11807\% & 0.14483\% & 0.17842\% & 0.20798\% & 0.25241\% & 0.30817\% & 0.36527\% \\
    \midrule
    \midrule
    \multicolumn{1}{|l}{\textbf{MUFJ-BTMUFJ}} & Market & 0.05547\% & 0.06740\% & 0.08868\% & 0.10900\% & 0.14743\% & 0.18566\% & 0.21703\% & 0.23960\% \\
\cmidrule{3-10}          & 1-factor & 0.05548\% & 0.06744\% & 0.08876\% & 0.10883\% & 0.14740\% & 0.18556\% & 0.21729\% & 0.23964\% \\
          & 3-factor & 0.05551\% & 0.06738\% & 0.08871\% & 0.10911\% & 0.14725\% & 0.18562\% & 0.21704\% & 0.23945\% \\
    \midrule
    \midrule
    \multicolumn{1}{|l}{\textbf{BACR-Bank}} & Market & 0.06269\% & 0.08022\% & 0.13055\% & 0.18031\% & 0.22667\% & 0.26994\% & 0.32310\% & 0.35594\% \\
\bottomrule
    \end{tabular}%
    }
     \caption{Model calibration to USD LIBOR panel banks on 22/03/2017 }
  \label{CDSfit170322}%
\end{table}%

\begin{table}[t]
  \centering

  \scalebox{0.65}{
    \begin{tabular}{|rl|cccccccc|}
    \toprule
    \multicolumn{1}{|l}{LIBOR Panel} & Maturity & 0.5   & 1     & 2     & 3     & 4     & 5     & 7     & 10 \\
    \midrule
 \multicolumn{1}{|l}{\textbf{BACORP}} & Market & 0.04765\% & 0.07869\% & 0.10314\% & 0.12627\% & 0.15811\% & 0.19605\% & 0.25590\% & 0.31003\% \\
\cmidrule{3-10}          & 1-factor & 0.04765\% & 0.07868\% & 0.10315\% & 0.12630\% & 0.15813\% & 0.19609\% & 0.25582\% & 0.31001\% \\
          & 3-factor & 0.04769\% & 0.07869\% & 0.10324\% & 0.12620\% & 0.15799\% & 0.19595\% & 0.25577\% & 0.31000\% \\
    \midrule
    \midrule
    \multicolumn{1}{|l}{\textbf{MUFJ-BTMUFJ}} & Market & 0.04112\% & 0.05484\% & 0.07999\% & 0.10463\% & 0.13901\% & 0.17205\% & 0.20889\% & 0.23103\% \\
\cmidrule{3-10}          & 1-factor & 0.04114\% & 0.05484\% & 0.07991\% & 0.10465\% & 0.13904\% & 0.17205\% & 0.20874\% & 0.23104\% \\
          & 3-factor & 0.05551\% & 0.06738\% & 0.08871\% & 0.10911\% & 0.14725\% & 0.18562\% & 0.21704\% & 0.23945\% \\
    \midrule
    \midrule
    \multicolumn{1}{|l}{\textbf{BACR-Bank}} & Market & 0.06283\% & 0.07470\% & 0.10530\% & 0.13799\% & 0.17674\% & 0.20894\% & 0.24773\% & 0.27961\% \\
\bottomrule
    \end{tabular}%
    }
     \caption{Model calibration to USD LIBOR panel banks on 31/10/2017 }
  \label{CDSfit171031}%
\end{table}%

\begin{figure}[t]
  \centering
    \begin{subfigure}[b]{0.53\textwidth}
        \includegraphics[width=\textwidth]{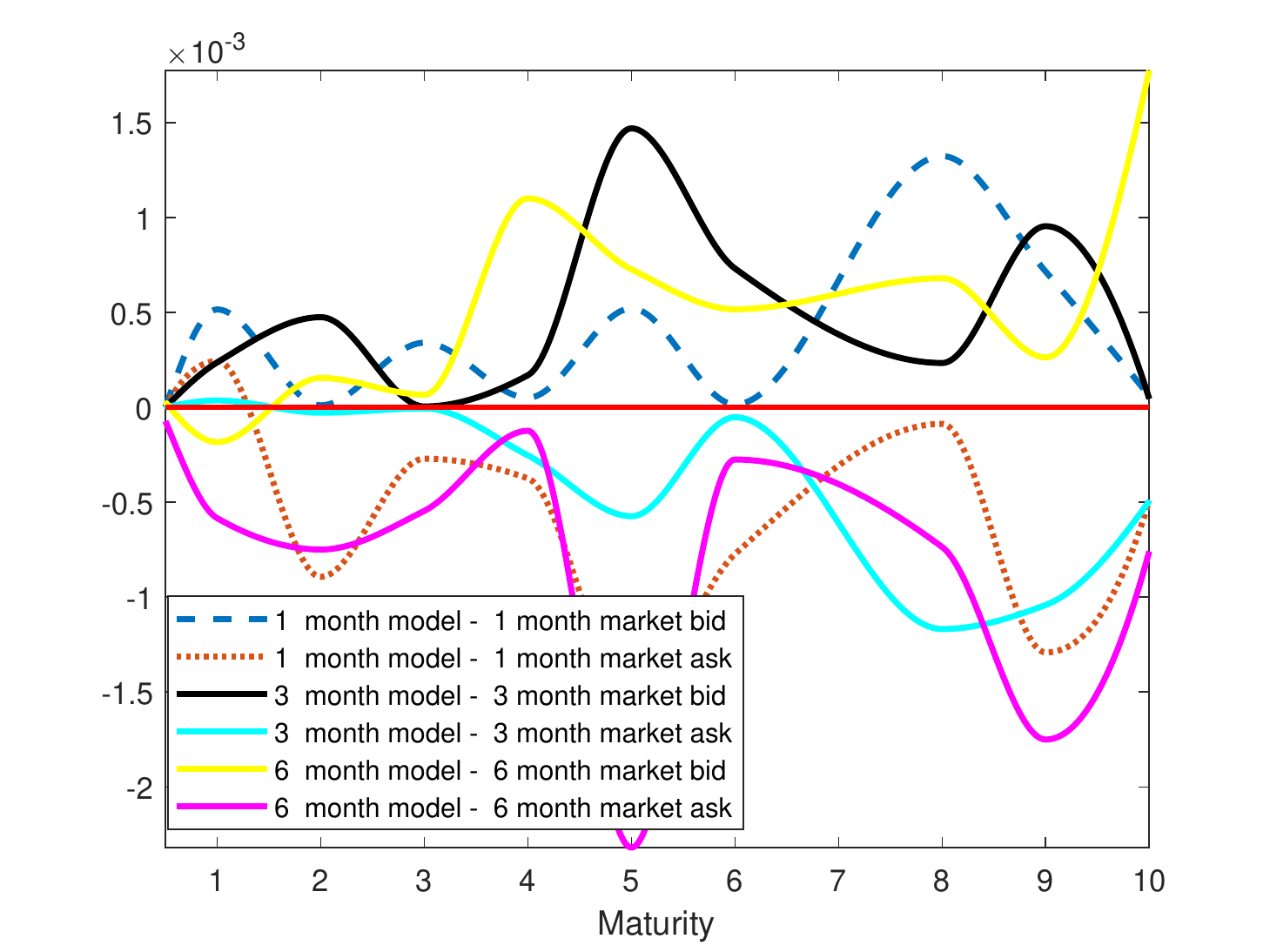}
        \caption{1 January 2013}
        \label{fig:uc}
    \end{subfigure}
    ~
        \begin{subfigure}[b]{0.55\textwidth}
        \includegraphics[width=\textwidth]{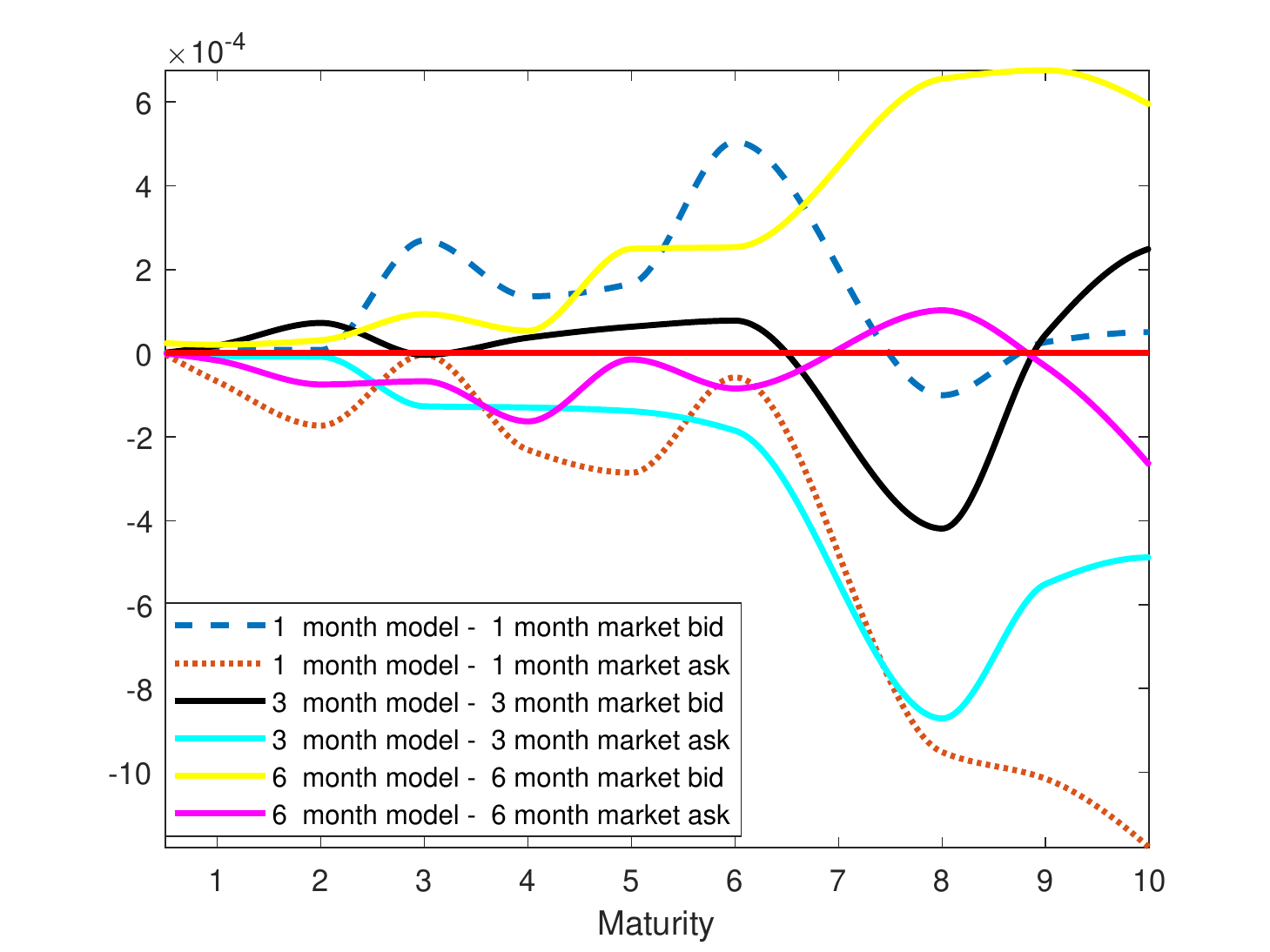}
       \caption{08 September 2014}
        \label{figxc}
    \end{subfigure}

    \begin{subfigure}[b]{0.55\textwidth}
        \includegraphics[width=\textwidth]{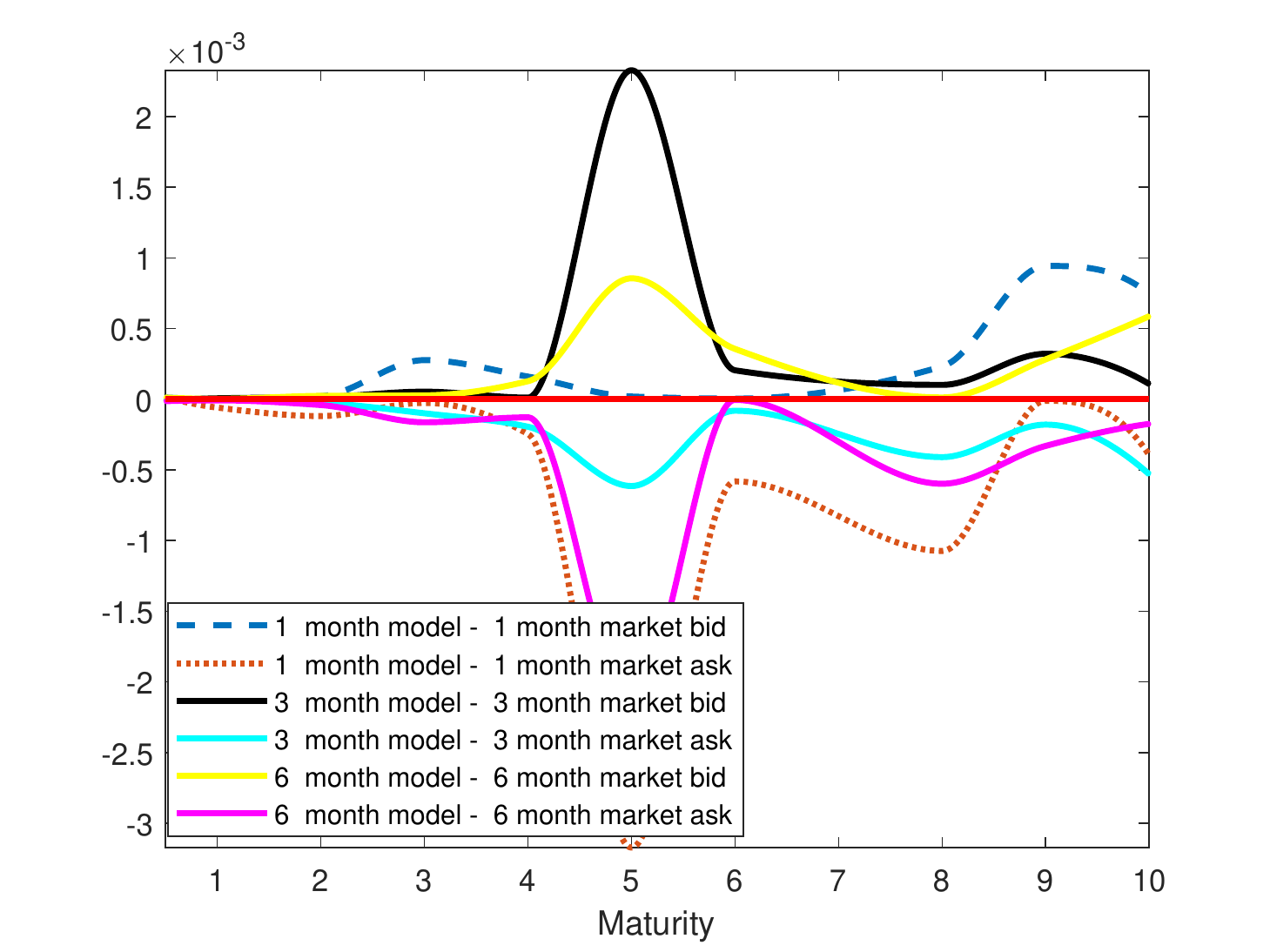}
       \caption{18 June 2015}
        \label{figxc}
    \end{subfigure}
          ~
    \begin{subfigure}[b]{0.53\textwidth}
        \includegraphics[width=\textwidth]{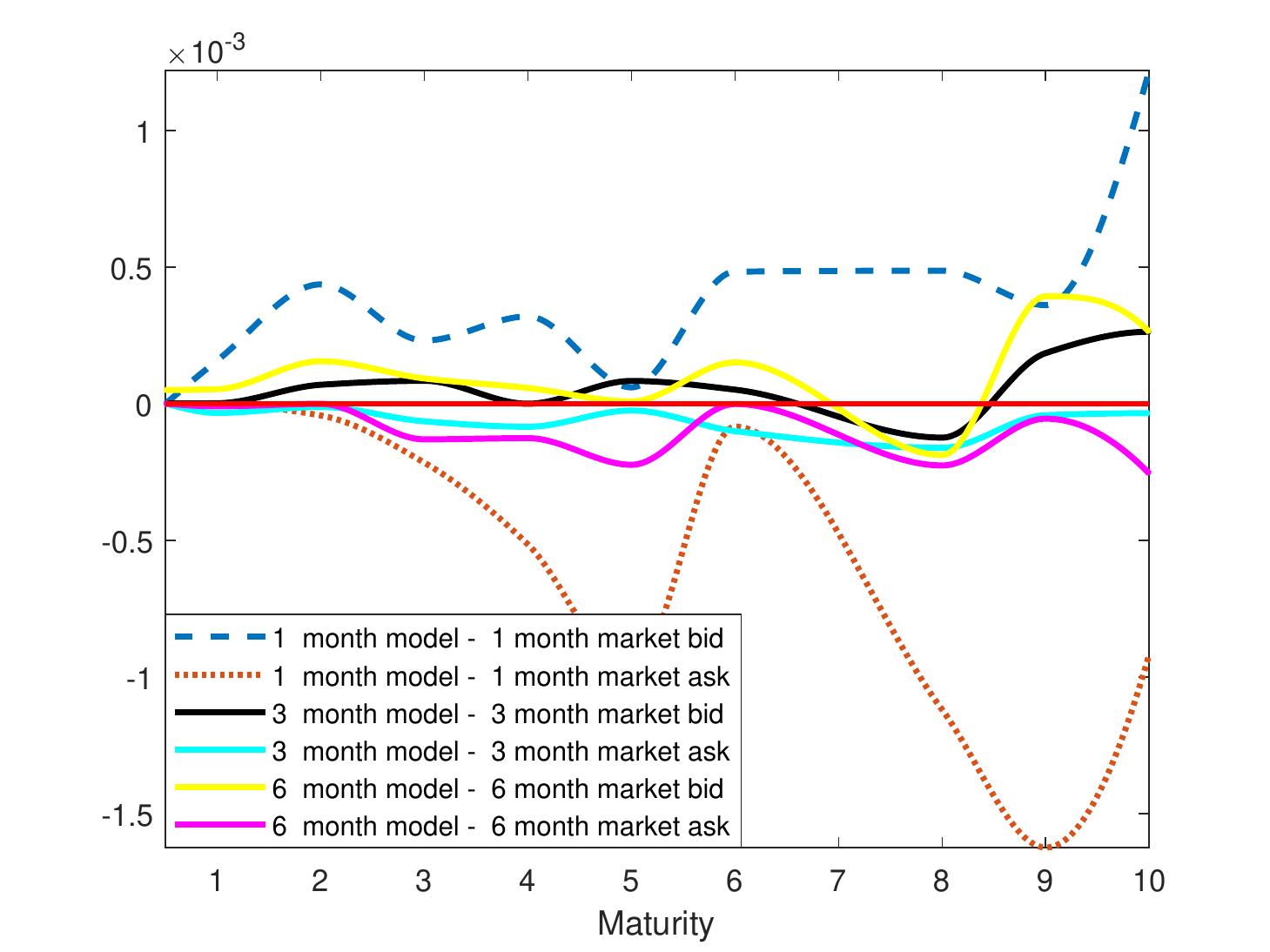}
       \caption{20 April 2016}
        \label{figxb}
    \end{subfigure}

     \begin{subfigure}[b]{0.53\textwidth}
        \includegraphics[width=\textwidth]{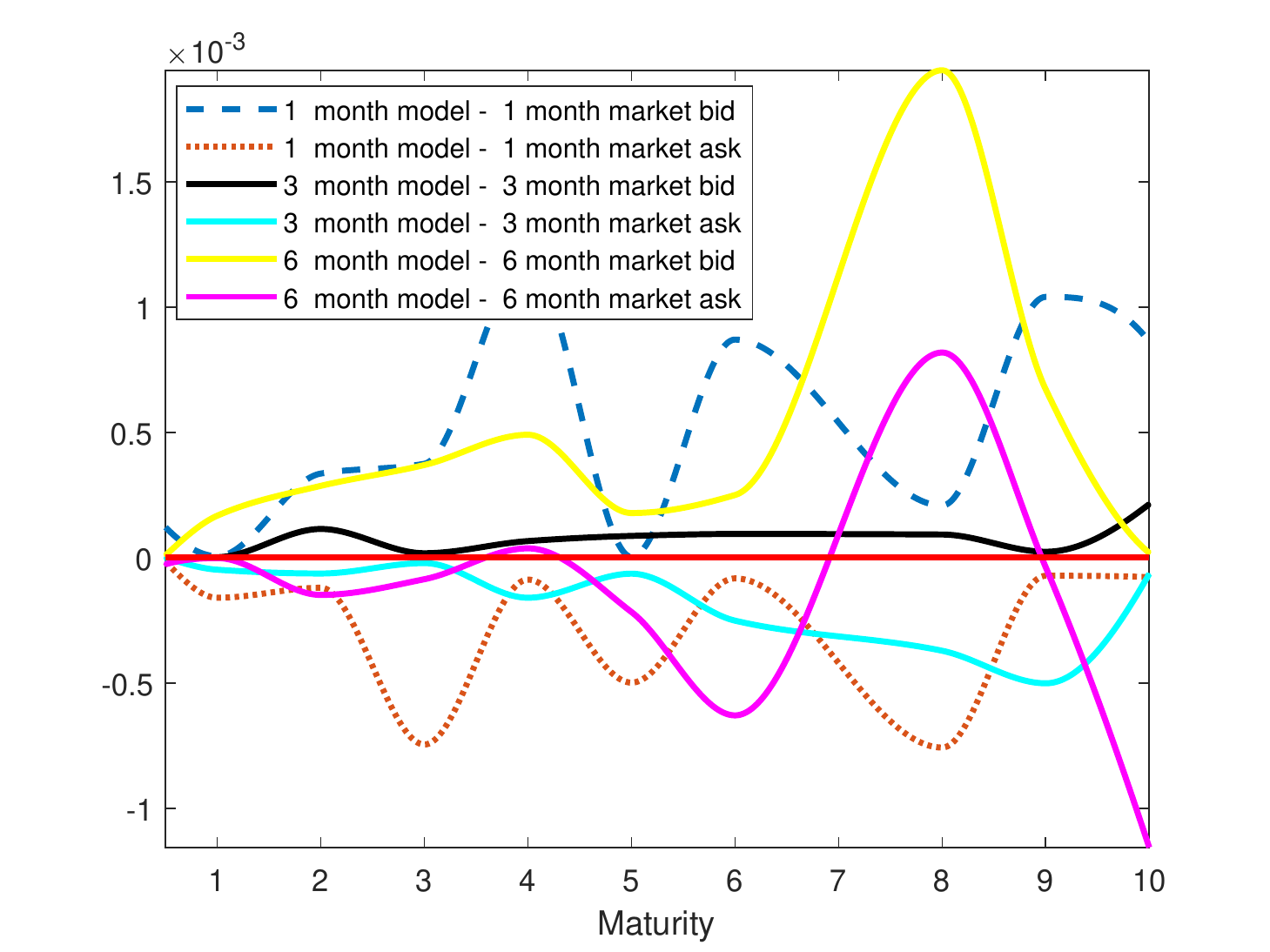}
       \caption{22 March 2017}
        \label{figx}
    \end{subfigure}
    ~
     \begin{subfigure}[b]{0.53\textwidth}
        \includegraphics[width=\textwidth]{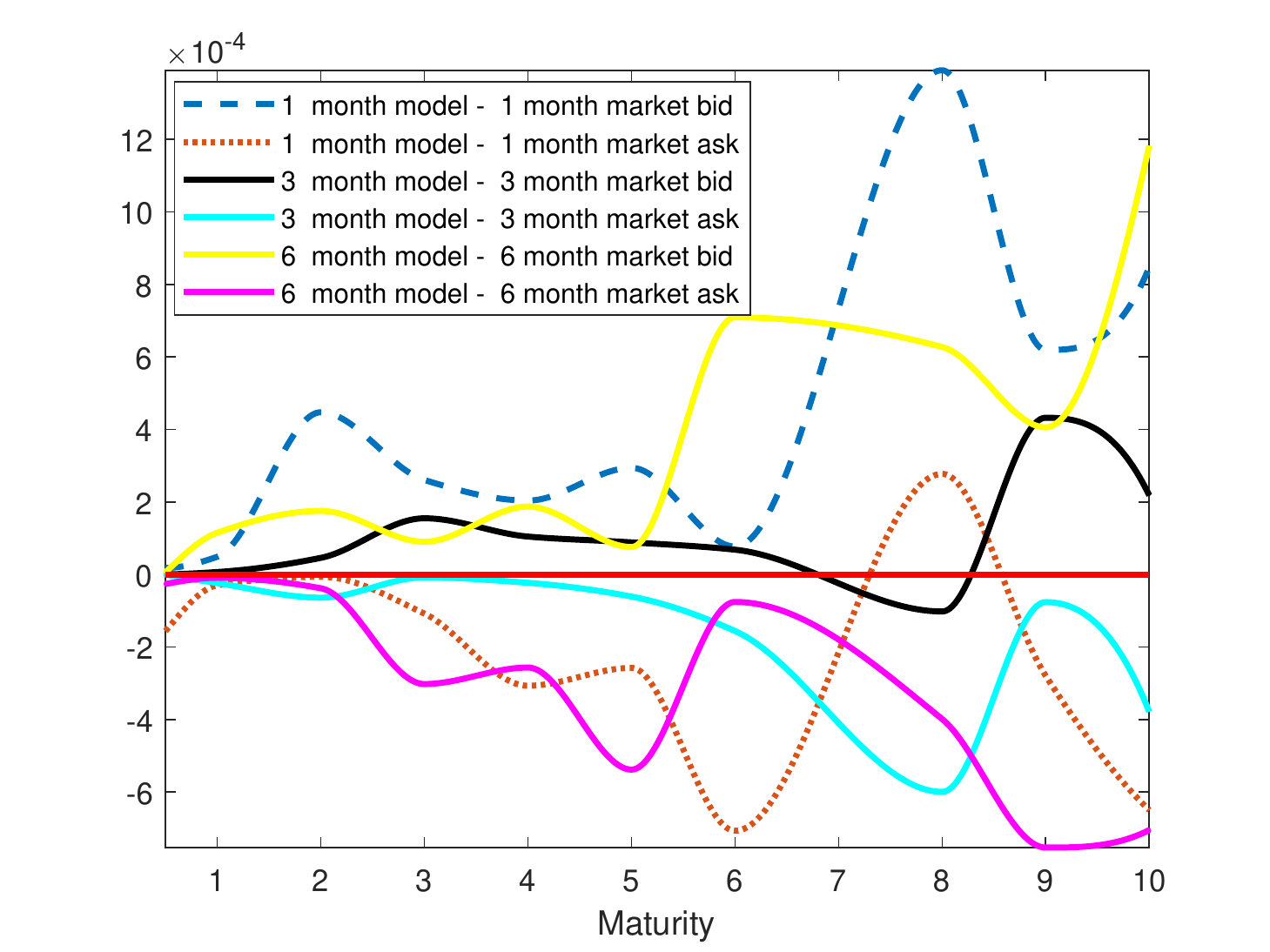}
       \caption{31 October 2017}
        \label{figxa}
    \end{subfigure}
\caption{One--factor model fit to Basis Swaps on 1 January 2013,  18 June 2015 and 20 April 2016 (based on data from Bloomberg)}\label{CDSbasis1factor}
\end{figure}

 \begin{figure}[t]
  \centering

 \includegraphics[width=\textwidth]{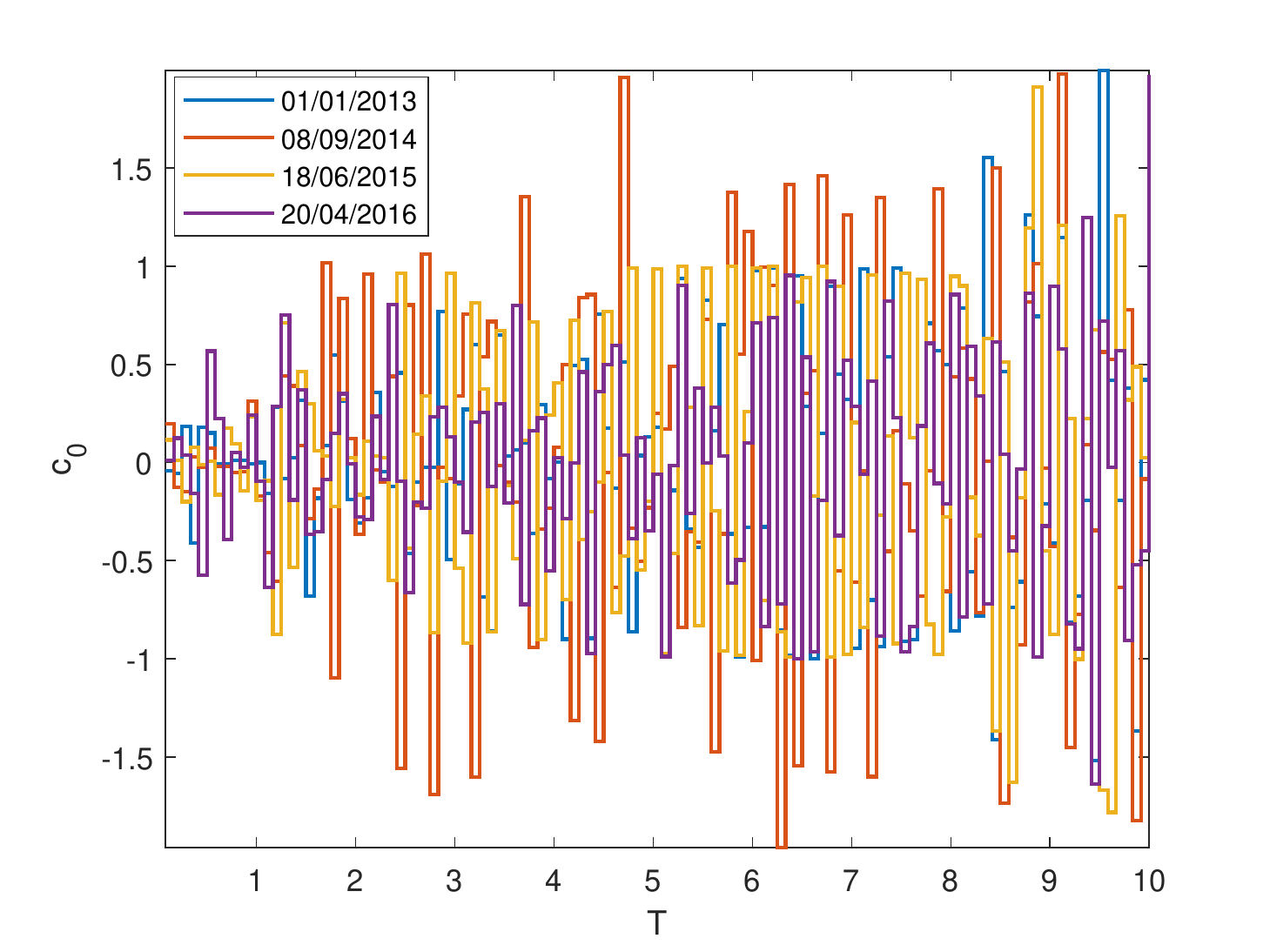}

\caption{Time-dependent calibrated parameter $c_0$ }\label{c0plot}
\end{figure}

 \begin{figure}[t]
  \centering

 \includegraphics[width=\textwidth]{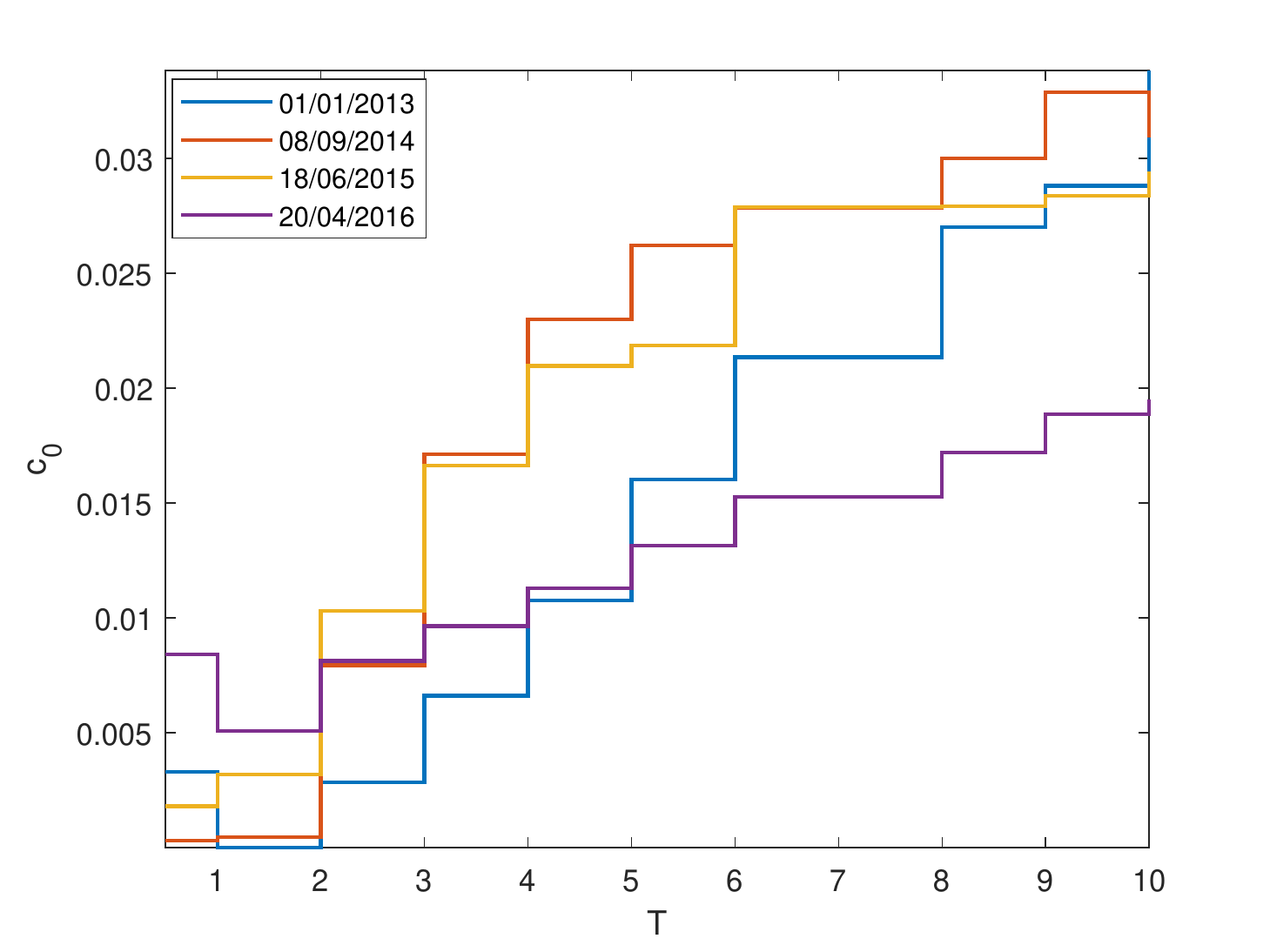}

\caption{Time-dependent calibrated parameter $c_0$ on a 3--month tenor }\label{c03m}
\end{figure}

\section{Conclusion}
Since seeking to profit from the frequency basis (be it the LIBOR/OIS spread or basis swaps involving different LIBOR tenors) entails lending at a longer tenor against rolling over borrowing at a shorter tenor, modelling the risk of such a strategy, i.e. ``roll--over risk,'' is a natural approach to obtain a consistent model of all possible tenors in a relatively parsimonious fashion. Such a model --- as we have shown --- can be calibrated to a large number of relevant market instruments simultaneously. Here, we have deliberately included two distinct calibration exercises: In the first, the model is calibrated to OIS and all available USD interest rate swaps (vanilla IRS and basis swaps) --- this demonstrates that the model can represent the interest rate markets for different tenor frequencies as effectively as the more common, \emph{ad hoc} models in the literature, which typically model spreads between term structures for different tenors directly, without recourse to any underlying theoretical justification for the existence of such spreads. In the second calibration exercise, the component of roll--over risk associated with the risk of a credit downgrade of the entity attempting to roll over their borrowing is explicitly separated from the component interpreted as funding liquidity risk, by including credit default swap data in the calibration. Again, it is demonstrated that the model can fit the market represented by this expanded set of calibration instruments (OIS, IRS, basis swaps and CDS).

Thus, in the sense of the bulk of the derivatives pricing literature, the present paper takes the practitioners' approach of static calibration to (cross--sectional) market data on a given day, for a model to be used for pricing less liquid derivatives in an arbitrage--free manner relative to a set of calibration instruments. The consistent model guarantees the absence of arbitrage, though as in other applications where markets are incomplete\footnote{In the sense of \citeasnoun{Radner:72}.} (credit derivatives are a prominent example), we remain agnostic as to the extent to which arbitrage--free prices are unique. Unlike approaches which do not model the risks driving tenor spreads, the model presented here can be used to provide basis spreads for bespoke tenors not quoted in the market. This could also be employed to ``fill in the gaps'' in markets where key term structures (such as OIS) are missing, potentially improving on less structured, regression--based approaches hitherto available (such as \citeasnoun{jakarasi_estimating_2015}).

The relationship between tenor basis spreads (in particular the LIBOR/OIS spread) and roll--over risk, as expressed in the model presented in this paper and calibrated to market data, also serves as a cautionary note in the present debate on the replacement of LIBOR--type benchmarks with benchmarks based on overnight lending, such as the Secured Overnight Financing Rate (SOFR). As \citeasnoun{Henrard:18} notes in a critical overview of the current discussions on this issue, most proposals suggest no more than adjusting overnight benchmarks for the existing LIBOR/OIS spread by adding a fixed spread to e.g. SOFR. This would imply taking the view that roll--over risk (as priced by the market) is known and constant for all time, and even a cursory analysis (such as calibrating the model presented here to market data on a few exemplary dates) shows that this is not the case.

\appendix

\section{Affine process specification}\label{affine}

 The Markov process $X$ is affine if it is stochastically continuous and its Fourier-Laplace transform has exponential-affine dependence on the initial state, that is there exist
some deterministic functions $\Phi_u:\mathbb{R}_+\rightarrow \mathbb{C}$ and $\Psi_u:\mathbb{R}_+\rightarrow \mathbb{C}^d$ such that  the semigroup $P$ acts as follows:
\begin{align}
\int_E e^{\langle u ,w \rangle} P_t(x,dw) &
= \mathbb{E}^{\mathbb{Q}}\left[e^{\langle u,X(t)\rangle}  \right]\\
&= e^{\Phi_u(t)+\langle\Psi_u(t) ,x\rangle}
\label{affineeq}
\end{align}
for all $t\geq0$, $X(0)=x\in E$ and $u\in {\cal D}_T$, where ${\cal D}_T\subseteq \mathbb{C}^d$ is such that $\mathbb{E}^{\mathbb{Q}}\left[e^{\langle u,X(t)\rangle}  \right]<\infty$ for all $u\in{\cal D}_T$ and $t\leq T$.
It can be shown (see e.g. Cuchiero et al. \cite{CuchieroEtAl2011}) that the process $X$ is a semimartingale with characteristics
\begin{eqnarray*}
A_t&=& \int_0^t \alpha(X({s-}))ds,\\
B_t&=&\int_0^t \beta(X({s-}))ds,\\
\nu (\omega,dt,d\xi ) &=& K(X({t-},\omega),d\xi)dt,
\end{eqnarray*}
with $\alpha(x),\beta(x),K(x,d\xi)$ affine functions:
\begin{eqnarray*}
\alpha(x)&=&\alpha^0+ x_1 \alpha^1+...+x_d \alpha^d,\\
\beta(x)&=&\beta^0+ x_1 \beta^1+...+x_d \beta^d,\\
K(x,d\xi)&=&\mu^0(d\xi)+x_1 \mu^1 (d\xi)+...+x_d \mu^d(d\xi),
\end{eqnarray*}
where $\alpha(x)$ (the diffusion coefficient) is a positive semidefinite $d\times d$ matrix, $\beta(x)$ is the $\mathbb{R}^d$-vector of the drift, and $K(x,d\xi)$ is a Radon measure on $\mathbb{R}^d$ associated to the affine jump part and it is such that
\begin{equation*}
\int_{\mathbb{R}^d} \left( \lVert \xi \rVert ^2 \wedge 1\right)K(x,d\xi) < \infty
\end{equation*}
and $K(x, \{ 0\})=0$.\\

The deterministic functions $\Phi_u(t),\Psi_u(t)$ solve the $generalized\ Riccati\ equations$

\begin{align*}
\frac{\partial}{\partial t} \Phi_u(t) & = \frac{1}{2}\langle \Psi_u(t),\alpha^0\Psi_u(t) \rangle+\langle \beta^0,\Psi_u(t) \rangle+\int_{\mathbb{R}^d\setminus \{ 0\}}
\left(e^{- \langle\xi, \Psi_u(t)\rangle}-1-\langle h(\xi),\Psi_u(t) \rangle\right)\mu^0(d\xi),\\
\Phi_u(0) & = 0,
\end{align*}
and for all $i=1,...,d:$
\begin{align*}
\frac{\partial}{\partial t} \Psi^i_u(t)&  = \frac{1}{2}\langle\Psi_u(t) ,\alpha^i \Psi_u(t) \rangle+\langle\beta^i,\Psi_u(t) \rangle+
\int_{\mathbb{R}^d\setminus \{ 0\}}\left(e^{- \langle\xi, \Psi_u(t)\rangle}-1-\langle h(\xi),\Psi_u(t) \rangle\right)\mu^i (d\xi),\\
\Psi_u(0) & = u,
\end{align*}
where $h(\xi)={\bf 1}_{\{\parallel \xi\parallel \leq 1\} }\xi$ is a truncation function.\\

It is also useful to consider the process
$ (X,Y^{\gamma}) := (X,\int_0^{\cdot} \langle\gamma ,X(u)\rangle du)$ which is an affine process with state space
$E \times \mathbb{R}$ starting from $(x,0)$.

\begin{lemma}\label{lem:lemmaRiccatiXY}
Let $\tilde{P}^{\gamma}$ be the semigroup of the process $ (X,Y^{\gamma})$. Then we have
\begin{equation}
\int_{E \times \mathbb{R}} e^{\langle u ,w\rangle + vz} \tilde{P}^{\gamma}_t((x,y),(dw,dz))
=e^{ \Phi_{(u,v)}(t,\gamma) +\langle\Psi_{(u,v)}(t,\gamma), x\rangle + v y}\label{Fourierextended}
\end{equation}
for every $(u,v)\in {\cal V}_T\subseteq \mathbb{C}^{d+1}$ such that the LHS of Eq. \eqref{Fourierextended} is finite. Here the functions $\Phi_{(u,v)}(\cdot,\gamma)$ and $\Psi_{(u,v)}(\cdot,\gamma)$ satisfy the following system of generalized Riccati ODEs
\begin{align}
\frac{\partial}{\partial t} \Phi_{(u,v)}(t,\gamma)= & \  \frac{1}{2}\langle \Psi_{(u,v)}(t,\gamma),\alpha^0\Psi_{(u,v)}(t,\gamma)
 \rangle+\langle \beta^0,\Psi_{(u,v)}(t,\gamma) \rangle \nonumber\\
 & +\int_{\mathbb{R}^d\setminus \{ 0\}}
\left(e^{- \langle\xi, \Psi_{(u,v)}(t,\gamma)\rangle}-1-\langle h(\xi),\Psi_{(u,v)}(t,\gamma) \rangle\right)\mu^0(d\xi),\label{riccatiPhi}\\
\Phi_{(u,v)}(0,\gamma) =& \ 0,\nonumber
\end{align}
and for $i=1,...,d$
\begin{align}
\frac{\partial}{\partial t} \Psi^i_{(u,v)}(t,\gamma) =& \ v\gamma^i + \frac{1}{2}\langle\Psi_{(u,v)}(t,\gamma) ,\alpha^i \Psi_{(u,v)}(t,\gamma)
 \rangle+\langle\beta^i,\Psi_{(u,v)}(t,\gamma) \rangle \nonumber\\
& + \int_{\mathbb{R}^d\setminus \{ 0\}}\left(e^{- \langle\xi, \Psi_{(u,v)}(t,\gamma)\rangle}-1-\langle h(\xi),\Psi_{(u,v)}(t,\gamma)\rangle\right)\mu^i (d\xi),\label{riccatiPsi}\\
\Psi_{(u,v)}(0,\gamma) =& \ u.\nonumber
\end{align}
\end{lemma}

In our calibration example we consider the important case where $X$ follows a $d-$ dimensional vector of independent Cox--Ingersoll--Ross (CIR), for which the Riccati ODE admit a closed form expression. This  dynamics is of particular interest because,  as discussed for the one--factor case in Cox et al. (\citeyear*{CIR1985}),  a simple condition ensures the positivity of the modelled object. In order to keep the model analytically tractable, we do not allow for time--dependent coefficients at this stage\footnote{For example, following Schl\"ogl and Schl\"ogl \cite{OZ:Sch&Sch:00}, the coefficients could be made piecewise constant to facilitate calibration to at--the--money option price data, while retaining most analytical tractability.}. What is more, in the CIR case there exists a complete explicit characterization of the maximal domains ${\cal D}_T,{\cal V}_T$ for the transforms. This turns out to be very useful in the calibration procedure when we shall provide some bounds for the parameters.

\section{Computation of the relevant expectations}\label{expectations}


Let us first consider the expectations involved in formula \eqref{noarb}, i.e. $\mathbb{E}_t^\mathbb{Q}\left[e^{\int_{t}^T\phi(s)ds}\right]$, $\mathbb{E}_t^\mathbb{Q}\left[e^{-\int_{t}^T(r_c(s)+\lambda(s)q)ds}\right] $ and $D^{OIS}(t,T)=\mathbb{E}_t^\mathbb{Q}\left[e^{-\int_t^Tr_c(s)ds}\right].
$

 Assume that the deterministic functions $a_0(t), b_0(t), c_0(t)$ are integrable and $a,b,c\in E$ are constant. Thanks to the time-homogeneity of the affine  process $X$, it follows that all these conditional expectations will depend on time through the difference $(T-t)$.

Using the notation of the Appendix \ref{affine}, it is easy to show that
\begin{align}
\mathbb{E}_t^\mathbb{Q}\left[e^{\int_{t}^T\phi(s)ds}\right]=&
e^{\int_t^Tc_0(s)ds}\mathbb{E}_t^\mathbb{Q}\left[e^{\int_t^T\langle c, X(s)\rangle ds}\right]\nonumber\\
=&e^{\int_t^Tc_0(s)ds+\Phi_{(0,1)}(T-t,c)+\langle \Psi_{(0,1)}(T-t,c),X(t)\rangle },\label{transformphi}
\end{align}
where the deterministic functions $\Phi_{(0,1)}, \Psi_{(0,1)}$ satisfy the generalized Riccati ODE in Lemma \ref{lem:lemmaRiccatiXY}.

Analogously,
\begin{align}
\quad \mathbb{E}_t^\mathbb{Q}\left[e^{-\int_{t}^T(r_c(s)+\lambda(s)q)ds}\right]=&
e^{-\int_t^T(a_0(s)+b_0(s)q)ds}\mathbb{E}_t^\mathbb{Q}\left[e^{-\int_t^T\langle a+bq, X(s)\rangle ds}\right]\nonumber\\
=&e^{-\int_t^T(a_0(s)+b_0(s)q)ds+\Phi_{(0,1)}(T-t,-(a+bq))+\langle \Psi_{(0,1)}(T-t,-(a+bq)),X(t)\rangle }\label{transformrclambda}
\end{align}
and
\begin{align}
D^{OIS}(t,T)=& \mathbb{E}_t^\mathbb{Q}\left[e^{-\int_{t}^Tr_c(s)ds}\right]\nonumber\\
=&
e^{-\int_t^Ta_0(s)ds}\mathbb{E}_t^\mathbb{Q}\left[e^{-\int_t^T\langle a, X(s)\rangle ds}\right]\nonumber\\
=&e^{-\int_t^Ta_0(s)ds+\Phi_{(0,1)}(T-t,-a)+\langle \Psi_{(0,1)}(T-t,-a),X(t)\rangle }.\label{OISbondprice}
\end{align}
The proof follows immediately from the definition of $r_c,\lambda,\phi$ and the Lemma \ref{lem:lemmaRiccatiXY}.

Let us now turn out attention to formula \eqref{nnoarb}, which involves the following expectation:
\begin{align}
&\mathbb{E}_t^\mathbb{Q}\left[e^{-\int_{t}^{T_{j-1}}(r_c(s)+\lambda(s)q)ds}e^{\int_{T_{j-1}}^{T_{j}}\phi (s)ds}\right]  \nonumber\\
&=\mathbb{E}_t^\mathbb{Q}\left[e^{-\int_{t}^{T_{j-1}}(r_c(s)+\lambda(s)q)ds}\mathbb{E}_{T_{j-1}}^\mathbb{Q}\left[e^{\int_{T_{j-1}}^{T_{j}}\phi (s)ds}\right]\right]  \nonumber\\
& =\mathbb{E}_t^\mathbb{Q}\left[e^{-\int_{t}^{T_{j-1}}(r_c(s)+\lambda(s)q)ds}e^{\int_{T_{j-1}}^{T_{j}}c_0(s)ds+\Phi_{(0,1)}(T_{j}-T_{j-1},c)+\langle \Psi_{(0,1)}(T_{j}-T_{j-1},c),X(T_{j-1})\rangle}
 \right]  \nonumber\\
&=e^{-\int_t^{T_{j-1}}(a_0(s)+b_0(s)q)ds +\int_{T_{j-1}}^{T_{j}}c_0(s)ds+\Phi_{(0,1)}(T_{j}-T_{j-1},c)}\nonumber\\
&\quad .\mathbb{E}_t^\mathbb{Q}\left[e^{\int_t^{T_{j-1}}\langle -(a+bq), X(s)\rangle ds +\langle \Psi_{(0,1)}(T_{j}-T_{j-1},c) , X(T_{j-1})\rangle }\right]\nonumber\\
&=e^{-\int_t^{T_{j-1}}(a_0(s)+b_0(s)q)ds +\int_{T_{j-1}}^{T_{j}}c_0(s)ds+\Phi_{(0,1)}(T_{j}-T_{j-1},c)}\nonumber\\
&\quad . e^{\Phi_{(\Psi_{(0,1)}(T_{j}-T_{j-1},c),1)}(T_{j-1}-t,-(a+bq))+\langle
\Psi_{(\Psi_{(0,1)}(T_{j}-T_{j-1},c),1)}(T_{j-1}-t,-(a+bq)),X(t)\rangle}
\end{align}

%
In conclusion, all the expectations in formulae  \eqref{noarb} and \eqref{nnoarb} can be explicitly computed.

\subsection{The LIBOR and the pricing of swaps}\label{CIRLIBOR}

In this subsection we compute the expectations appearing in swap transaction formula \eqref{swapcond}, which is crucial to the calibration procedure.
The expression of $D^{OIS}$ has already been computed in \eqref{OISbondprice}, so we focus now on the first expectation in \eqref{swapcond} involving the forward LIBOR.

From \eqref{libor} it follows immediately
\begin{align*}
\delta L(T_{j-1},T_j)=& -1+ e^{\int_{T_{j-1}}^{T_j}(c_0(s)+a_0(s)+qb_0(s))ds
+ \Phi_{(0,1)}(T_{j}-T_{j-1},c)-\Phi_{(0,1)}(T_{j}-T_{j-1},-(a+qb))}\\
& .e^{\langle \Psi_{(0,1)}(T_{j}-T_{j-1},c)-\Psi_{(0,1)}(T_{j}-T_{j-1},-(a+qb)),X(T_{j-1})\rangle}.
\end{align*}
We have
\begin{align*}
&\mathbb{E}_t^\mathbb{Q}\left[e^{-\int_{t}^{T_j}r_c(s)ds}L(T_{j-1},T_j)\right]\\
&=\mathbb{E}_t^\mathbb{Q}\left[e^{-\int_{t}^{T_{j-1}}r_c(s)ds}L(T_{j-1},T_j)D^{OIS}(T_{j-1},T_j)\right]\\
&=\mathbb{E}_t^\mathbb{Q}\left[e^{-\int_{t}^{T_{j-1}}(a_0(s)+\langle a, X(s)\rangle) ds
-\int_{T_{j-1}}^{T_{j}}a_0(s)ds+\Phi_{(0,1)}(T_{j}-T_{j-1},-a)+\langle \Psi_{(0,1)}(T_{j}-T_{j-1},-a), X(T_{j-1})\rangle }
L(T_{j-1},T_j)\right]\\
&=e^{-\int_{t}^{T_{j}}a_0(s)ds+\Phi_{(0,1)}(T_{j}-T_{j-1},-a)}
\mathbb{E}_t^\mathbb{Q}\left[e^{-\int_{t}^{T_{j-1}}\langle a, X(s)\rangle ds
+\langle \Psi_{(0,1)}(T_{j}-T_{j-1},-a), X(T_{j-1})\rangle}L(T_{j-1},T_j)\right].
\end{align*}
Into this equation we insert the expression of the forward  LIBOR in terms of the sum of the constant and the two exponentials and we apply Lemma \ref{lem:lemmaRiccatiXY}:
\begin{align*}
&\mathbb{E}_t^\mathbb{Q}\left[e^{-\int_{t}^{T_j}r_c(s)ds}L(T_{j-1},T_j)\right]\\
&=-\frac{1}{\delta}e^{-\int_{t}^{T_{j}}a_0(s)ds+\Phi_{(0,1)}(T_{j}-T_{j-1},-a)}
\mathbb{E}_t^\mathbb{Q}\left[e^{-\int_{t}^{T_{j-1}}\langle a, X(s)\rangle ds
+\langle \Psi_{(0,1)}(T_{j}-T_{j-1},-a), X(T_{j-1})\rangle}\right]\\
&+ \frac{1}{\delta}e^{-\int_{t}^{T_{j}}a_0(s)ds+\int_{T_{j-1}}^{T_j}(c_0(s)+a_0(s)+qb_0(s))ds
+\Phi_{(0,1)}(T_{j}-T_{j-1},-a)+\Phi_{(0,1)}(T_{j}-T_{j-1},c)-\Phi_{(0,1)}(T_{j}-T_{j-1},-(a+qb))}\\
&.\mathbb{E}_t^\mathbb{Q}\left[e^{-\int_{t}^{T_{j-1}}\langle a, X(s)\rangle ds
+\rangle \Psi_{(0,1)}(T_{j}-T_{j-1},-a)+ \Psi_{(0,1)}(T_{j}-T_{j-1},c)-\Psi_{(0,1)}(T_{j}-T_{j-1},-(a+qb)),X(T_{j-1})\rangle}\right]\\
&=-\frac{1}{\delta}e^{-\int_{t}^{T_{j}}a_0(s)ds+
\Phi_{(0,1)}(T_{j}-T_{j-1},-a)+
\Phi_{(\Psi_{(0,1)}(T_{j}-T_{j-1},-a),1)}(T_{j-1}-t,-a)+
\langle \Psi_{(\Psi_{(0,1)}(T_{j}-T_{j-1},-a),1)}(T_{j-1}-t,-a),X(t)\rangle}\\
&+ \frac{1}{\delta}e^{-\int_{t}^{T_{j}}a_0(s)ds+\int_{T_{j-1}}^{T_j}(c_0(s)+a_0(s)+qb_0(s))ds
+\Phi_{(0,1)}(T_{j}-T_{j-1},-a)+\Phi_{(0,1)}(T_{j}-T_{j-1},c)-\Phi_{(0,1)}(T_{j}-T_{j-1},-(a+qb))}\\
&.e^{
\Phi_{(\Psi_{(0,1)}(T_{j}-T_{j-1},-a)+ \Psi_{(0,1)}(T_{j}-T_{j-1},c)-\Psi_{(0,1)}(T_{j}-T_{j-1},-(a+qb)),1)}(T_{j-1}-t,-a)}\\
&.e^{
\langle\Psi_{(\Psi_{(0,1)}(T_{j}-T_{j-1},-a)+ \Psi_{(0,1)}(T_{j}-T_{j-1},c)-\Psi_{(0,1)}(T_{j}-T_{j-1},-(a+qb)),1)}(T_{j-1}-t,-a),X(t)\rangle}.
\end{align*}


Therefore, also the expectations in \eqref{swapcond} can be computed explicitly and the formula can be obtained in closed form.

\subsection{The pricing of caps}

Let us first consider a caplet on the LIBOR with maturity $T_{j}$, whose payoff is given by
\begin{align*}
K\delta\left( L(T_{j-1}, T_j)-R\right)^+,&
\end{align*}
where $K$ denotes the notional and $R$ is the strike price of the caplet.

The price at time $t\leq T_{j-1}$ is given by
\begin{align}
Caplet_t=&\mathbb{E}_t^\mathbb{Q}\left[e^{-\int_{t}^{T_j}r_c(s)ds}K\delta\left(L(T_{j-1},T_j)-R\right)^+\right]\nonumber\\
=& D^{OIS}(t,T_j)\mathbb{E}_t^{\mathbb{Q}^{T_j}}\left[K\delta\left(L(T_{j-1},T_j)-R\right)^+\right]\nonumber\\
=& D^{OIS}(t,T_j)K\delta\mathbb{E}_t^{\mathbb{Q}^{T_j}}\left[
\frac{1}{\delta}\left(\frac{\mathbb{E}_{T_{j-1}}^\mathbb{Q}\left[e^{\int_{T_{j-1}}^{T_{j}}\phi(s)ds}\right]}{\mathbb{E}_{T_{j-1}}^\mathbb{Q}
\left[e^{-\int_{T_{j-1}}^{T_{j}}(r_c(s)+\lambda(s)q)ds}\right]}-{\left( 1+\delta R\right)}\right)^+\right]\nonumber\\
=& D^{OIS}(t,T_j)K\left( 1+\delta R\right)\mathbb{E}_t^{\mathbb{Q}^{T_j}}\left[\left(
e^{Z (T_{j-1})}-1\right)^+\right],\label{Cap}
\end{align}
where
\begin{align}
e^{Z (T_{j-1})}=& \frac{1}{1+\delta R}\frac{\mathbb{E}_{T_{j-1}}^\mathbb{Q}\left[e^{\int_{T_{j-1}}^{T_{j}}\phi(s)ds}\right]}{\mathbb{E}_{T_{j-1}}^\mathbb{Q}\left[e^{-\int_{T_{j-1}}^{T_{j}}(r_c(s)+\lambda(s)q)ds}\right]}\label{Z}.\\
\end{align}

Using \eqref{transformphi} and \eqref{transformrclambda}, we can now write

\begin{align}
e^{Z(T_{j-1})}=& \frac{1}{1+\delta R}
\frac{e^{\int_{T_{j-1}}^{T_{j}}c_0(s)ds+
\Phi_{(0,1)}(T_{j}-T_{j-1},c)+
\langle \Psi_{(0,1)}(T_{j}-T_{j-1},c), X(T_{j-1})\rangle}}
{e^{-\int_{T_{j-1}}^{T_{j}}(a_0(s)+b_0(s)q)ds+
\Phi_{(0,1)}(T_{j}-T_{j-1},-(a+bq))+
\langle \Psi_{(0,1)}(T_{j}-T_{j-1},-(a+bq)), X(T_{j-1})\rangle}}\nonumber\\
=& e^{f(T_{j-1})+\langle g(T_{j-1}), X({T_{j-1}})\rangle},
\end{align}
where
\begin{align}
f(T_{j-1}) =& -\ln (1+\delta R) +
\int_{T_{j-1}}^{T_{j}}(c_0(s)+a_0(s)+b_0(s)q)ds\nonumber\\
& +\Phi_{(0,1)}(T_{j}-T_{j-1},c)-\Phi_{(0,1)}(T_{j}-T_{j-1},-(a+bq)),\label{f1}\\
g(T_{j-1})=& \Psi_{(0,1)}(T_{j}-T_{j-1},c)-\Psi_{(0,1)}(T_{j}-T_{j-1},-(a+bq)).\label{g1}
\end{align}

In order to compute Equation \eqref{Cap}, we apply Fourier transform techniques well-explored by e.g. Carr and Madan  (\citeyear*{CarrMadan:1999}). For this purpose, we derive the expression for the  characteristic function  or the moment generating function of $Z$ under the $\mathbb{Q}^{T_j}-$forward probability measure.

\begin{proposition}\label{fcX1X2}
The conditional  characteristic function under the $\mathbb{Q}^{T_j}-$forward probability measure of the random variable   $Z$ defined in  \eqref{Z} is given by
\begin{align*}
\varphi^{T_j}_Z(u)&=\mathbb{E}_t^{\mathbb{Q}^{T_j}}\left[e^{iuZ (T_{j-1})}\right]\\
&=\frac{1}{D^{OIS}(t,T_j)}e^{iuf(T_{j-1})-\int_t^{T_{j}}a_0(s)ds+\Phi_{(0,1)}(T_{j}-T_{j-1},-a)}\\
&\quad .e^{\Phi_{(iug(T_{j-1})+\Psi_{(0,1)}(T_{j}-T_{j-1},-a),1)}(T_{j-1}-t,-a)+\langle \Psi_{(iug(T_{j-1})+\Psi_{(0,1)}(T_{j}-T_{j-1},-a),1)}(T_{j-1}-t,-a),X(t)\rangle
}.
\end{align*}
The transform is well defined for  $u\in\mathbb{C}$ such that $\Phi_{(iug(T_{j-1})+\Psi_{(0,1)}(T_{j}-T_{j-1},-a),1)}(T_{j-1}-t,-a)$ and $\Phi_{(iug(T_{j-1})+\Psi_{(0,1)}(T_{j}-T_{j-1},-a),1)}(T_{j-1}-t,-a)$ are bounded.
\end{proposition}

\begin{proof}
From \eqref{OISbondprice} it follows that
\begin{align*}
\varphi^{T_j}_Z(u)&=\mathbb{E}_t^{\mathbb{Q}^{T_j}}\left[e^{iuZ (T_{j-1})}\right]\\
&= e^{iuf(T_{j-1})}\mathbb{E}_t^{\mathbb{Q}^{T_j}}\left[e^{iu\langle g(T_{j-1}), X({T_{j-1}})\rangle}\right]\\
&= \frac{e^{iuf(T_{j-1})}}{D^{OIS}(t,T_j)}
\mathbb{E}_t^{\mathbb{Q}}\left[e^{-\int_t^{T_j}r_c(s)ds
+iu\langle g(T_{j-1}), X({T_{j-1}})\rangle}\right]\\
&= \frac{e^{iuf(T_{j-1})}}{D^{OIS}(t,T_j)}
\mathbb{E}_t^{\mathbb{Q}}\left[e^{-\int_t^{T_{j-1}}r_c(s)ds
+iu\langle g(T_{j-1}), X({T_{j-1}})\rangle}
\mathbb{E}_{T_{j-1}}^{\mathbb{Q}}\left[e^{-\int_{T_{j-1}}^{T_{j}}r_c(s)ds}\right]\right]\\
&= \frac{e^{iuf(T_{j-1})}}{D^{OIS}(t,T_j)}
\mathbb{E}_t^{\mathbb{Q}}\left[e^{-\int_t^{T_{j-1}}r_c(s)ds
+iu\langle g(T_{j-1}), X({T_{j-1}})\rangle}
\mathbb{E}_{T_{j-1}}^{\mathbb{Q}}\left[e^{-\int_{T_{j-1}}^{T_{j}}a_0(s)ds-\int_{T_{j-1}}^{T_{j}}\langle a,X(s)\rangle ds}\right]\right]\\
&= \frac{e^{iuf(T_{j-1})-\int_t^{T_{j}}a_0(s)ds+\Phi_{(0,1)}(T_{j}-T_{j-1},-a)}}{D^{OIS}(t,T_j)}\\
&\quad .
 \mathbb{E}_t^{\mathbb{Q}}\left[e^{-\int_t^{T_{j-1}}\langle a,X(s)\rangle ds+\langle iug(T_{j-1})+\Psi_{(0,1)}(T_{j}-T_{j-1},-a),X(T_{j-1})\rangle}
\right].
\end{align*}
Now we apply once again Lemma \ref{lem:lemmaRiccatiXY} and we get the result.
\end{proof}

With the explicit expression of the characteristic function of the process $Z$ we can apply the Carr and Madan (\citeyear*{CarrMadan:1999}) methodology in order to get the price of a caplet on LIBOR.
For more complex product like swaptions analytical formulas are no more available. However, one can easily apply efficient approximations in the spirit of Caldana et al. (\citeyear*{cfgg14}).

\section{Expressions for the Pricing of a CDS}\label{CDS}

First of all, let us consider the time $t$-value of a zero-recovery loan with notional $1$ to  a $j$ bank within the LIBOR panel  given by Formula \eqref{unsecure}:
\begin{eqnarray}
\begin{aligned}
B^j(t,T)&=D^{\tiny \mbox{OIS}}(t,T)\mathbb{E}^{\mathbb{Q}^T}\left[ e^{-\int_t^{T}(\hat{\lambda}_j(s)-q\Lambda(s)) ds}\mid \mathcal{H}_t \right]\\
&=D^{\tiny \mbox{OIS}}(t,T)
e^{-\int_t^T (\hat{b}_0^j(s)-q\beta_0(s))ds+ \Phi_{(0,1)}(T-t,\hat{b}^j-q\beta)+\langle \Psi_{(0,1)}(T-t,\hat{b}^j-q\beta),X(t)\rangle },
\end{aligned}
\end{eqnarray}
where the functions $\Phi_{(0,1)},\Psi_{(0,1)}$ solve the Riccati system \eqref{riccatiPhi}, \eqref{riccatiPsi}.

Now let us find an expression for
 $\mathbb{E}^{\mathbb{Q}^{T}}\left[e^{-\int_t^u \hat{\lambda}_j(s)ds}\hat{\lambda}_j(u)\right]$.

This can be done as follows:
\begin{align*}\label{EXP}
\mathbb{E}^{\mathbb{Q}^{T}}\left[e^{-\int_t^u \hat{\lambda}_j(s)ds}\hat{\lambda}_j(u)\right]=&-\frac{\partial}{\partial u}\mathbb{E}^{\mathbb{Q}^{T}}\left[e^{-\int_t^u \hat{\lambda}_j(s)ds}\right]\\
=&-\frac{\partial}{\partial u}\mathbb{E}^{\mathbb{Q}^{T}}\left[e^{-\int_t^u \hat{b}_0^j(s)ds-\int_t^u \langle\hat{b}^j, X(s)\rangle ds}\right]\\
=&-\frac{\partial}{\partial u}\left\{e^{-\int_t^u \hat{b}_0^j(s) ds+\Phi_{(0,1)}(u-t,\hat{b}^j)+\langle\Psi_{(0,1)}(u-t,\hat{b}^j),X(t)\rangle} \right\}\\
=&-\left(-\hat{b}_0^j(u)+ \partial_u\Phi_{(0,1)}(u-t,\hat{b}^j)+\langle\partial_u \Psi_{(0,1)}(u-t,\hat{b}^j),X(t)\rangle\right)\\
&.e^{-\int_t^u \hat{b}_0^j(s) ds+\Phi_{(0,1)}(u-t,\hat{b}^j)+\langle\Psi_{(0,1)}(u-t,\hat{b}^j),X(t)\rangle} ,
\end{align*}
where the derivatives of the functions $\Phi_{(0,1)},\Psi_{(0,1)}$ can be easily computed from the Riccati system \eqref{riccatiPhi}, \eqref{riccatiPsi}.

\bibliographystyle{agsm}
\bibliography{bib}
\end{document}